\documentclass[11pt]{article}
\usepackage[utf8]{inputenc}
\usepackage[backend=biber, 
style=alphabetic, useprefix=true,
minalphanames=4,
maxalphanames=5,
maxbibnames=99]{biblatex}

\usepackage{amsmath}
\usepackage{amsfonts}
\usepackage{amssymb}
\usepackage{amsthm}
\usepackage{stmaryrd}
\usepackage{graphicx}
\usepackage{mathtools}
\usepackage[shortlabels]{enumitem}
\usepackage{xspace}
\usepackage{proof}
\usepackage{xcolor}
\usepackage[vlined,ruled, linesnumbered]{algorithm2e}
\SetArgSty{textnormal}
\usepackage{color}
\usepackage[colorlinks=true, allcolors=blue]{hyperref}

\usepackage[letterpaper, top=1in, left=1in, right=1in, bottom=1in]{geometry}

\setlist[itemize]{noitemsep}
\setlist[enumerate]{noitemsep}

\DeclarePairedDelimiter\ceil{\lceil}{\rceil}
\DeclarePairedDelimiter\floor{\lfloor}{\rfloor}

\newtheorem*{claim*}{Claim}

\newtheorem{theorem}{Theorem}[section]
\newtheorem{lemma}[theorem]{Lemma}

\newtheorem{observation}[theorem]{Observation}

\newtheorem{cor}[theorem]{Corollary}

\theoremstyle{definition}

\theoremstyle{remark}
\newtheorem{remark}[theorem]{Remark}

\newcommand{\cD}{\mathcal{D}}

\newcommand{\cR}{\mathcal{R}}
\newcommand{\cL}{\mathcal{L}}
\newcommand{\cO}{\mathcal{O}}
\newcommand{\cP}{\mathcal{P}}

\newcommand{\bF}{\mathbb{F}}

\newcommand{\cS}{\mathcal{S}}

\newcommand{\eps}{\epsilon}
\newcommand{\tr}{\texttt{r}}
\newcommand{\tb}{\texttt{b}}

\newcommand{\rpdt}[1][]{\textsf{R}^{\oplus\textsf{-dt}}}

\newcommand{\rpdtsize}[1][]{\textsf{RSize}^{\oplus \textsf{-dt}}}

\newcommand{\rpsdt}[1][]{\textsf{R}^{\oplus, *\textsf{-dt}}}

\newcommand{\rodt}[1][]{\textsf{R}^{1\textsf{-dt}}}

\newcommand{\free}{\mathrm{Free}}
\newcommand{\fix}{\mathrm{Fix}}

\newcommand{\Bern}{\textup{Bernoulli}}

\newcommand{\poly}{\text{poly}}

\newcommand{\ind}{\text{IND}}
\newcommand{\bind}{\text{bIND}}
\newcommand{\maj}{\text{MAJ}}

\newcommand{\ip}{\text{IP}}

\newcommand{\lift}{\mathrm{Lift}}

\newcommand{\res}{\mathrm{Res}}
\newcommand{\php}{\mathrm{PHP}}
\newcommand{\bphp}{\mathrm{BPHP}}
\newcommand{\coll}{\mathrm{Coll}}
\newcommand{\ppf}{\mathrm{PI}}
\newcommand{\fphp}{\mathrm{FPHP}}

\DeclareMathOperator{\B}{\{0,1\}}
\DeclareMathOperator{\T}{\{0,1,*\}}
\DeclareMathOperator{\E}{\mathbb{E}}
\DeclareMathOperator{\supp}{supp}

\DeclareMathOperator{\rk}{rk}
\DeclareMathOperator{\cl}{Cl}

\DeclareMathOperator{\proj}{\text{proj}}
\DeclareMathOperator{\sdim}{sdim}

\bibliography{refs}

\title{Lower Bounds for Bit Pigeonhole Principles in \\ Bounded-Depth Resolution over Parities}
\author{
Farzan Byramji \thanks{University of California, San Diego. \href{mailto:fbyramji@ucsd.edu}{fbyramji@ucsd.edu}. Supported by NSF AF:Medium 2212136.} \and
Russell Impagliazzo \thanks{University of California, San Diego. \href{mailto:rimpagliazzo@ucsd.edu}{rimpagliazzo@ucsd.edu}. Supported by NSF AF:Medium 2212136.}} 
\date{November 24, 2025}

\begin{document}

\maketitle

\begin{abstract}
We prove lower bounds for proofs of the bit pigeonhole principle (BPHP) and its generalizations in bounded-depth resolution over parities ($\res(\oplus)$).
   For weak $\bphp_n^m$ with $m = cn$ pigeons (for any constant $c > 1$) and $n$ holes, for all $\eps > 0$, we prove that any depth $N^{1.5 - \eps}$ proof in $\res(\oplus)$ must have size $\exp(\widetilde{\Omega}(N^{2\eps}))$,
     where $N = cn\log n$ is the number of variables. 
    
    Inspired by recent work in TFNP on multicollision-finding, we consider a generalization of the bit pigeonhole principle, denoted $t$-$\bphp_n^m$, asserting that there is a map from $[m]$ to $[n]$ ($m > (t-1)n$) such that each $i \in [n]$ has fewer than $t$ preimages. We prove that any depth $N^{2-1/t-\eps}$ proof in $\res(\oplus)$ of $t$-$\bphp_n^{ctn}$ (for any constant $c \geq 1$) must have  size $\exp(\widetilde{\Omega}(N^{\eps t/(t-1)}))$.
    
    For the usual bit pigeonhole principle, we show that any depth $N^{2-\eps}$ $\res(\oplus)$ proof of $\bphp_n^{n+1}$ must have  size  $\exp(\widetilde{\Omega}(N^{\eps}))$. As a byproduct of our proof, we obtain that any randomized parity decision tree for the collision-finding problem with $n+1$ pigeons and $n$ holes must have depth $\Omega(n)$, which matches the upper bound coming from a deterministic decision tree.

We also prove a lifting theorem for bounded-depth $\res(\oplus)$ with a constant size  gadget which lifts from $(p, q)$-DT-hardness, recently defined by Bhattacharya and Chattopadhyay (2025). By combining our lifting theorem with the $(\Omega(n), \Omega(n))$-DT-hardness of the $n$-variate Tseitin contradiction over a suitable expander, proved by Bhattacharya and Chattopadhyay, we obtain an $N$-variate constant-width unsatisfiable CNF formula with $O(N)$ clauses for which any depth $N^{2-\eps}$ $\res(\oplus)$ proof requires size $\exp(\Omega(N^\eps))$. Previously no superpolynomial lower bounds were known for $\res(\oplus)$ proofs when the depth is superlinear in the size of the formula  (instead of the number of variables).

\end{abstract}
\vspace{\fill}
\pagebreak

\tableofcontents

\pagebreak
\pagenumbering{arabic}
\section{Introduction}

The resolution proof system is perhaps the most well-studied proof system in propositional proof complexity. Lines in resolution are clauses (disjunctions of literals). The resolution rule allows one to derive from clauses $A \vee x$ and $B \vee \neg x$, the clause $A \vee B$ (where $A$ and $B$ are arbitrary clauses). A resolution refutation of a CNF formula $\varphi$ is a sequence of clauses deriving the empty clause (which is clearly unsatisfiable) from the clauses of $\varphi$ thereby proving that $\varphi$ cannot have a satisfying assignment. The first superpolynomial lower bound for resolution was proved by Haken \cite{haken85} for the unary CNF encoding of the pigeonhole principle. Many more lower bounds for resolution proofs of several other natural formulas have been proved in the following decades \cite{urquhart1987hard, chvatal1988many, ben1999short, beame2002efficiency, razborov2004resolution, beame2005resolution, beame2007resolution}. 

Going beyond reasoning with clauses, it is natural to consider proof systems where the lines are more powerful circuits like $AC^0$-circuits or $AC^0[p]$-circuits. The first superpolynomial lower bound against $AC^0$-Frege was shown by Ajtai  \cite{ajtai1994complexity} for the pigeonhole principle which was later strengthened to an exponential lower bound \cite{pitassi1993exponential, krajicek1995exponential}. On the other hand, proving any superpolynomial lower bound for $AC^0[p]$-Frege is a longstanding challenge. 

A natural subsystem of $AC^0[2]$-Frege is the proof system resolution over parities, $\res(\oplus)$, introduced by Itsykson and Sokolov \cite{itsykson2014lower, itsykson2020resolution} a decade ago.
$\res(\oplus)$ extends the power of resolution to allow linear algebra over $\bF_2$. Lines in $\res(\oplus)$ are linear clauses, which are disjunctions of linear equations over $\bF_2$. In $\res(\oplus)$, a resolution step can derive from linear clauses $A$ and $B$, any linear clause $C$ such that every $x \in \B^n$ satisfying both $A$ and $B$ also satisfies $C$. In particular, for any linear form $v$, we can derive $A \vee B$ from $A \vee (v = 0)$ and $B \vee (v = 1)$.

Proving lower bounds for this seemingly simple strengthening of resolution has turned out to also be quite challenging in general, though there has been progress on restricted subsystems of $\res(\oplus)$. Several works like \cite{itsykson2020resolution, garlik2018some, itsykson2021proof, gryaznov2024resolution, cmss23, beame2023, bi25} have shown lower bounds for tree-like $\res(\oplus)$ via various techniques. In the last couple of years, there has been exciting progress on restricted DAG-like subsystems of $\res(\oplus)$  starting with the work of Efremenko, Garlik and Itsykson \cite{egi24} who showed an exponential lower bound for bottom-regular $\res(\oplus)$ proofs of the bit pigeonhole principle. Bhattacharya, Chattopadhyay and Dvo\v{r}\'{a}k \cite{bhat2024} and later Alekseev and Itsykson \cite{alek2024} gave such exponential lower bounds for formulas that have polynomial size proofs in ordinary resolution. 

Going beyond bottom-regular $\res(\oplus)$, \cite{alek2024} gave an exponential lower bound for $\res(\oplus)$ proofs whose depth is at most $O(N \log \log N)$ where $N$ is the number of variables of the formula. The depth limit was further pushed to $O(N \log N)$ by Efremenko and Itsykson \cite{ei25}. These lower bounds were proved for the Tseitin formula on a logarithmic degree expander lifted with 2-stifling gadgets. Bhattacharya and Chattopadhyay \cite{bc25} have recently proved an exponential lower bound for any depth $N^{2 - \eps}$ $\res(\oplus)$ proof of the Tseitin formula on a constant degree expander lifted with the inner product gadget of sufficiently large logarithmic size. More generally, they show that any gadget with sufficiently small correlation will parities suffices for their lower bound.

Also recently, Itsykson and Knop \cite{ik25} gave a formula for which any depth $O(N \log N)$ $\res(\oplus)$ proof has exponential size, and the formula also has a resolution refutation of quasipolynomial size, thereby establishing a supercritical tradeoff for resolution over parities. They also showed how one can obtain a size-depth lower bound in $\res(\oplus)$ by lifting resolution width. Their lifting theorem gives several other examples of formulas with $\poly(N)$ clauses for which any $\res(\oplus)$ refutation must have depth $\Omega(N \log N)$ or size $2^{\Omega(N/\log N)}$. This lifting theorem also gives superpolynomial lower bounds for larger depths, up to depth $O(N^{1.5-\eps})$. However, one shortcoming of these formulas is that they have exponentially many clauses and so the size lower bound is only quasipolynomial in terms of the number of clauses. 

\paragraph{Our contributions.}

In this work, we prove size-depth lower bounds in $\res(\oplus)$ for the bit pigeonhole principle (BPHP) and some related formulas. The bit pigeonhole principle (or binary pigeonhole principle) CNF formula asserts the contradictory claim that there are $m$ bitstrings of length $\log n$ (where $m > n$) such that no two of them are equal. We have $m \log n$ variables $x_{i, j}$ for $i \in [m], j \in [\log n]$. For every pair of distinct indices $i_1$ and  $i_2$ in $[m]$ and every $z \in \{0,1\}^{\log n}$, there is a clause stating that $x_{i_1} \neq z$ or $x_{i_2} \neq z$.

We prove an exponential lower bound on the size of $\res(\oplus)$ proofs of the weak bit pigeonhole principle with $n$ holes when the depth is at most $O(n^{1.5 - \eps})$ for any constant $\eps > 0$.

\begin{theorem} \label{thm:wbphp}
    Suppose there exists a $\res(\oplus)$ proof of $\bphp_n^{m}$ $(m > n)$ whose size is $S$ and depth is $D$. Then $D \sqrt{\log S} = \Omega(n^{1.5})$.
\end{theorem}

This strengthens the bottom regular $\res(\oplus)$  lower bound for $\bphp_n^{n+1}$ \cite{egi24} in multiple ways.
When $m = cn$ for any constant $c > 1$, the number of variables $N$ is $cn \log n$. So Theorem \ref{thm:wbphp} implies that for any $\eps > 0$, any depth $O(N^{1.5 - \eps})$ proof of $\bphp_n^{cn}$ must have size $\exp(\tilde{\Omega}(N^{2\eps}))$.

Another natural setting of parameters is when the depth is bounded by the number of variables $m \log n$, which is true for any natural notion of regularity. Theorem~\ref{thm:wbphp} implies that any such $\res(\oplus)$ proof of $\bphp_n^{m}$ must have size at least $\exp\left(\Omega\left(\frac{n^3}{m^2 \log^2 n}\right)\right)$. In particular, we get an exponential lower bound if $m = n^{1.5 - \Omega(1)}$. 

It is known that every $\res(\oplus)$ proof of $\bphp_n^m$ must have depth $\Omega(n)$ \cite{egi24}, and there exists a tree-like resolution proof of depth $O(n)$ (and, hence, size $2^{O(n)}$) \cite{dantchev2024proof}. This implies that the bound $\Omega(n^{1.5})$ in Theorem~\ref{thm:wbphp} cannot be improved in general. But there can still be other incomparable size-depth lower bounds for $\bphp_n^m$. 

There are two natural ways of making weak BPHP potentially harder. The more common way is to keep $m$ small with respect to $n$. Another way would be to consider the following generalized pigeonhole principle. For a positive integer $t$, if $m$ pigeons are assigned to $n$ holes and $m> (t-1)n$, then there must be some $t$ pigeons which are assigned the same hole. We will consider both these ways, starting with the latter since the proof of that lower bound is a relatively simple extension of Theorem~\ref{thm:wbphp}.

We use $t$-$\bphp_n^m$ for $m > (t-1)n$ to denote the following CNF formula, which is a generalization of $\bphp_n^m$. We have $m \log n$ variables $x_{i, j}$ for $i \in [m], j \in [\log n]$. For every size $t$ subset $\{i_1, i_2, \dots, i_t\} \subseteq [m]$ and every $z \in \{0,1\}^{\log n}$, there is a clause stating that for some $k \in [t]$, $x_{i_k} \neq z$. So $t$-$\bphp_n^m$ consists of $\binom{m}{t}n$ clauses. The corresponding search problem of finding $t$-collisions was recently studied from a TFNP perspective in the independent works of Jain, Li, Robere, Xun \cite{jlrx24}, and Bennett, Ghentiyala, Stephens-Davidowitz \cite{bgs25}, which inspired us to consider the proof complexity of the corresponding CNF formulas. 

We generalize Theorem~\ref{thm:wbphp} to obtain exponential lower bounds for depth $O(n^{2-1/t-\eps})$ $\res(\oplus)$ proofs of $t$-$\bphp_n^m$.

\begin{theorem}\label{thm:tbphp} 
    For all $t \coloneqq t(n) \geq 2$, all $m > (t-1)n$, if there exists a $\res(\oplus)$ proof of $t$-$\bphp_n^m$ whose size is $S$ and depth is $D$, then $D (\log S)^{1-1/t} = \Omega(tn^{2-1/t})$.
    
        In particular, for $t = \log n$, all $m > (t-1)n$, if there exists a $\res(\oplus)$ proof of $t$-$\bphp_n^m$ whose size is $S$ and depth is $D$, then $D \log S = \Omega(n^{2}\log n)$.
\end{theorem}  

For any positive real $\eps$, by taking $t$ to be large enough in terms of $\eps$, Theorem~\ref{thm:tbphp} gives a family of polynomial size CNF formulas which is hard for depth $N^{2-\eps}$ $\res(\oplus)$. For instance, taking $t = \ceil{2/\eps}$ and $m = \Theta(tn) = \Theta(n)$ we have that $t$-$\bphp_n^m$ has $N = \Theta(n \log n)$ variables,  $N^{O(1)}$ clauses and requires depth $N^{2-\eps}$ $\res(\oplus)$ refutations of size $\exp(\widetilde{\Omega}(N^{\eps/2+\eps^2/4}))$. 

By considering $t = \log n$ and $m = \Theta(tn) = \Theta(n \log n)$, we get a family of formulas with $N = \Theta(n \log^2 n)$ variables for which any depth $D$ $\res(\oplus)$ proof must have size $\exp\left(\Omega\left(\frac{N^2}{D \log^3 N}\right)\right)$. Up to a $\log N$ factor, this matches the size-depth lower bound obtained in \cite{bc25} (which was for a different formula) in terms of the number of variables. However, one undesirable aspect is that for $t = \log n$, $t$-$\bphp_n^{tn}$ has quasipolynomially many clauses ($\binom{tn}{t}n = n^{t + \Theta(1)}$). The bit pigeonhole principle which we discuss next does not have this limitation.

The bit pigeonhole principle with $n+1$ pigeons and $n$ holes has $N = (n+1)\log n$ variables and $\binom{n+1}{2}n = \Theta(n^3)$ clauses. We prove that any $\res(\oplus)$ proof of $\bphp_n^{n+1}$ with depth $D$ must have size $\exp\left(\Omega\left(\frac{n^2}{D}\right)\right) = \exp\left(\Omega\left(\frac{N^2}{D \log^2 N}\right)\right)$. More generally, we prove the following lower bound when the number of pigeons is $n + o(n)$.

\begin{theorem}\label{thm:s_bphp} There exists a constant $\delta > 0$ such that the following holds. Let $k \leq \delta n$ and $m \coloneqq n+k$. If there exists a $\res(\oplus)$ proof of $\bphp_n^m$ whose size is $S$ and depth is $D$, then 
\[
\log S = \Omega \left(\min\left( \frac{n^2}{D}, \frac{n^4}{kD^2} \right)\right).
\]
\end{theorem}
For $k = \delta n$, the above bound becomes $\Omega(\frac{n^3}{D^2})$ since we always have $D = \Omega(n)$ \cite{egi24}. This coincides with the bound in Theorem~\ref{thm:wbphp}.

In the course of proving Theorem~\ref{thm:s_bphp}, we also obtain a nearly tight lower bound on the depth of zero-error randomized parity decision trees solving collision-finding $\coll_n^m$ for all $m$ and all abort probabilities which are at most a constant. In the collision-finding search problem $\coll_n^m$, we are given an assignment $x \in (\{0,1\}^{\log n})^m$ for $\bphp_n^m$ and the goal is to find two pigeons (distinct $i_1, i_2 \in [m]$) sent to the same hole, i.e.~$x_{i_1} = x_{i_2}$ (this is a slight relaxation of the false clause search problem for $\bphp_n^m$).  Here the randomized parity decision tree (RPDT) must be correct whenever it outputs a pair ${i_1, i_2} \subseteq [m]$, but it is allowed to abort (output a special symbol $\bot$) with some probability $\eps$. 

\begin{theorem} \label{thm:rpdt_coll}
    There exists a constant $\delta > 0$ such that the following holds. Let $k \leq \delta n$ and $m \coloneqq n+k$. Let $\eps \in [2^{-k}, 1/2]$. Any randomized parity decision tree solving $\coll_n^m$ with probability $1-\eps$ must have depth $\Omega(n\sqrt{\log(1/\eps)/k})$.
\end{theorem}

The above bound is tight up to an $O(\log n)$ factor even for randomized ordinary decision trees by a birthday-paradox-like calculation (where the additional $\log n$ factor comes from the size of each block). This shows that parity queries do not help significantly for the problem of collision-finding. In the case where $k$ and $\eps$ are constants (or more generally, when $\log(1/\eps) = \Omega(k)$), the additional $O(\log n)$ factor is not required in the upper bound since there is even a deterministic decision tree of depth $O(n)$ which solves collision-finding \cite{dantchev2024proof}.  For any $m$, \cite{egi24, bi25} had earlier implicitly shown a lower bound of $\Omega(\sqrt{n \log(1/\eps)})$ for RPDTs solving $\coll_n^m$ and that bound is nearly tight for $m = n+\Omega(n)$. 
The lower bound in Theorem~\ref{thm:rpdt_coll} also holds for randomized parity decision trees that are allowed to err with probability $\eps$ since any RPDT making errors can be turned into a zero-error RPDT by making $\log n$ additional parity queries to verify the output. 

\begin{remark}
    In contrast to parity decision trees, deterministic two-party communication protocols can solve  collision-finding under the natural bipartition for $m = n+1$ with $\widetilde{O}(\sqrt{n})$ bits of communication and this is nearly tight even for randomized protocols \cite{itsykson2021proof}. When $m = n+\Omega(n)$, the randomized communication complexity further reduces to $\widetilde{O}(n^{1/4})$ \cite{gj22} which is also known to be almost tight \cite{yz24}. 
    
    In the multiparty number-in-hand setting, stronger lower bounds were shown by Beame and Whitmeyer \cite{beame2025multiparty}  though those bounds also do not imply the RPDT lower bounds in Theorem~\ref{thm:rpdt_coll} for constant error.
    \cite[Theorem 1.1]{beame2025multiparty} implies an RPDT lower bound of $\widetilde{\Omega}(n/2^{2\sqrt{\log n}})$ for $k \leq 2^{2\sqrt{\log n}-2}$ and $\widetilde{\Omega}(n/k)$ for larger $k$ (where $k = m - n$). For $k$ larger than $\sqrt{n}$, the better bound $\widetilde{\Omega}(\sqrt{n})$ is implied by \cite[Theorem 1.3]{beame2025multiparty}. Beame and Whitmeyer conjecture that their bounds for small $k$ can be improved. In particular, it seems quite plausible that for $m = n+1$, the randomized number-in-hand complexity with $\log n$ players should be $\widetilde{\Omega}(n)$. \hfill $\blacktriangle$
\end{remark}

Finally, we prove a lifting theorem for bounded-depth $\res(\oplus)$ using a constant size gadget which, informally, lifts lower bounds for bounded-depth resolution that can be proved using the random walk with restarts approach of Alekseev and Itsykson \cite{alek2024}. Bhattacharya and Chattopadhyay \cite{bc25} formalized the notion of $(p, q)$-DT-hardness which captures some of the requirements for such a bounded-depth resolution lower bound. To state the definition of $(p, q)$-DT-hardness, the following notation will be convenient. For any decision tree $T$ on $\B^n$ and distribution $\mu$ over $\B^n$, we use $\sigma_{T, \mu}$ to denote the partial assignment corresponding to the random leaf reached when $T$ is run on the distribution $\mu$.

For a CNF formula $\varphi$ on $n$ variables, a set $\cP \subseteq \T^n$ is called $(p, q)$-DT-hard for $\varphi$ if it has the following properties:
\begin{itemize}
    \item No partial assignment in $\cP$ falsifies a clause of $\varphi$.
    \item $\cP$ contains the empty partial assignment $*^n$.
    \item $\cP$ is downward closed: if $\rho'$ can be obtained from $\rho$ by forgetting the value of some fixed variable and $\rho \in \cP$, then $\rho' \in \cP$.
    \item For every partial assignment $\rho \in \cP$ which fixes at most $p$ variables, there is some distribution $\nu_\rho$ over assignments to the variables not fixed by $\rho$ such that the following holds. If $T$ is any decision tree of depth $q$ which only queries variables not fixed by $\rho$, then $\Pr[\rho \; \cup \; \sigma_{T, \nu_\rho} \in \cP ] \geq 1/3$.
\end{itemize}

The choice of the constant $1/3$ in the above definition is not important. In fact, it can be meaningful to consider much smaller probabilities, though we only consider the above definition in this section for simplicity.

For formulas where the variables can be naturally partitioned into blocks, we can consider a variant of the above definition, where we only deal with complete blocks (or, equivalently, treat the input as consisting of symbols from a large alphabet). More precisely, for a formula $\varphi$ over $n$ blocks of $l$ variables each, a $(p, q)$-block-DT-hard set $\cP$ consists of only block-respecting partial assignments, where each block of the partial assignment is either completely fixed or completely unfixed. In the decision tree requirement above, we now consider decision trees that query an entire block in a single query. All the other requirements stay the same.

For such formulas $\varphi$ with a block structure, instead of lifting via standard composition which replaces each variable by a fresh block of variables for some gadget $g$, we can consider block-lifting, where we replace each block by a copy of a multi-output gadget $g' : \B^{l'} \rightarrow \B^l$. This can be  useful for formulas that have large width but small block-width, where the block-width of a formula $\varphi$ is the maximum number of distinct blocks involved in any clause of $\varphi$. The gadget we consider here is block-indexing $\bind_l : \B \times \B^l \times \B^l \rightarrow \B^l$ from \cite{de2021automating}, which is defined by $\bind_l(s, x^0, x^1) = x^s$. We will use $\lift(\varphi)$ to denote the block-lift of $\varphi$ with block-indexing of appropriate size (see Section \ref{sec:lifting} for details about the clauses of $\lift(\varphi)$).

We can now state our lifting theorem.
\begin{theorem} \label{thm:lift}
    Let $\varphi$ be a formula over $nl$ variables, partitioned into $n$ blocks of $l$ variables each. Let $m$ be the number of clauses of $\varphi$, $w$ its width and $b$ its block-width. Suppose there exists a $(p, q)$-block-DT-hard set for $\varphi$. 
    Then the formula $\lift(\varphi)$ over $n(2l+1)$ variables contains $2^bm$ clauses of width $w + b$ and every depth $D$ $\res(\oplus)$ proof of $\lift(\varphi)$ must have size $\exp(\Omega(pq/D))$.
\end{theorem}
If we did not care about block-lifting, we could use standard composition with any constant size stifling gadget \cite{cmss23}  (like indexing $\ind_{1+2}$, inner product $\ip_{2+2}$, majority $\maj_3$) in the above statement to get the lower bound $\exp(\Omega(pq/D))$ on the size of any depth $D$ proof. Theorem \ref{thm:lift} can be strengthened (Theorem \ref{thm:lift_gen}) to also give size-depth lower bounds when we only have small probabilities in the definition of DT-hardness. This generalization recovers (and slightly improves) the size-depth lifting theorem from \cite{alek2024, ei25} based on the advanced Prover-Delayer game defined in \cite{alek2024}.

The way we lift to bounded-depth $\res(\oplus)$ is quite different from how Bhattacharya and Chattopadhyay \cite{bc25} do it (though at a high level both proofs follow the strategy from \cite{alek2024}). In particular, we do not use their notion of $(p, q)$-PDT-hardness, and, we do not know if a constant size gadget suffices to lift DT-hardness to PDT-hardness as they define it. 

Instead, the overall proof strategy for Theorem \ref{thm:lift} is closer to how we prove a lower bound for weak $\bphp$, Theorem \ref{thm:wbphp}, with the core components now being (a slight modification of) a lifting theorem for randomized parity decision trees \cite{bi25} and some ideas from prior work on $\res(\oplus)$ \cite{bhat2024, alek2024} about stifling gadgets. 

Let us consider some formulas which are DT-hard. Bhattacharya and Chattopadhyay \cite{bc25} proved that the Tseitin formula on $n$ variables over a suitably expanding constant degree graph has an $(\Omega(n), \Omega(n))$-DT-hard set. Since this is a constant width formula, we do not need to use the block-lift above and any stifling gadget suffices.
\begin{cor}
    Let $\varphi$ be the Tseitin contradiction with $n$ variables on a $(2n/d, d, \lambda)$ expander where $2n/d$ is odd, and $d$ and $\lambda$ are constants with $\lambda$ sufficiently small. Let $g$ be any constant size stifling gadget. Then the formula $\varphi \circ g$ has $O(n)$ variables, constant width, $O(n)$ clauses, and any depth $D$ $\res(\oplus)$ proof of $\varphi \circ g$ must have size $\exp(\Omega(n^2/D))$.
\end{cor}
Note that the above gives superpolynomial lower bounds on size when the depth is up to nearly quadratic even in the formula size, instead of just the number of variables. Earlier lower bounds \cite{alek2024, ei25, bc25, ik25} do not apply to depth superlinear in the number of clauses. 

Instead of a constant size stifling gadget, Bhattacharya and Chattopadhyay \cite{bc25} used a gadget which has sufficiently small correlation with all parities (a concrete example being inner product on $500 \log n$ bits where $n$ is the number of variables of the base formula). Their lower bound was also $\exp(\Omega(n^2/D))$, but because of the gadget size, the number of variables in the resulting formula was $N = 500n \log n$ and the number of clauses was a large polynomial. Very recently, in a later version \cite{bc25b}, they have improved the quantitative requirement for the gadget in their lifting theorem. Specifically, the inner product gadget with $(4 + \eta)\log n$ bits (for any constant $\eta > 0$) satisfies the condition required by their lifting theorem, which they use to improve the number of clauses in their resulting lifted formula to $O(n^{13+\eps})$ for any constant $\eps > 0$. 

Another natural formula which is $(\Omega(n), \Omega(n))$-block-DT-hard is, of course, the bit pigeonhole principle $\bphp_n^{n+1}$ where we have a block of $\log n$ bits for each of the $n+1$ pigeons. The block-width of $\bphp_n^{n+1}$ is $2$. So $\lift(\bphp_n^{n+1})$ has the same number of clauses as $\bphp_n^{n+1}$ up to a constant factor and any depth $D$ $\res(\oplus)$ proof of $\lift(\bphp_n^{n+1})$ must have size $\exp(\Omega(n^2/D))$. Even though this is already implied by Theorem \ref{thm:s_bphp}, we point this out because proving Theorem \ref{thm:lift} and then giving an $(\Omega(n), \Omega(n))$-block-DT-hard set for $\bphp_n^{n+1}$ is much simpler than the proof of Theorem \ref{thm:s_bphp}. 

Let us now consider an example of a formula which has large width but small block-width, so that the use of the block-lift above becomes essential to get a polynomial size formula. The functional pigeonhole principle $\fphp_n^{n+1}$ is the unary encoding of $\bphp_n^{n+1}$. We have $n+1$ blocks of $n$ variables, each block corresponding to a pigeon and indicating which holes it flies to. The clauses encode that each pigeon flies to exactly one hole and each hole contains at most one pigeon. The block-width of $\fphp_n^{n+1}$ is $2$. 

We obtain an $(\Omega(n), \Omega(n))$-block-DT-hard set for $\fphp_n^{n+1}$ by simply considering the unary version of the $(\Omega(n), \Omega(n))$-block-DT-hard set for $\bphp_n^{n+1}$. This implies that any depth $D$ $\res(\oplus)$ proof of $\lift(\fphp_n^{n+1})$ must have size $\exp(\Omega(n^2/D))$.

At first glance, this may seem trivial. One might think that since $\fphp_n^{n+1}$ requires resolution proofs of depth $\Omega(n^2)$ \cite[Lemma 2]{dantchev2001tree}, the lift of $\fphp_n^{n+1}$ should require $\res(\oplus)$ proofs of depth $\Omega(n^2)$. However, this is not true, since we are only applying a block-lift. In fact, there is a $\res(\oplus)$ proof for $\lift(\fphp_n^{n+1})$ of depth $O(n \log n)$. So we get exponential lower bounds on size as long as the allowed depth is up to almost quadratic in the minimum depth needed to refute $\lift(\fphp_n^{n+1})$. Though the allowed depth is not even linear in the number of variables $(n+1)(2n+1)$, we feel this is a meaningful comparison given that the lifted Tseitin formulas considered in prior work \cite{alek2024, ei25, bc25} require $\res(\oplus)$ proofs of depth linear in the number of variables. 

\paragraph{Pigeonhole principles in proof complexity.} The pigeonhole principle tautology and its variations have played a central role in proof complexity (see \cite{razborov2001proof} for a survey), especially in connections between propositional proofs and bounded arithmetic. As noted earlier, the first superpolynomial lower bounds for resolution \cite{haken85} and $AC^0$-Frege \cite{ajtai1994complexity, pitassi1993exponential, krajicek1995exponential} were for the pigeonhole principle. 

On the other hand, Buss gave a polynomial upper bound for PHP in Frege \cite{buss1987polynomial}, and Paris, Wilkie and Woods \cite{paris1988provability} gave a quasipolynomial upper bound for the weak pigeonhole principle in constant depth Frege, which they used to show that the existence of arbitrarily large primes was provable in $I\Delta_0$. (See also \cite{maciel2002new}.) The dual weak
pigeonhole principle was used by Je{\v{r}}{\'a}bek \cite{jerabek2004dual} to define a bounded arithmetic theory that captures probabilistic polynomial time reasoning.  

Because of the variety of surprising proofs of these tautologies that have been found, which often make fine distinctions between variations, and because of their use in formalizing counting arguments, precisely characterizing the proofs of the pigeonhole principle and its variants is an ongoing theme of work in proof complexity.

The complexity of refuting the bit pigeonhole principle ($\bphp_n^m$) has recently been studied in various proof systems including resolution and DNF-resolution \cite{atserias2015lower, dantchev2024proof}, cutting planes and its strengthenings \cite{hrubes2017random, itsykson2021proof, beame2025multiparty, de2025lifting}, Sherali Adams and Sum of Squares \cite{dantchev2024proof}. Out of these, since resolution is the most relevant proof system for our discussion, let us point out that Dantchev, Galesi, Ghani and Martin \cite{dantchev2024proof} showed that any resolution proof of $\bphp_n^m$ for any $m > n$ must have size $2^{\Omega(n/\log n)}$, which almost matches the tree-like resolution upper bound $2^{O(n)}$ for $\bphp_n^{n+1}$.

\paragraph{Organization.} Section \ref{sec:overview} contains a detailed overview of our proofs and a comparison of our techniques with prior work. In Section \ref{sec:prel}, we collect definitions and facts about them, mostly from prior work, that we use. In Section \ref{sec:tbphp}, we prove the lower bound for $t$-$\bphp$. In Section \ref{sec:sbphp}, we prove the lower bound for $\bphp_n^{n+k}$ where $k = o(n)$ and the randomized parity decision tree lower bound for $\coll_n^{n+k}$. In Section \ref{sec:lifting}, we prove the lifting theorem for bounded-depth $\res(\oplus)$. We end with some concluding remarks in Section \ref{sec:conc}. Appendix \ref{app:dt_hard} contains a proof of the DT-hardness of $\bphp$ and $\fphp$ and related discussion. Appendix \ref{app:pd_games} shows how our lifting theorem implies the lifting theorem from \cite{alek2024, ei25} based on the advanced Prover-Delayer game and discusses its limitations.
\pagebreak
\section{Proof overview}
\label{sec:overview}

In this section, we give an overview of the proofs of our lower bounds. For simplicity, in this section, we will mostly focus on showing bounds that are somewhat weaker than those stated in the introduction, but these weaker bounds on size will still be exponential for the same range of depths allowed by the stronger bounds. 

The proofs of all our lower bounds at a high level rely on the random walk with restarts approach of Alekseev and Itsykson \cite{alek2024} but with some key differences. To build intuition, in all cases, we will start by sketching how their approach can be used to prove lower bounds for ordinary resolution (without parities) and then show how those ideas can be ported to $\res(\oplus)$.

Instead of working directly with resolution and $\res(\oplus)$ proofs, it will be more convenient to work with the corresponding standard query models, subcube DAGs and affine DAGs respectively (see Section \ref{sec:aff_dags} for definitions). 
It is well known that any resolution refutation of a formula $\varphi$ can be transformed into a subcube DAG solving the corresponding false clause search problem for $\varphi$ without increasing the size or depth \cite{drgr22}. Similarly, a $\res(\oplus)$ refutation of $\varphi$ can be transformed into an affine DAG solving the corresponding false clause search problem \cite{egi24}.

\subsection{Weak $t$-$\bphp$}
We will first illustrate the random walk with restarts approach \cite{alek2024} to show that any size $S$ resolution proof of $\bphp_n^m$ (for any $m > n$) must have depth $\Omega(n^{1.5}/\log S)$. It suffices to show that any size $S$ subcube DAG for collision-finding $\coll_n^m$ must have depth $D = \Omega(n^{1.5}/\log S)$.

We wish to find a sequence of nodes $v_0, v_1, \dots, v_k$ in the DAG $C$ solving $\coll_n^m$, where $v_0$ is the root, such that for each $i \in [k]$, there is a path of length $d \coloneqq \sqrt{n}$  from $v_{i-1}$ to $v_i$, and $k = \Omega(n/\log S)$. To ensure that $k$ is large, we will want the width  (the number of fixed coordinates) of the corresponding partial assignment to only increase by $O(\log S)$ when going from node $v_{i-1}$ to $v_{i}$. We will maintain the property that for each node $v_i$ (for which the partial assignment $\rho_{v_i}$ has width $O(i \log S)$), there is a partial  assignment $\sigma_i$ fixing $O(i \log S)$ blocks of the input to distinct holes which satisfies the partial assignment $\rho_{v_i}$. 

Suppose we have found a node $v_i$ which is at depth $i\sqrt{n}$ from the root. We find the next node $v_{i+1}$ in the following way. Let $U_i \subseteq \B^l$ (where $l = \log n$) denote the collection of holes which have been assigned to some pigeon by $\sigma_i$ and let $\fix_i \subseteq [m]$ denote the blocks that are fixed by $\sigma_i$. Suppose $|U_i| \leq n/3$. Consider a random assignment $x \in (\B^l)^m$ which is picked uniformly from all assignments consistent with $\sigma_i$ and for which all of the pigeons in $[m]\setminus \fix_i$ are sent to holes in $\B^l \setminus U_i$. Note that the only collisions possible in such an assignment are between two pigeons from $[m] \setminus \fix_i$. 

Consider the depth $d$ decision tree obtained from the DAG rooted at $v_i$ and only considering nodes up to distance $d$, repeating nodes as necessary. By the standard birthday problem calculation, with probability at least $1 - d^2/2(n - |U_i|) \geq 1/4$, this decision tree has not found a collision under the above distribution. Since the DAG only has $S$ nodes, there is some node $w$ for which the corresponding leaves in the decision tree are reached with probability at least $1/(4S)$ without having found a collision. If $\rho_w$ mentions $r$ pigeons outside $\fix_i$, then the probability that a random assignment picked above satisfies $\rho_w$ is at most $(n/2(n-|U_i|))^r \leq (3/4)^r$. Together these imply that $r \leq O(\log S)$. So we can take $v_{i+1} \coloneqq w$ and there is some extension $\sigma_{i+1}$ of $\sigma_i$ fixing only $O(\log S)$ additional blocks to distinct holes in $\B^l \setminus U_i$. This argument can be repeated as long as $|U_{i+1}| \leq n/3$ implying that we can take $k = \Omega(n/\log S)$ as desired.

While not essential, we find it useful to rephrase the above argument in terms of restrictions that simplify the DAG. Having found a node $v_1$ at depth $d$ and some partial assignment $\sigma_1$ implying $\rho_{v_1}$, we can simplify $C$ under the restriction $\sigma_1$ to obtain another DAG $C_1$ where we can consider only the nodes reachable from $v_1$ since $\rho_{v_1}$ simplifies to the empty partial assignment after applying $\sigma_1$. While this new DAG $C_1$ does not quite solve $\coll_n^{m'}$ (where $m' = m - O(\log S)$ is the number of blocks not fixed by $\sigma_1$), it solves $\coll_n^{m'}$ when we only consider inputs which do not mention any of the used holes $U_1$. So in the next iteration, instead of using a distributional decision tree lower bound for general $\coll_n^{m'}$, we consider such a lower bound for a promise version of $\coll_n^{m'}$ where we only mention holes from $\B^l \setminus U_1$. This argument can be repeated as long as the number of forbidden holes is not too large. So if $D$ was the original depth of the DAG, the final DAG that we end up with after $\Omega(n/\log S)$ iterations has depth at most $D -\Omega(nd/\log S)$ which gives the desired lower bound.

\paragraph{Restrictions for affine DAGs.}

In trying to directly adapt the above argument involving restrictions, one difficulty is that standard restrictions which only fix few coordinates to bits do not suffice to fix large parities. But we can instead use affine restrictions (or substitutions) where we fix some variables as affine functions of other free variables. To illustrate the basic idea, consider a linear system only containing the equation $x_{i, j} + w = b$ where $b \in \bF_2$ and $w$ is some linear form not involving $x_{i, j}$. Then we can rewrite this equation as fixing $x_{i, j} = w+b$ and so we view this equation as simply affecting the variable $x_{i, j}$ while leaving all other variables completely free. 

At this point, we think of having lost control of the bit $x_{i, j}$ (which is determined by all the other variables) though all the other $(l-1)$ bits in block $x_i$ are still free. Efremenko, Garl\'{i}k and Itsykson \cite{egi24} observed that having control of $(l -1)$ bits in a block is almost as good as having control of all $l$ bits. Specifically, once we fix the  $(l-1)$ free bits, even though we do not know exactly which hole the pigeon flies to, we know that there are just two possible holes where the pigeon could go. So if we have not used up too many holes earlier, we can find some pair of holes to send this pigeon to without causing any collisions and moreover, only two additional holes are forbidden from being used in the future. 

More generally, \cite{egi24} showed that this argument works for any linear system which is \emph{safe}. A linear system of rank $r$ is said to be safe if it can be rewritten as fixing $r$ variables from distinct blocks, say $I \subseteq [m]$, as affine functions of the other variables. Again, as long as the rank $r$ is not too large and we have not used too many holes earlier, \cite{egi24} showed that we can fix the free bits in $I$ to send these $r$ pigeons to $2r$ holes in a way that avoids collisions. So there is an affine restriction implying a rank $r$ safe system such that only $r$ blocks are fixed and only $2r$ holes are used up by the restriction.

To handle general linear systems, \cite{egi24} introduce the notion of closure, which has several useful properties. For our discussion here, we will only need to know the following which was shown in \cite{egi24}. For any linear system $\Phi$, there exists a set $\cl(\Phi) \subseteq [m]$, called the closure of $\Phi$, such that after applying suitable row operations, $\Phi$ can be viewed as consisting of two linear systems $\Phi_1$ and $\Phi_2$, such that $\Phi_1$ captures all equations implied by $\Phi$ which only depend on variables in $\cl(\Phi)$, and $\Phi_2$ is safe even when we treat all the variables in $\cl(\Phi)$ as fixed (or equivalently, remove variables in $\cl(\Phi)$ from $\Phi_2$). The main property that makes this decomposition useful is the additional bound $\rk(\Phi_1) \geq |\cl(\Phi)|$, where $\rk(\Phi_1)$ denotes the rank of $\Phi_1$.

This implies that to find an affine restriction $\rho$ which implies $\Phi$ and for which the fixed blocks are sent to distinct holes, it is enough to first fix all the blocks in $\cl(\Phi)$ to bits according to some bit-fixing restriction $\sigma$ which implies $\Phi_1$ and sends all pigeons in $\cl(\Phi)$ to distinct holes, and then find an affine restriction $\rho'$ implying $\Phi_2|_\sigma$ by using that it is safe as above. The desired affine restriction $\rho$, which implies $\Phi$, is just the union of $\sigma$ and $\rho'$. Since $\rk(\Phi_1) \geq |\cl(\Phi)|$, the number of blocks fixed by $\rho$ is $|\cl(\Phi)| + \rk(\Phi_2) \leq \rk(\Phi_1) + \rk(\Phi_2) = \rk(\Phi)$. 

Of course, there are some linear systems $\Phi$ such that there is no bit-fixing restriction $\sigma$ which implies $\Phi_1$ and also sends all the pigeons in $\cl(\Phi)$ to distinct holes.
So we need to separately ensure that such a restriction exists for the linear system $\Phi$ being considered. \cite{egi24} call linear systems satisfying this condition locally consistent. With this in mind, our goal becomes to find some node $v$ in the affine DAG at depth $d = c\sqrt{n}$ (for some constant $c$) such that the corresponding system $\Phi_v$ has small rank $O(\log S)$ and is locally consistent.

\paragraph{Finding a locally consistent node when all holes are available.}

When all holes are available, how to find such a node was first shown by Efremenko, Garl\'{i}k and Itsykson \cite{egi24} as part of their lower bound on bottom-regular $\res(\oplus)$ proofs of $\bphp_n^{n+1}$. Their argument is essentially a parity decision tree lower bound for $\coll_n^m$ on the uniform distribution, where they directly analyze the probability that a uniform random input sends the pigeons in the closure to distinct holes. Note that this is somewhat stronger than just requiring a random input to reach a node which is locally consistent. When some holes $U \subseteq \B^l$ are forbidden, we do not know how to directly extend their argument. Instead, we extend an alternative proof of the same lower bound on the uniform distribution which was given by Byramji and Impagliazzo \cite{bi25}. The argument in \cite{bi25}, while still relying on ideas from \cite{egi24}, additionally uses ideas from lifting theorems for parity decision trees \cite{cmss23, beame2023, bi25}. 

We now sketch the argument from \cite{bi25} for the uniform distribution. Their proof does not directly keep track of the closure while analyzing the PDT, but at the end it still allows one to find an injective partial assignment on the closure. 
For now, the reader may wish to view the argument as proving an $\Omega(\sqrt{n})$ lower bound on the depth of any PDT which solves $\coll_n^m$ with constant probability on the uniform distribution, by mimicking the standard argument for ordinary decision trees. Suppose the first parity query is $x_{i, j} + w$ for some linear form $w$. Condition on the value $b \in \bF_2$ of this parity. We would like to view the constraint $x_{i, j} + w =b$ as just affecting $x_{i, j}$, since it can be rewritten as $x_{i, j} = w + b$. Moreover, observe that since the bit $x_{i, j}$ is a uniform random bit which is independent of all the other bits, even after conditioning on $x_{i, j} = w + b$, the distribution of all the other bits remains uniform. Informally, this means that the PDT has not learned anything about the blocks $i' \neq i$. 

At this point, since the bit $x_{i,j}$ has been affected, we also reveal the other bits in block $i$. Following \cite{egi24}, we interpret the revealed bits as giving us a pair of potential holes where pigeon $i$ can go. For later queries, we do the same except that we first substitute any variables whose values  we have already conditioned on, so that the resulting equivalent parity query only depends on the free blocks. After $d$ such queries, the only way the PDT could have found a collision with certainty is if there exist two revealed blocks for which the revealed bits agree, i.e.~the corresponding pairs of holes intersect. By the union bound, the probability that this happens is at most $n \cdot \binom{d}{2}(2/n)^2 \leq 2d^2/n$ since in any iteration, we reveal at most one block and the probability that the revealed bits for a block are consistent with any fixed hole in $\B^l$ is $2/n$. 

Going back to our original goal, we need that if none of the revealed blocks collide above, then the linear system $\Phi_v$ is locally consistent, where $v$ is the node in the DAG corresponding to the leaf reached in the above PDT. As long as $|\cl(\Phi_v)|$ is small (which is implied by $\rk(\Phi_v)$ being small), this can be done by finding a total assignment  which is injective on $\cl(\Phi_v)$ and is consistent with the system $L$ containing all the parity constraints that we conditioned on during the PDT analysis, since $L$ implies $\Phi_v$. Let $I \subseteq \cl(\Phi_v)$ be the set of all blocks in the closure which are not fixed by $L$. Now to find such an assignment we only need to ensure that the blocks in $I$ are sent to distinct holes, which are also different from the holes already assigned to $\cl(\Phi) \setminus I$ by $L$. The other free blocks can be set arbitrarily and we then extend to a full assignment by using the constraints in $L$. 

For the uniform distribution, the fact that there exists such a node whose rank is $O(\log S)$ is now implied by the same bottleneck counting argument that was used for subcube DAGs, since any rank $r$ system has probability $2^{-r}$ of being satisfied under the uniform distribution.

\paragraph{Finding a locally consistent node when some holes are forbidden.}
We now explain how to extend the above argument when we can only use holes from some subset $A \subseteq \B^l$. This will use ideas from the randomized PDT lifting theorem in \cite{bi25}, which builds on \cite{cmss23, beame2023, bhat2024}. 

We assume that $|A| \geq 2n/3$. Consider the distribution $\mu$ which is uniform on $A^m$. For the uniform distribution on $(\B^l)^m$ discussed previously, to localize a parity $x_{i, j} + w$ to $x_{i, j}$, we used that $x_{i, j}$ is distributed uniformly and independently of all the other bits. This is no longer true for $\mu$ in general, but conditioning on all $x_{i, j'} \; (j' \neq j)$ still makes $x_{i, j}$ a uniform random bit with decent probability. Specifically, since $|A| \geq 2n/3$, there are only at most $n/3$ strings $y$ in $A$ such that $y^{\oplus j} \notin A$ (where $y^{\oplus j}$ is $y$ with the $j^{th}$ bit flipped). So the probability that the bits $x_{i, j'} \; (j' \neq j)$ together uniquely determine $x_{i, j}$ is at most $(n/3)/(2n/3) = 1/2$ and with the remaining probability, $x_{i, j}$ is uniformly distributed. In the case that $x_{i, j}$ is uniformly distributed, we can repeat the argument above to condition on the value of $x_{i, j} + w$ without affecting the distribution on all blocks $i' \neq i$, since we still have independence across different blocks. In the other case where $x_{i, j}$ is determined by the other bits in $x_i$, we simplify the current parity and repeat with some other variable which appears in it.

So in expectation, when simulating a parity query, we reveal at most two blocks before we can move down to the next parity query. By linearity, when analyzing a depth $d$ PDT, we reveal at most $2d$ blocks in expectation. By Markov's inequality, the probability we reveal more than $8d$ blocks is at most $1/4$. To give an upper bound on the probability that there is a potential collision among the revealed blocks, it now suffices to upper bound the probability that there is a potential collision among at most $8d$ revealed blocks. It is easy to see that the union bound argument we gave earlier for the uniform distribution also works for the distribution $\mu$. This probability is at most $|A|\binom{8d}{2}(2/|A|)^2 \leq O(d^2/n)$ since $|A| \geq 2n/3$. So for a small enough constant $c$, setting $d = c\sqrt{n}$, the probability that there is a potential collision among the revealed blocks is at most $1/2$. 

Finally, for the bottleneck counting argument to work with $\mu$, we also need that the probability that a rank $r$ system is satisfied under $\mu$ is $2^{-\Omega(r)}$. Such a bound was shown by Bhattacharya, Chattopadhyay and Dvo\v{r}\'{a}k \cite{bhat2024} for any distribution that is a product of balanced distributions (with some loss depending on the block size). However, since $\mu$ is just uniform on $A^m$, we can also use a simple counting argument to give the bound $(n/2|A|))^r \leq (3/4)^r$. 

Putting everything together, we obtain an $\Omega(nd/\log S) = \Omega(n^{1.5}/\log S)$ lower bound on the depth of any size $S$ affine DAG solving $\coll_n^m$.

\paragraph{Extending to $t$-$\bphp$.}

The above arguments easily generalize to $t$-$\bphp_n^m$ to give a depth lower bound of $\Omega(tn^{2-1/t}/\log S)$ for any size $S$ affine DAG for $t$-$\coll_n^m$. The only change is in the probability estimate for there being a potential $t$-collision among the revealed blocks during the PDT analysis. Now a union bound over all $t$-sets among $d$ revealed blocks gives that this probability is at most $|A|\binom{d}{t}(2/|A|)^t \leq \left(\frac{3ed}{t}\right)^t \frac{1}{n^{t-1}}$ which can be made less than any desired constant by choosing $d = ctn^{1-1/t}$ for a small enough constant $c$. Note that for $t = \log n$, this argument already gives the bound $D \log S = \Omega(n^2 \log n)$ stated in Theorem \ref{thm:tbphp}.

The stronger bound $D(\log S)^{1-1/t} = \Omega(tn^{2-1/t})$ in Theorem \ref{thm:tbphp} follows from allowing $d$ to be larger. Specifically, assuming $S$ is not too large, we can take $d = ctn^{1-1/t}(\log S)^{1/t}$ while still ensuring (by a different argument) that the probability of reaching a locally $t$-consistent node is at least $S^{-O(1)}$. 
Having a lower bound of $S^{-O(1)}$ on the probability of being good on the closure is sufficient for the random walk with restarts approach to work as noted by Alekseev and Itsykson \cite{alek2024}. 

\subsection{$\bphp$ with $n+o(n)$ pigeons and $n$ holes}
\label{subsec:s_bphp_overview}
Let us start by considering just $\bphp_n^{n+1}$. To go beyond depth $n^{1.5}$, we need to consider a hard distribution for $\coll_n^{n+1}$ which is not uniform, because for the uniform distribution, there is an ordinary decision tree of depth $\widetilde{O}(\sqrt{n})$ solving $\coll_n^{n+1}$. 

A natural hard distribution for $\coll_n^{n+1}$ is the uniform distribution over strings $x \in (\B^l)^{n+1}$ such that two blocks in $x$ are $0^l$ and all other blocks are distinct strings from $\B^l \setminus \{0^l\}$. Under this distribution, any ordinary decision tree solving $\coll_n^{n+1}$ with constant probability must make $\Omega(n)$ queries. It is also easy to see that this lower bound continues to hold if we only consider such strings where $x_{n+1}$ is fixed to $0^l$. In other words, under this simpler distribution, the problem is essentially to find the unique block $i \in [n]$ such that $x_i = 0^l$ when the blocks of $x \in (\B^l)^n$ form a permutation of $\B^l$.

We will use $\ppf_n$ to denote this problem of permutation inversion. Formally, $\ppf_n = \{(x, i) \mid x \text{ is a permutation of } \B^l \text{ and } x_i = 0^l\} \cup (\{x \mid \text{there exist } i \neq p \text{ such that } x_i = x_p \} \times [n])$. The latter set is simply to allow all outputs when the input is not a permutation. We now quickly sketch how the above distribution can be used to show that any size $S$ subcube DAG for $\ppf_n$ must have depth $D = \Omega(n^2/\log S)$. The overall approach stays the same as in the previous subsection. At any stage, there is some collection of available strings $A \subseteq \B^l$ which contains $0^l$ and is large, $|A| \geq 5n/6$. The current DAG $C$, which is obtained from the original DAG after simplifications in previous stages, finds $0^l$ when the input string is a permutation of $A$. 

We now find a node $v$ in $C$ at depth $d = n/6$ such that the corresponding partial assignment $\rho_v$ involves only $O(\log S)$ blocks and there is some partial assignment $\sigma_v$ implying $\rho_v$ which assigns distinct strings from $A$ to these blocks none of which is $0^l$. Consider the uniform distribution $\mu$ on permutations of $A$. The depth $d$ decision tree $T$ obtained by starting at the root of $C$ finds $0^l$ in such a random permutation with probability at most $d/|A| \leq 1/5$. So there is some node $v$ in $C$ such that the corresponding leaves in $T$ are reached with probability at least $(4/5)/S$ without $0^l$ having been found. 

To see that $\rho_v$ must involve only $O(\log S)$ blocks, let us give an upper bound on the probability that a partial assignment touching $r$ blocks is satisfied under $\mu$. For any block $i \in [|A|]$ which appears in $\rho_v$, there can only be at most $n/2$ consistent assignments for that block. Using this, we can argue inductively that the probability that $\rho_v$ which involves $r$ blocks is satisfied is at most $\prod_{i = 1}^r \frac{n/2}{|A| - i+1} \leq (n/2(|A| - r))^r$. Since $v$ is at depth $d$ in $C$, we also know $r \leq d$. So this quantity can be bounded by $\left(\frac{n/2}{2n/3}\right)^r$ where we also used the assumption $|A| \geq 5n/6$. Combining this with the lower bound $4/5S$, we obtain $r \leq O(\log S)$. 

After applying an appropriate restriction to $C$ which only uses up $O(\log S)$ strings from $A$ and decreases the depth of the DAG by $d$, we move on to the next stage if $A$ still satisfies $|A| \geq 5n/6$. So we can repeat this for  $\Omega(n/\log S)$ stages to end up with a DAG of depth $D - \Omega(dn/\log S)$, which gives the desired bound.

\paragraph{Affine DAGs for Permutation Inversion.} We would now like to adapt the above argument to make it work for affine DAGs. To get started, consider only the first stage when $A = \B^l$, so that the distribution $\mu$ is just the uniform distribution over permutations of $\B^l$. One reason why we cannot immediately apply ideas from the previous subsection to parity decision trees for $\mu$ is that the blocks are not independent. In the previous subsection, when trying to localize a parity $x_{i, j} + w$ to $x_{i, j}$, simply conditioning on all $x_{i, j'} (j' \neq j)$ was enough to make the bit $x_{i, j}$ independent of all variables that appear in the form $w$. 

Another part of the proof where independence was used previously was for estimating the probability that a rank $r$ system is satisfied. In the case of permutations, we cannot naively first condition on all non-pivot blocks and then give a small universal upper bound on the conditional probability for the simplified linear system. So we would like to show an affine version of the above argument used for subcube DAGs which iteratively conditions on at most $r$ blocks to estimate the probability that a partial assignment of width $r$ is satisfied.

The key observation which is used to handle both the above issues is a suitable conditioning that lets us find a uniform random bit which is independent of all other unfixed blocks with decent probability. As warm-up, let us try to show that any non-zero linear form $v$ is roughly uniform on the distribution $\mu$, in the sense that $\Pr[v = 1] \in [1/2 - \delta, 1/2 + \delta]$ for some constant $\delta < 1/2$. First, let us note that this cannot actually be true for all parities. Fix $j \in [l]$. For any permutation $x$ of $\B^l$, exactly half of the bits $x_{i,j}$ are $1$ and in particular, $\sum_{i \in [n]} x_{i, j}$ is fixed over $\mu$. These are essentially the only fixed parities however. 

Suppose we have a parity $v$ which does not lie in the span of the above $l$ parities which are always fixed for permutations. Let us assume  that $x_{1,1}$ appears in $v$ and there are at most $n/2$ many $i \in [n]$ such that $x_{i,1}$ appears in $v$. Condition on the values of $x_{1, h}, h \neq 1$. Additionally condition on the unique $p \in [n]\setminus \{1\}$ for which we have $x_{1, h} = x_{p, h}$ for all $h \in [l] \setminus \{1\}$. The index $p$ is uniformly distributed among $[n]\setminus \{1\}$ which implies that with probability at least $1/2$, $x_{p,1}$ does not appear in $v$. Conditioned on $x_{p,1}$ not appearing in $v$, we now see that the parity $v$ is uniformly distributed since $x_{i, 1}$ is a uniform random bit appearing in $v$ which is independent of all the other variables in $v$. So for each $b \in \B$,  the overall probability that $v = b$ is at least $1/4$.

\paragraph{PDTs for Permutation Inversion.} We now explain how the above conditioning can be used to prove an $\Omega(n)$ lower bound on the depth of deterministic PDTs that solve $\ppf_n$ with constant probability. This would imply that any randomized PDT for $\ppf_n$ must have depth $\Omega(n)$, which also gives the same lower bound for $\coll_n^{n+1}$. 

We will give a simulation similar to what was done for the uniform distribution, where we reveal a small number of blocks in expectation when analyzing a parity query. Since the distribution here is over permutations, we will keep track of a set $A$ of available strings such that conditioned on the events so far, all unrevealed blocks form a uniform random permutation of $A$. In the beginning, we have $A = \B^l$. We also keep a set $L$ containing linear equations describing the parities assigned to the revealed/fixed blocks.

Suppose the first query is $v$. If $v$ lies in the span of the above $l$ parities which are fixed for every permutation, then we do not reveal any block and can simply move to the appropriate child. On the other hand, if $v$ does not lie in their span, we use the above conditioning. Suppose $j \in [l]$ is such that the projection of $v$ onto the variables $x_{i, j}, i \in [n]$ is not in $\{0, \sum_{i \in [m]}x_{i, j}\}$. If this projection has more than $n/2$ variables, by adding $\sum_{i = 1}^n x_{i, j} + \sum_{y \in A} y_j$ to this parity, we get an equivalent parity for which at most $n/2$ variables among $x_{i, j}, i \in [n]$ appear in the parity. 

Let $i \in [n]$ be such that $x_{i, j}$ appears in $v$, say, $v = x_{i, j} + w$. Now sample $x_{i, h}, h \neq 1$ uniformly according to $A$ (which in the beginning is just $\B^l$). Then pick $p \in [n]\setminus \{i\}$ uniformly which will be the unique partner of $i$ such that $x_{i, h} = x_{p, h}$ for all $h \neq j$. We also have $x_{i, j} = x_{p, j} + 1$. With probability at least $1/2$, $x_{p, j}$ does not appear in $v$. In this case, we can simply condition on the value $b$ of $x_{i, j} + w$ and localize the constraint $x_{i, j} = w + b$. Observe that in this case we have fixed the blocks $x_i, x_p$ as affine functions of the other blocks. Moreover, the distribution on the other blocks is uniform on all strings in $A \setminus \{x_i, x_p\}$, where we know which two strings lie in $\{x_i, x_p\}$ even though we do not know which is assigned to $x_i$.
In the unlucky case where $x_{p, j}$ appears in $v$, we just sample $x_{i, j}$ uniformly which also determines $x_{p, j}$ and then simplify the parity query without having made progress in this iteration.

Later iterations work in essentially the same way. We first simplify the parity according to constraints in $L$ and check if that parity is implied by the free blocks being a permutation of $A$. If not, then we do the sampling as above, now also taking care to check if $x_{i, j}$ is already determined by $x_{i, h}, h \neq j$. If $x_{i, j}$ is already determined, there is no partner and we simply move on to the next iteration (like what we did in the previous subsection). 

In each iteration, the size of $A$ decreases by at most $2$. So during the first few iterations, $A$ stays large and the probability that $x_{i, j}$ is determined by other bits of $x_i$ is not too large. Specifically, after $\hat{d}$ iterations, $|A| \geq n - 2\hat{d}$, so this probability is at most $2\hat{d}/(n-2\hat{d})$. Conditioned on this not happening, the probability that the partner of $i$ lies in the parity being simulated is at most $1/2$. So with probability at least $1/2 - \hat{d}/(n-2\hat{d})$, we move down the parity decision tree in that iteration.

Let us now upper bound the probability that we reveal $0^l$ during the above simulation, which is also an upper bound on the probability that the PDT successfully solves $\ppf_n$ with certainty. Suppose the PDT being simulated has depth $d = n/100$. It suffices to bound the probability that $0^l$ is revealed in the first $\hat{d}$ iterations or the simulation runs for more than $\hat{d}$ iterations, where we set $\hat{d} = 8d$. The probability that $0^l$ is not revealed in the first $\hat{d}$ iterations is at least $\prod_{i = 1}^{\hat{d}}(1-\frac{2}{n-2(i-1)}) = 1 - 2\hat{d}/n$. 

To estimate the probability that the simulation runs for more than $\hat{d}$ iterations, note that when this happens, there must be at least $\hat{d}-d$ iterations where the simulation did not move down the PDT. In each of the first $\hat{d}$ iterations, the probability that we do not go down in that iteration is at most $1/2 + \hat{d}/(n-2\hat{d})$. So by Markov's inequality, the probability that there are $\hat{d}-d$ such bad iterations is at most $\left(\frac{1}{2} + \frac{\hat{d}}{n-2\hat{d}}\right)\frac{\hat{d}}{\hat{d}-d}$. Combining the two bounds, we get that the probability of revealing $0^l$ is at most $9/10$. 

\paragraph{Probability of linear system being satisfied by a random permutation.} We would now like to prove some version of the statement that for every rank $r$ linear system $\Phi$, where $r \leq cn$ for some small constant $c > 0$, the probability that $\Phi$ is satisfied over the distribution $\mu$ is $2^{-\Omega(r)}$. 

Using the conditioning above, it is not hard to give an inductive argument to prove the bound $2^{-\Omega((r-l)/l)}$ for small $r$. Specifically, after picking any equation $x_{i, j} + w = b$ in $\Phi$ which does not hold with probability $1$, we apply the above conditioning to make $x_{i, j}$ a uniform random bit which is independent of $w$ with decent probability. In all cases, we consider a system $\Phi'$ which is implied by $\Phi$ but which only depends on the free blocks. Since the number of fixed blocks is at most $2$, we can find such a $\Phi'$ whose rank is at least $r - 2l$. The probability that $\Phi'$ is satisfied can now be estimated by induction. When estimating the overall probability that $\Phi$ is satisfied, we do better by a factor of $1/2$ compared to what is given by induction in the case where $x_{i, j}$ is a uniform random bit, which itself happens with constant probability. This gives the bound $2^{-\Omega((r-l)/l)}$.

To avoid the $\log n$ multiplicative factor loss in the exponent, we recall that for our final goal of finding a small affine restriction, we only really need an upper bound on $|\cl(\Phi)| + \rk(\Phi[\setminus \cl(\Phi)])$. It follows from arguments in \cite{egi24, ei25} that for any system $\Phi$, there is a safe system $\Psi$ implied by $\Phi$ whose rank is $|\cl(\Phi)| + \rk(\Phi[\setminus \cl(\Phi)])$. Hence, we call this quantity the safe dimension of $\Phi$. So it suffices to give an upper bound just for safe systems. For the argument, it is convenient to assume that the safe system actually contains all the $l$ equations implied by being a permutation. We observe that for such a safe system $\Psi$, when doing the above inductive argument, the safe dimension only decreases by a constant because of the fixed blocks (instead of $\log n$ above). So we get a $2^{-\Omega(r)}$ upper bound for such a safe system whose rank is $r + l$ when $r \leq cn$ for some small constant $c$.

\paragraph{Finding an affine restriction.} Using the above arguments, we can find a node $v$ at depth $\Omega(n)$ which is reached without revealing $0^l$ in the above simulation and for which $\Phi_v$ has safe dimension $O(\log S)$. We now explain how to find an affine restriction $\rho$ implying $\Phi_v$ for which the problem $\ppf_n$ does not become easy after applying $\rho$. Finding a bit-fixing restriction satisfying the system supported on the closure of $\Phi_v$ is again easy. 

With the remaining safe system, we need to be more careful. Since we are dealing with permutations, we can now no longer afford forbidding two strings for each fixed block. The affine restriction we seek must use up exactly the same number of strings as the number of blocks fixed.

Suppose the safe system $\Phi$ for which we want to find an affine restriction has rank $r$. Roughly speaking, we show that there is an affine restriction $\rho$ implying $\Phi$ with the following properties. The restriction $\rho$ fixes $2r$ blocks in such a way that the $2r$ blocks are paired up, and for each such pair of blocks, we assign a pair of strings that differ on a single bit. The constructed affine restriction ensures that for each such pair of blocks, the two blocks always form a permutation of the strings assigned. Moreover, any two such assigned pairs of strings are disjoint to ensure that all the fixed blocks are always assigned distinct strings by this affine restriction.  Finding suitable blocks for this affine restriction is done via arguments that are similar to those in \cite{egi24, ei25} for proving structural properties of safe systems, closure and amortized closure.


The above ingredients can now be combined in the same way as we did previously to obtain the desired $\Omega(n^2/\log S)$ bound on the depth of any size $S$ DAG solving $\ppf_n$. 

\paragraph{Extending to $\bphp_n^{n+k}$.} 

To extend the above argument to $\bphp_n^{n+k}$, $k = o(n)$, we can consider the variant of $\ppf_n$ where on a permutation $x$, the goal is to find any $i \in [n]$ such that $x_i$ belongs to some fixed collection $B$ of $k$ strings from $\B^l$. 
Let $\ppf_{n, k}$ denote this problem. The above arguments generalize to show that any zero-error randomized PDT for $\ppf_{n, k}$ which succeeds with constant probability requires depth $\Omega(n/k)$ and that any depth $D$ affine DAG for $\ppf_{n, k}$ must have size $\exp(\Omega(n^2/kD))$.

By considering PDTs of larger depth, $d = \eps n$ for some small constant $\eps$ (independent of $k$), the bound can be strengthened to show that for all $D \leq \gamma n^2/k$ (where $\gamma$ is some constant), any depth $D$ affine DAG for $\ppf_{n, k}$ must have size $\exp(\Omega(n^2/D))$. Note that this gives the bound stated in Theorem \ref{thm:s_bphp} for depth at most $\gamma n^2/k$. 

To also handle larger depths, we work with $\coll_n^{n+k}$ directly. While we do not discuss permutation inversion in our actual proofs, the core ideas remain those that we discussed above. 

\subsection{Lifting to bounded-depth $\res(\oplus)$}
\label{subsec:lift_overview}
Here we will give a sketch of Theorem \ref{thm:lift} when each block contains just one variable. So $\varphi$ is a formula on $n$ variables and we assume that there is some $(p, q)$-DT-hard set $\cP$ for $\varphi$. Let $\cR^\varphi$ denote the false clause search problem corresponding to $\varphi$. The proof works for any constant size stifling gadget \cite{cmss23} but for concreteness we will consider the $3$-bit Majority $\maj_3$ gadget. We want to show that any size $S$ affine DAG $C$ for $\cR^\varphi \circ g$ (where $g$ is $\maj_3$) must have depth $\Omega(pq/\log S)$.

The strategy stays the same at a high level. We keep track of a partial assignment $\rho \in \cP$ such that $|\fix(\rho)| \leq p$ and an affine DAG solving $\cR^\varphi|_{\rho} \circ g$. In each phase, we find a node $w$ in the DAG at depth $d = q/48$ such that there is some affine restriction $\tau$ which implies the linear system at $w$. Moreover, $\tau$ fixes only $O(\log S)$ blocks, say $I \subseteq \free(\rho)$ and there is some restriction $\rho'$ on the variables of $\varphi$ also fixing bits in $I$, such that $\tau$ induces the restriction $\rho'$ when we apply $g$ to each of the blocks in $I$. Additionally $\rho \cup \rho' \in \cP$. So after applying the affine restriction $\tau$ to the DAG, we get a DAG which solves $\cR|_{\rho \cup \rho'} \circ g$. We can repeat this until the restriction $\rho$ has fixed more than $p$ bits. For this to happen, there must have been $p/O(\log S)$ iterations. This gives the desired bound $\Omega(pd/\log S) = \Omega(pq/\log S)$.

\paragraph{Finding a suitable node $w$ at depth $d$.} To find $w$, we will lift the hard distribution $\nu_{\rho}$ for decision trees to a hard distribution for parity decision trees. 

For $b \in \B$, define $\mu_b$ to be the distribution which is uniform on $g^{-1}(b)$. These distributions have the following useful property noted by Bhattacharya, Chattopadhyay and Dvo\v{r}\'{a}k \cite{bhat2024}. For any $j \in [3]$, if we sample $y \sim \mu_b$, then the probability that after revealing all bits in $y$ except $y_j$, the bit $y_j$ is uniformly distributed is $1/2$. They use this to show that for any product distribution $\mu_z \coloneqq \mu_{z_1} \times \mu_{z_2} \times \dots \times \mu_{z_n}$, for any linear system $\Phi$ of rank $r$, the probability that $x \sim \mu_z$ satisfies $\Phi$ is at most $2^{-\Omega(r)}$. In particular, this holds for all mixtures of such product distributions.
The hard distribution for parity decision trees will be $\nu_{\rho} \circ (\mu_0, \mu_1)$ which is defined in the following way. Sample $z \sim \nu_{\rho}$ and then return $x \sim \mu_z$.

Byramji and Impagliazzo \cite{bi25} showed that the above property also leads to a lifting theorem for randomized parity decision trees. They show that for any PDT $T$ of depth $d$ on $(\B^3)^n$, there is a relatively efficient randomized decision tree $T'$ which on input $z \in \B^n$, returns the leaf in $T$ reached by $x \sim \mu_z$. For each $z$, this decision tree only makes $2d$ queries in expectation. By combining this lifting theorem with the assumption on $\nu_{\rho}$, it can be shown that when we consider the $d$ step random walk on the affine DAG starting at the root according to $\nu_{\rho} \circ (\mu_0, \mu_1)$, with probability at least $1/6$ we end at a node where there is some closure assignment such that its induced partial assignment $\sigma$ satisfies $\rho \cup \sigma \in \cP$.

Now combining the two bounds above, we find a node $w$ such that the corresponding linear system $\Phi$ has rank $r \leq O(\log S)$. Moreover, there is some bit-fixing restriction for $\cl(\Phi)$ such that the corresponding induced partial assignment $\sigma$ satisfies $\rho \cup \sigma \in \cP$.

\paragraph{Finding an affine restriction.} To finish the argument, we also need an affine restriction which implies $\Phi$ after having substituted for the closure assignment. The resulting system, say $\Phi'$, is safe. Let $(i_1, j_1), (i_2, j_2), \dots, (i_d, j_d)$ be from distinct blocks such that $\Phi'$ can be expressed as fixing variables $x_{i_t, j_t} ,\; t \in [d]$ in terms of other variables. Let $I = \{i_t \mid t \in [r]\}$.
Assuming $|\fix(\rho \cup \sigma)| \leq p$, we can use the distribution $\nu_{\rho \cup \sigma}$ to find a partial assignment $\sigma'$ on the unlifted variables fixing $I$ such that $\rho \cup \sigma \cup \sigma' \in \cP$. We now fix all the bits $(i_t, j), \; t \in [r], j \neq j_t$ to force  $g(x_{i_t}) = \sigma'(z_{i_t})$. The desired affine restriction now combines the above bit-fixing restrictions with the rewritten $\Phi'$.

\subsection{Comparison with prior work}
\label{subsec:comp}

\paragraph{Restarting the random walk.}
A basic difference between our proofs and earlier proofs comes from how we use affine restrictions to simplify the affine DAG before picking our hard distributions. 

Let us briefly recall how \cite{alek2024} restart their random walks. Having found a node whose (relative) rank is small (or in \cite{ei25}, its amortized closure) and for which there is a good closure assignment, they pick their hard distribution to be simply uniform over all inputs which are consistent with this good closure assignment and which satisfy the system $\Phi$ at that node (which after applying the closure assignment is safe). In other words, their distribution keeps some blocks fixed and the remaining blocks are uniform but conditioned on satisfying a particular safe system. 
This is quite natural in their case since the proof of their random walk probability lower bound already involves analyzing the uniform distribution conditioned on satisfying a safe system.

\cite{bc25} extend this approach to lifted distributions by proving a conditional fooling lemma for gadgets with exponentially small Fourier coefficients. Informally, their key technical lemma lets them show that a hard lifted distribution remains relatively hard even when conditioned on the input satisfying any given safe system, and the probability of a linear system being satisfied under such a conditional distribution also remains as small as one would hope.

Our approach sidesteps trying to understand distributions after conditioning on satisfying a safe system. We do this by fixing more blocks than in prior work\footnote{Readers familiar with amortized closure may wish to think of this as fixing all of the amortized closure instead of just closure, though our actual proof also allows for other choices of blocks.}, though this fixing involves parities instead of just bits to also take care of the linear system. After that, we are free to pick any hard distribution on the free blocks. One can phrase our hard distribution as coming from a hard distribution on some subset of blocks and then deterministically extending to the other blocks by means of some known linear system which is almost bit-fixing in the sense that for each of the fixed blocks, there is at most one variable which depends on the free blocks. We find this alternative perspective to be less intuitive  compared to using affine restrictions to first simplify the DAG in each phase before picking the distribution. 

One advantage of simplifying the DAG using an affine restriction is that giving an upper bound on the probability of a linear system being satisfied under the new distribution is considerably simplified\footnote{For $\bphp_n^{n+o(n)}$, this still requires some work, but for a different reason.}. Specifically, we no longer need to explicitly track which equations in the system are already implied by the fixed system (thereby holding with probability $1$) and which equations contribute to the probability being small.

\paragraph{Random walk analysis.}
Another place where earlier work requires understanding distributions conditioned on satisfying some safe system is for analyzing the actual random walk to see how the probability of being good on the closure changes as the closure grows. We get around this by instead relying on ideas from lifting theorems for parity decision trees \cite{cmss23, beame2023, pod25, bi25}. Even though this is not what first led us to our approach,  this is similar in spirit to recent work in communication complexity where ideas from lifting theorems have been useful to prove lower bounds for problems that do not have a composed form. We will discuss this in more detail later in this subsection.

In the worst case, our random walk analysis can reveal all blocks. This is unlike closure whose size is always bounded by the sum of the number of queries and the rank of the starting linear system. Still, with high probability, in our analysis we do not reveal too many blocks. Let us also point out that the only property of closure we directly use to find a good closure assignment from a successful run of the simulation is an upper bound on the size of the closure.

\paragraph{Amortized closure.}
Though our work does not explicitly use the notion of amortized closure \cite{ei25}, there are some aspects of our proofs which are closely related to it. Specifically, we define the notion of safe dimension which is the largest dimension of a safe collection contained in the span of the linear forms under consideration. This quantity can be shown to be equal to the size of the amortized closure \cite{ei25}, but we think that the term safe dimension provides a  simpler description of how it is used in our proofs. We do not rely on any other properties of the amortized closure, such as the canonical choice of blocks from which the linearly independent columns are picked or that its size only grows by at most one on adding a linear form.

\paragraph{Structure-vs-randomness in communication complexity.}
Many recent works in communication complexity \cite{wang2023communication, yz24, mao2025gadgetless, goos2025quantum, beame2025multiparty, huang2025min, riazanov2025searching} rely on ideas from query-to-communication lifting theorems \cite{raz1997separation,  goos2016rectangles, gjpw18, goos2020query, lovett2022} to prove lower bounds for problems not having a composed form. Our distributional PDT lower bounds for collision-finding can be viewed as being in the same vein combining ideas from \cite{cmss23, bi25} with \cite{egi24}. It is worth noting that one of the problems for which this structure-vs-randomness framework in communication complexity has been useful is the collision-finding problem \cite{yz24, beame2025multiparty} and these communication complexity bounds have led to better lower bounds in proof complexity.

We should also point out that, in retrospect, the proof of the distributional PDT lower bound from \cite{egi24} can be phrased in terms of min-entropy in a way that is at a very high level similar to (but simpler than) \cite{yz24, beame2025multiparty}. Indeed, a linear system $\Phi$ is safe if the uniform distribution on inputs from $\Phi$ has blockwise min-entropy rate $l-1$ (where $l$ is the block size) and the analysis in \cite{egi24} can be viewed as revealing some relevant blocks right before the closure is about to change, i.e.~the blockwise min-entropy rate on the unfixed part is about to drop below $l-1$.

In contrast to the above block min-entropy approach which also underlies \cite{alek2024, bc25} in some form, the lifting theorems for parity decision trees which we adapt here \cite{cmss23, bi25}\footnote{The simulations in \cite{beame2023, pod25} use both localization and entropy-based arguments.} that are based on localizing parities to a block seem to have more in common with the earlier thickness-based approach to query-to-communication lifting theorems \cite{raz1997separation, goos2018deterministic, chattopadhyay2019simulation, wyy17}. Very roughly, these thickness-based lifting theorems consider some form of conditional min-entropy of a block with respect to all other unfixed blocks. When this quantity becomes small for some block, instead of fixing that block completely, it is fixed relative to all the other blocks in a suitable way.
This is similar to the approach in \cite{cmss23} which can be interpreted as keeping track of conditional entropy of a block with respect to all later blocks. When this is about to become too small for some block, we consider the intersection with a suitable subcube on that block ensuring that the output of the gadget when applied to that block is consistent with the actual input.
\section{Preliminaries}
\label{sec:prel}

For a search problem $\cR \subseteq \B^N \times \cO$, for any $x \in \B^N$, we use $\cR(x)$ to denote $\{o \in \cO \mid (x, o) \in \cR\}$.
For a string $y \in \B^l$ and $j \in [l]$, we use $y^{\oplus j}$ to denote the string which is the same as $y$ but with the $j^{th}$ bit flipped ($(y^{\oplus j})_j \neq y_j$).

We will mostly consider inputs $x$ from $(\B^l)^m$ viewed as $m$ blocks of $l$ bits (though we should warn the reader that our notation for the block size and the number of blocks will change in Section \ref{sec:lifting}).

It will be convenient to allow parities to also contain a constant term. So a parity is just a linear polynomial on $(\B^l)^m$ over $\bF_2$, $b + \sum_{i \in [m], j \in [l]} c_{i, j} x_{i, j}$ where $b \in \bF_2$ and $c_{i, j} \in \bF_2$ for all $i \in [m], j \in [l]$. 
We will also use the term affine function to refer to such a linear polynomial when discussing affine restrictions (defined later in Subsection \ref{sec:aff_dags}).
The support of a parity $P$, denoted $\supp(P)$, is the subset of $[m]$ containing all blocks which appear in $P$. For $B \subseteq [m]$, we use $\cL_{B}$ to denote the collection of all parities $P$ whose support is contained in $B$, $\supp(P) \subseteq B$.

A linear form is a homogeneous polynomial of degree $1$. For a set of linear forms $V$, we use $\langle V \rangle$ to denote the span of $V$ and $\dim(V)$ to denote the dimension of $\langle V\rangle$.

If $v$ is a linear form and $b \in \bF_2$, $v = b$ is a linear equation. Sometimes we will write $P = b$ where $P$ is a parity, with the intended linear equation being $v = b + c$ if $P = v + c$ where $v$ is a linear form and $c \in \bF_2$. A linear system is a collection of linear equations. For any satisfiable linear system, we will assume that it is represented by a collection of linearly independent equations. 

For a linear system $\Phi$, we use $L(\Phi)$ to denote the set of linear forms appearing in $\Phi$. We use $\rk(\Phi)$ to denote the rank of $\Phi$. We use $\langle \Phi \rangle$ to denote the collection of all equations obtained by taking linear combinations of some equations in $\Phi$. In other words, $\langle \Phi \rangle$ contains all equations implied by $\Phi$. Even though we will  allow for inconsistent linear systems in the definition of affine DAGs below, all the linear systems $\Phi$ for which we will actually use the notation $\rk(\Phi)$ or $\langle \Phi \rangle$ will be satisfiable.


\subsection{Safe collections and closure}
    
We use the notions of closure and safe collections of linear forms, introduced by \cite{egi24}. 

A set  $V$ of linear forms is said to be safe if there is no subset $W \subseteq \langle V \rangle$ such that the linear forms in $W$ are linearly independent and the support of $W$ ($\cup_{P \in W} \supp(P)$) has size less than $|W|$. 

A set $D$ of linearly independent forms is said to be dangerous if $\supp(D) < |D|$. So the above definition of safe can be rephrased as saying that $V$ is safe if its span contains no dangerous sets.

For $I \subseteq [m]$, $V[\setminus I]$ denotes the linear forms obtained from $V$ by setting all variables in blocks in $I$ to $0$. The closure of $V$, denoted $\cl(V)$, is the minimal set $I \subseteq [m]$ such that $V[\setminus I]$ is safe. (\cite{egi24} showed that there is a unique minimal set satisfying the condition defining closure.)
Abusing notation, for a linear system $\Phi$, we will use $\cl(\Phi)$ to denote $\cl(L(\Phi))$. Similarly we say that $\Phi$ is safe if $L(\Phi)$ is safe.
The fact that safety and closure do not depend on the choice of basis \cite{egi24} allows us to freely apply invertible operations to the rows of a linear system or a collection of linear forms.

\begin{theorem}[\cite{egi24}] \label{lem:safe}
    Let $V$ be a collection of $k$ independent linear forms and $M$ the corresponding coefficient matrix. $V$ is safe if and only if we can pick $k$ variables, no two from the same block, such that the corresponding columns in $M$ are linearly independent.
\end{theorem}

Noting that the definition of safety really depends on just the span of $V$, we can also use the above theorem for sets $V$ of not necessarily independent forms if we replace $k$ in the second statement above by $\dim(V)$.


\begin{lemma}[\cite{egi24}] \label{lem:size_bound}
If $V$ is a collection of linear forms, then
    $|\cl(V)| + \dim(V[\setminus \cl(V)]) \leq \dim(V)$.
\end{lemma}

For a collection $V$ of linear forms over $(\B^l)^m$, define the safe dimension of $V$, denoted $\sdim(V)$, to be the maximum dimension of a safe collection $W \subseteq \langle V \rangle$.

The inequality in the following lemma sharpens Lemma \ref{lem:size_bound}. It is implicit in \cite{egi24} and can also be obtained from the lemmas in \cite{ei25}. We will need a somewhat stronger statement, so we provide a proof. 

\begin{lemma} \label{lem:sdim_bound}
    For any collection $V$ of linear forms, $\sdim(V) \geq |\cl(V)| + \dim(V[\setminus \cl(V)])$. Moreover, for any set $U$ of linearly independent forms such that $U \subseteq \langle V \rangle$ and $\dim(U[\setminus \cl(V)]) = |U|$, there exists a safe collection $W \subseteq \langle V \rangle $ such that $U \subseteq W$ and $\dim(W) = |\cl(V)| + \dim(V[\setminus \cl(V)])$.
\end{lemma}
\begin{proof}
Note that showing the latter statement also implies $\sdim(V) \geq |\cl(V)| + \dim(V[\setminus \cl(V)])$ since we can take $U = \emptyset$. 

So consider any $U \subseteq \langle V \rangle $ satisfying $\dim(U[\setminus \cl(V)]) = |U|$. 
 We will modify the closure algorithm from \cite{egi24} to find $W$. To prove the safety of $W$, we will find linearly independent $\dim(W) = |\cl(V)| + \dim(V[\setminus \cl(V)])$ columns  of the matrix for $W$ which lie in distinct blocks. This suffices by Theorem \ref{lem:safe}.

 In the algorithm below, we make use of some facts shown by \cite{egi24} to which the reader is referred for more details.

 \begin{enumerate}
     \item Set $S = \emptyset$, $W = \emptyset$, $I = \emptyset$.
     \item While $\langle V[\setminus S]\rangle$ contains some dangerous set:
     \begin{enumerate}
         \item Let $D = \{v_1, v_2, \dots, v_{h+1}\}$ be a minimal dangerous set in $\langle V[\setminus S] \rangle$. Set $T = \supp(D)$ where $|T| = h$ since otherwise $D$ would not be minimal.
         \item Set $X = \{v_1, v_2, \dots, v_h\}$ which is safe.
         \item By Theorem \ref{lem:safe}, there exist indices $(i_1, j_1), (i_2, j_2), \dots, (i_h, j_h) \in T \times [l]$ such that the corresponding columns in the matrix for $X$ are linearly independent and $\{i_1, i_2, \dots, i_h\} = T$. Add these $h$ column indices to $I$.
         \item For each $v_i \in X$, add to $W$ a linear form $w \in \langle V \rangle$ such that $w[\setminus S] = v_i$.
         \item $S \gets S \cup T$.
     \end{enumerate}
     \item $V[\setminus S]$ is now safe. Set $d = \dim( V[\setminus S])$.
     \item By Theorem \ref{lem:safe}, there exist indices $(i_1, j_1), (i_2, j_2), \dots, (i_d, j_d) \in ([m]\setminus S) \times [l]$ such that the corresponding columns in the coefficient matrix for $V[\setminus S]$ are linearly independent and the indices lie in $d$ distinct blocks. Add these to $I$.
     \item Find a basis $B$ of $V[\setminus S]$ which extends $U[\setminus S]$.
     \item Add all forms in $U$ to $W$.
     \item For each $v \in B \setminus U[\setminus S]$, add to $W$ a linear form $w \in \langle V \rangle$ such that $w[\setminus S] = v$.
 \end{enumerate}
 \cite{egi24} showed that when the while loop terminates, the set $S$ is the closure of $V$.  By assumption, we know that $U[\setminus \cl(V)]$ is linearly independent and $U \subseteq V$. So it makes sense to find a basis $B$ of $V[\setminus \cl(V)] = V[\setminus S]$ which extends $U[\setminus \cl(V)]$. We have $U \subseteq W$ and $W \subseteq \langle V \rangle$ by construction. 
 
 We now show $|I| = |W| = |\cl(V)| + \dim (V[\setminus \cl(V)])$. In each iteration of the loop, we add as many forms to $W$ as number of blocks added to $S$. The same number of elements is also added to $I$. At the end of the loop, we have $S = \cl(V)$ and so up to that point $W$ contains $|\cl(V)|$ forms. After that, we add $d = \dim( V[\setminus S]) = \dim(V[\setminus \cl(V)])$ more forms to $W$ and $d$ indices to $I$, where we are using that $B$ is a basis of $V[\setminus S]$ which extends $U[\setminus S]$. So $|B| = d$ and $|B \setminus U[\setminus S]| = |B| - |U[\setminus S]| = |B| - |U|$.

 Finally, to check that $W$ is safe, we verify that the columns of the coefficient matrix for $W$ which are in $I$ are linearly independent. Here we assume that the rows are ordered according to when they were added by the algorithm, so the forms that were added in the first iteration appear at the top. Consider the submatrix $M$ containing only the columns of the $W$ matrix which lie in $I$ and reorder the columns according to when they were added to $I$ by the above procedure. $M$ can be seen to have a block lower triangular structure since whenever some $w$ is added to $W$ during the loop, $w[\setminus S]$ is only supported in the set $T$ being considered in that iteration. Moreover each such block on the diagonal is precisely the matrix for the safe forms $\{v_1, v_2, \dots, v_h\} \subseteq V[\setminus S]$ considered in that iteration except we are only considering the columns that are indexed by $T$. By the choice of $T$, each such block has full rank. This is also true for the last block (possibly empty) which contains $B$ restricted to the $d$ columns added to $I$. It is easy to see by induction on the blocks, that any such block lower triangular matrix, where each block on the diagonal has full rank, must also have full rank. This finishes the proof.
\end{proof}

Though we do not really need this, it can also be verified that $\sdim(V) \leq |\cl(V)| + \dim(V[\setminus \cl(V)])$. This implies $\sdim(V) = |\cl(V)| + \dim(V[\setminus \cl(V)])$. Thus, safe dimension is the same as the size of the amortized closure \cite{ei25}, but since we do not actually use the notion of amortized closure in our proofs, we have given it a name which better describes how it is used.

\subsection{Affine DAGs and affine restrictions} \label{sec:aff_dags}

We will use affine DAGs to reason about $\res(\oplus)$ proofs in a top-down way as in prior work. These are essentially the same as $\res(\oplus)$ refutation graphs \cite{egi24, alek2024}. We prefer the term affine DAGs so that we can discuss search problems that are not necessarily false clause search problems. 

An affine DAG for a search problem $\cR \subseteq \B^N \times \cO$ is a directed acyclic graph with one source where every internal node has outdegree $2$ and it is labeled in the following way:
\begin{itemize}
    \item At each internal node $v$, we have a linear system $\Phi_v$ on $\B^N$ and a parity query $P_v$ on $\B^N$.
    \item At each sink $w$, we have a linear system $\Phi_w$ on $\B^N$ and an output label $o_w \in \cO$.
\end{itemize}
The linear system at the source node (also called the root) must be the empty system (which is satisfied by all $x \in \B^N$).
For every internal node $v$, one of the outgoing edges is labeled $P_v = 0$ and the other is labeled $P_v = 1$. We have the consistency requirement that if the edge $(v, w)$ is labeled by $P_v = b$, then the system $\Phi_w$ is implied by the system $\Phi_v \cup \{P_v = b\}$ (this is the linear system containing equations from $\Phi_v$ \emph{and} the equation $P_v = b$). For the DAG to correctly solve $\cR$, we require that for every sink $w$, every $x \in \B^N$ satisfying $\Phi_w$ must also satisfy $(x, o_w) \in \cR$.

Semantically, each node in an affine DAG corresponds to an affine subspace of $\bF_2^N$ and for an internal node $v$, if its immediate successors are $w_1$ and $w_2$, then the affine subspace corresponding to node $v$ is contained in the union of the affine subspaces at nodes $w_1$ and $w_2$. It will be convenient to allow inconsistent systems in an affine DAG. Similarly, an empty set will be considered an affine subspace. An inconsistent system implies every system.

A subcube DAG is an affine DAG where we only allow bit queries and each linear equation appearing in the DAG only involves one variable. In other words, instead of a linear system at each node, we have a partial assignment fixing some subset of the variables to constants.

Every $\res(\oplus)$ refutation of a CNF formula $\varphi$ gives an affine DAG solving the search problem $\cR^\varphi$ associated with $\varphi$ \cite{egi24} whose size and depth are no larger than the size and depth, respectively, of the $\res(\oplus)$ refutation. With this in mind, we will mainly discuss affine DAGs from now on.

The following lemma is essentially Lemma 2.3 in \cite[]{egi24}.
\begin{lemma}[\cite{egi24}] \label{lem:path_implies}
Consider nodes $u$ and $v$ in an affine DAG with systems $\Phi_u$ and $\Phi_v$ respectively such that there is some path $p$ from $u$ to $v$. Let $\Psi$ be the system consisting of all equations labeling the edges of $p$. Then $\Phi_u \cup \Psi$ implies $\Phi_v$.
\end{lemma}

We now consider inputs in $(\B^l)^m$. We will make use of affine restrictions in our proof. Most of our affine restrictions will be block-respecting in the sense that each block will be either completely fixed (determined by some affine functions of unfixed blocks) or completely unfixed. 
 Recall that $\cL_{B}$ denotes the collection of all affine functions supported on $B$.
    Let $A \subseteq [m]$. We say that $\rho : \{x_{i, j} \mid i \in A, j \in [l]\} \rightarrow \cL_{[m] \setminus A}$ is an affine restriction fixing $A$.

For one of the proofs, as an intermediate step towards a block-respecting affine restriction, we will also need affine restrictions that are not necessarily block-respecting but simply fix some variables as affine functions of all other variables (which we call free). We write $\rho(x_{i, j}) = *$ if $x_{i, j}$ is not in the domain of $\rho$.
A bit-fixing restriction $\sigma$ is an affine restriction where each of the fixed variables is assigned a bit $0$ or $1$. For affine restrictions $\rho_1$ and $\rho_2$ defined on disjoint sets of variables, $\rho_1 \cup \rho_2$ is the affine restriction whose domain is the union of the domains of $\rho_1$ and $\rho_2$ and it fixes each variable in its domain according to $\rho_1$ or $\rho_2$. This is well-defined since $\rho_1$ and $\rho_2$ fix disjoint sets of variables.

For a parity $P$ on $(\B^l)^m$ and an affine restriction $\rho$ fixing $A \subseteq [m]$, we use $P|_\rho$ to denote the parity obtained by substituting for all $x_{i, j}, i \in A, j \in [l]$ according to $\rho$. This results in a parity in $\cL_{[m] \setminus A}$. Similarly, we use $\Phi|_\rho$ to denote the linear system obtained by substituting in $\Phi$ all $x_{i, j}, i \in A, j \in [l]$ according to $\rho$. This substitution could make the system inconsistent in which case we simply represent the resulting system by $1 = 0$. Otherwise, we tacitly assume that all equations in $\Phi|_\rho$ are linearly independent by removing any equations which are implied by others. Again the exact choice of which redundant equations to remove is not important for what follows. 

Let us give an equivalent way of describing $\Phi|_\rho$ when it is satisfiable. 
We can rewrite $\rho$ as a linear system $\Psi(\rho)$ as follows. For each $x_{i, j}$ fixed by $\rho$, say $\rho(x_{i, j}) = P$ for some parity $P$, we add the equation $x_{i, j} + P = 0$ to $\Psi(\rho)$. Now abusing notation, we will write $\rho$ implies a linear system $\Phi$ to mean $\Psi(\rho)$ implies $\Phi$. An equivalent way to define $\rho$ implies $\Phi$ is if each equation in $\Phi$ becomes equivalent to $0 = 0$ after substituting according to $\rho$.
Similarly, for a linear system $\Phi$, we write $\Phi \cup \rho$ to denote the linear system $\Phi \cup \Psi(\rho)$. (Let us emphasize here that this is the union of the linear systems and so corresponds to the intersection of the corresponding affine subspaces.) 
The collection of all equations implied by $\Phi|_\rho$ is the same as the collection of equations implied by $\Phi \cup \rho$ whose support does not contain any variable from $A$.

If $V \subseteq (\bF_2^l)^m$ is the affine subspace defined by $\Phi$ and $W$ is the set of inputs consistent with $\rho$,
 the affine subspace defined by $\Phi|_\rho$ is obtained by considering the projection of $V \cap W$ onto the blocks not fixed by $\rho$.

For a given affine DAG $\cD$ and affine restriction $\rho$ fixing $A$, we use $\cD|_\rho$ to denote the affine DAG obtained by applying $\rho$ to each parity and linear system appearing in $\cD$. Observe that the consistency condition still holds for the DAG obtained after applying the restriction. This is perhaps easiest to see from the semantic view above of how a restriction affects an affine subspace. 
$\cD|_\rho$ does not mention any variable from $A$. If $\cD$ solves a search problem $\cR \subseteq (\B^l)^m \times \cO$, then $\cD|_\rho$ solves $\cR|_\rho$ where $\cR|_\rho \subseteq (\B^l)^{[m] \setminus A} \times \cO$ is defined by $\cR|_\rho(y) = \cR(x)$ where $x$ is the unique extension of $y$ according to $\rho$. 

\subsection{Bit pigeonhole principle and collision-finding}

The bit pigeonhole principle $\bphp_n^m$ on $m$ pigeons and $n=2^l$ holes ($m > n$) encodes the unsatisfiable statement that $m$ pigeons can be placed into $n$ holes such that no two pigeons are in the same hole. For each pigeon $i \in [m]$, we have variables $x_{i, j}, j \in [l]$ encoding the hole it flies to. For all distinct $i, k \in [m]$ and all $z \in \B^l$, we have a clause encoding $\neg (\llbracket x_i = z \rrbracket \wedge \llbracket x_k = z \rrbracket)$. Each such clause contains $2l$ literals since $\llbracket x_i = z \rrbracket$ is a conjunction of $l$ literals. So $\bphp_n^m$ is a CNF formula containing $\binom{m}{2}n$ clauses of width $2 \log n$.
In the associated false clause search problem, given such an assignment $x \in (\B^l)^n$, the goal is to find distinct $i, k \in [m]$ and $z \in \B^l$ such that $x_i = x_k = z$.

We will consider the closely related search problem of collision-finding $\coll_n^m \subseteq (\B^l)^m \times \binom{[m]}{2}$ where we only need to find distinct $i, k$ such that $x_i = x_k$. It is easy to see that any affine DAG solving the false clause search problem associated with $\bphp_n^m$ also solves the collision-finding problem once we change the output labels suitably.

The $t$-collision-finding problem is defined similarly. In $t$-$\coll_n^m \subseteq (\B^l)^m \times \binom{[m]}{t}$, the goal is to find a set $I \in \binom{[m]}{t}$ such that for all $i, k \in I$, we have $x_i = x_k$. If $m > (t-1)n$, $t$-$\coll_n^m$ is total.

The CNF formula corresponding to $t$-$\coll_n^m$, denoted $t$-$\bphp_n^m$ is defined as follows over the same variables $x_{i, j}\; (i \in [m], j \in [l])$ as for $\bphp_n^m$. For each set $I \subseteq [m]$ of size $t$ and each $z \in \B^l$, there is a clause $\neg (\wedge_{i \in I} \llbracket x_i = z\rrbracket)$. So $t$-$\bphp_n^m$ contains $\binom{m}{t}n$ clauses of width $t \log n$.
\section{Size-depth lower bound for multicollision-finding}
\label{sec:tbphp}
In this section, we prove Theorem \ref{thm:tbphp} (which implies Theorem \ref{thm:wbphp}) giving a size-depth lower bound on $\res(\oplus)$ proofs of $t$-$\bphp$. We will do this by proving the corresponding lower bound for affine DAGs solving the $t$-collision-finding problem. 

Let $n = 2^l$ denote the total number of holes. 
For our proof, we will need to consider a promise version of $t$-collision-finding where each pigeon is promised to only fly into a collection of available/allowed holes $A \subseteq \B^l$. We use $t$-$\coll_{A}^m \subseteq (\B^l)^m \times \binom{[m]}{t}$ to denote this problem. Formally, $t$-$\coll_{A}^m = \{(x, I) \mid x_i = x_k; \; \forall i, k \in I \} \cup \{(x, I) \mid x_k \notin A \text{ for some } k \in [m]\}$. The second set is to simply allow all outputs for any input violating the promise.


\subsection{PDT lower bound for multicollision-finding}
\label{subsec:pdt_tcoll}

In this subsection, we prove a distributional lower bound (of a structured kind) for deterministic PDTs solving $t$-$\coll_A^m$, where $A \subseteq \B^l$ is a set of available holes.
Let $\mu$ be the uniform distribution on $A^m$.
Let $T$ be a deterministic parity decision tree on $(\B^l)^m$ of depth $d$.

We describe a procedure, Algorithm \ref{alg:sim_multicoll},  which will let us estimate a lower bound on the probability that $T$ has not solved $t$-$\coll_A^m$ when run on the distribution $\mu$. We call it a simulation since the procedure is similar to recent simulations for randomized parity decision trees \cite{pod25, bi25, bess24}, primarily relying on ideas from \cite{bi25}.

\begin{algorithm} \label{alg:sim_multicoll}
\caption{$t$-$\coll_A^m$ simulation on uniform distribution}

\KwIn{ $A \subseteq \B^l$, PDT $T$}
$C \gets []$ \tcp*{List of holes used during simulation}
$F \gets [m]$ \tcp*{Free blocks} 
$L \gets \emptyset$ \tcp*{Collection of equations}
$v \gets$ root of $T$\;
$iter \gets 0$ \;

\While{$v$ is not a leaf}{
$P' \gets$ query at $v$ \;
$P \gets P'$ after substituting according to $L$\; \label{algline:cleanup}
\tcp{$P$ now depends only on blocks in $F$ and is equivalent to $P'$ under $L$}
\eIf{$P$ is a constant, $b \in \bF_2$}{ \label{algline:trivial}
Update $v$ according to $b$\;
}
{ \label{algline:nontrivial}
$(i, j) \gets \min\{(i, j) \in F \times [l] \mid x_{i, j }\text{ appears in } P\}\}$ \; 
Pick $y$ uniformly at random from $A$\;
$L \gets L \cup \{x_{i, h} = y_h \mid h \in [l]\setminus \{j\}\}$\;
\eIf{$y^{\oplus j} \notin A$}{ \label{algline:notinA}
$L \gets L \cup \{x_{i, j} = y_j\}$\;
Append $\{y\}$ to $C$\;
}
{\label{algline:notinPj}
Pick $b \in \bF_2$ uniformly at random\;
$L \gets L \cup \{P = b\}$\;
Append $\{y, y^{\oplus j}\}$ to $C$\;
Update $v$ according to $b$\;
}
$F \gets F \setminus \{i\}$\;
}

$iter \gets iter+1$\; \label{algline:iter}
}
\eIf{there exist distinct $i_1, i_2, \dots, i_t$ such that $\cap_{k \in [t]} C_{i_k}  \neq \emptyset$}{
\KwRet FAIL \tcp*{There is a potential $t$-collision}
}
{
\KwRet $v, C, L, F$
}
\end{algorithm}

Let us make some clarifying remarks about Algorithm \ref{alg:sim_multicoll}. We will show below that at the beginning of each iteration, after applying suitable row operations, $L$ can be viewed as fixing all the variables in the fixed blocks $[m] \setminus F$ as affine functions of the free blocks $F$. With this in mind, in Line \ref{algline:cleanup}, we substitute for all fixed variables appearing in $P'$ according to $L$. A more formal description of this can be found in the row reduction subroutine described in \cite{cmss23, beame2023}, though there are some minor differences between what is done in \cite{cmss23} and here. 

Also note that since in each iteration of Algorithm \ref{alg:sim_multicoll}, $|F|$ decreases or we move down the parity decision tree, the simulation always terminates.

\begin{lemma} \label{lem:inv}
    The simulation maintains the following invariants in the beginning of each iteration of the while loop and at termination:
    \begin{enumerate}
        \item The collection of equations $L$ uniquely determines $x_i\; (i \in [m] \setminus F)$ as affine functions of $x_{i'} \;(i' \in F)$. 
        \item For each assignment to all blocks in $F$, if we assign values to $x_i \; (i \in [m] \setminus F)$ according to $L$, we have $x_i \in A$ for each  $i \in [m] \setminus F$.
        \item $L$ implies all the parity constraints on the path from the root to the current node $v$.    
    \end{enumerate}
\end{lemma}
\begin{proof}
    The proof is by induction on the number of iterations, $iter$. At the beginning of the first iteration, $iter = 0$ and all the statements are seen to be trivially true.

    So suppose the statement holds at the beginning of the iteration when $iter = i$, $i \geq 0$. We want to show that all the conditions also hold at the end of the iteration (at which point we have $iter = i+1$). 

    In the case that the condition in Line \ref{algline:trivial} is satisfied, we do not change $L$ and the parity constraint labeling the edge traversed in this iteration is implied by the condition in Line \ref{algline:trivial}. So all the invariants are satisfied at the end of the iteration.

    Next let us consider the case where the else clause on Line \ref{algline:nontrivial} is executed.
    The case $y^{\oplus j} \notin A$ is clear since in this case we explicitly set $x_{i, j} = y_j$  ensuring that $x_i = y \in A$. We stay at the same node in this case.

    In the other case $y^{\oplus j} \in A$, we know that no matter what bit we set $x_{i, j}$ to we will have $x_i \in A$. In this case, we fix $x_{i, j}$ as determined by the constraint $P = b$. Since we ensured that under $L$, the constraint $P=b$ is equivalent to the original parity constraint at the edge just crossed, we have preserved the invariant that each parity constraint on the path from the root to the current node is implied by $L$.

    In all these cases, to see the first point, note that $L$ has a block upper triangular structure if we rearrange the columns so that all the fixed blocks appear in the order in which they were fixed before the free blocks.
\end{proof}

\begin{lemma} \label{lem:rpdt_dist} Let $W(v)$ be the event that node $v$ of $T$ is visited by Algorithm \ref{alg:sim_multicoll}. Let $V(v)$ be the event that for a random $x \sim \mu$, running $T$ on $x$ reaches $v$.
Then for every $v \in T$, we have $\Pr[W(v)] = \Pr[V(v)]$.
\end{lemma}
\begin{proof}
    Suppose at the beginning of an iteration, we have a linear system $L$ which fixes some subset $[m] \setminus F$ of blocks as affine functions of the blocks $F$ such that conditioned on $x \sim \mu$ satisfying $L$, the distribution on $x$ projected to the free blocks $F$ is uniform on $A^{F}$. We will show that the simulation adds equations to $L$ with the correct probability in this iteration and the resulting system $L$ also satisfies the above condition. The claim then follows by induction and Lemma \ref{lem:inv}.

    In the case that the parity we are trying to simulate is already determined by $L$, $L$ stays unchanged and so the distribution on the free blocks remains the uniform distribution on $A^F$.

Consider the case where the parity query is not determined by $L$. Let $P$ denote the equivalent parity query under $L$ which only depends on the blocks in $F$.
    Fix $i \in F, j \in [l]$ occurring in $P$. We condition on all $x_{i, h}, h \neq j$ (Lines 15-16). If these uniquely determine $x_{i, j}$ also according to $A$, then we fix $x_{i, j}$ to the unique possible value. This happens in the if clause at Line \ref{algline:notinA}. Since block $i$ is independent of the other blocks, the blocks in $F \setminus \{i\}$ continue to be distributed according to the uniform distribution on $A^{F \setminus \{i\}}$.

    In the case that $x_{i, j}$ is not uniquely determined by $x_{i, h}, h \notin [l]$, $x_{i, j}$ is a uniform random bit. Since the parity $P$ (which only depends on blocks in $F$) contains $x_{i, j}$, this parity is independent of $x_{i', h'}$ for $i' \in F \setminus \{i\} $ and is uniformly distributed. The independence implies that conditioned on $P = b$ for any $b$, the distribution on all blocks in $F \setminus \{i\}$ remains the uniform distribution on $A^{F \setminus \{i\}}$.
\end{proof}

We will use the following lemma to estimate the probability that the simulation succeeds (which corresponds to the PDT not having found a $t$-collision). 

\begin{lemma} \label{lem:multi_balls_prob}
    Let $G = (V = X \cup Y, E)$ be a bipartite graph on $r$ vertices where $r \geq 1$. Let $t \geq 2$ be an integer. Consider the following random process for throwing balls into $r$ bins. Each bin corresponds to a vertex of $G$. In each iteration, choose a matching $M$ in $G$ (depending on past choices and throws), where the matching does not need to be perfect but we always think of it as a spanning subgraph of $G$. Pick a vertex $v \in V$ uniformly at random. Throw a ball into each vertex in the connected component of $M$ containing $v$. 

    Consider $d$ iterations of the above process where $d \leq tr/e^6$. For every strategy of picking matchings in this process, the probability that no bin contains at least $t$ balls is at least $(1-1/e)\exp(-cd(3e^2d/tr)^{t-1})$ for some constant $c$.
\end{lemma}
\begin{proof}
    We consider the case $t = 2$ separately since it is a simple variant of the birthday paradox. In any iteration, conditioned on the events of the previous iterations, the probability that any particular bin receives a ball is at most $2/r$ since each component of a matching has size at most $2$. This implies that in the $i^{th}$ iteration, conditioned on the previous iterations, the probability that a ball is sent into a non-empty bin is at most $2(i-1) \cdot 2/r = 4(i-1)/r$ since at most $2(i-1)$ balls could have been thrown in the first $i-1$ iterations. Therefore the probability that no bin has at least $2$ balls at the end is at least
    \begin{align*}
        \prod_{i = 1}^d \left(1 - \frac{4(i-1)}{r}\right) 
        &\geq \prod_{i = 1}^d \exp\left(-2 \cdot \frac{4(i-1)}{r}\right) 
        \geq \exp\left(-\frac{8}{r}\binom{d}{2}\right) 
        \geq \exp\left(-\frac{4d^2}{r}\right) .
    \end{align*}
    The first inequality holds if $\frac{4(d-1)}{r} \leq 0.75$ which is implied by the assumption $d \leq tr/e^6$. This finishes the case $t = 2$.

    Suppose $t \geq 3$. We first consider the case $d \geq tr^{1-\frac{1}{t-1}}/3e$.
    Let $E$ denote the event we are interested in: no bin contains at least $t$ balls at the end. Let $B$ denote the bad event that there are at least $2D-1$ bins such that each of them contains at least $t-1$ balls, where $D$ is a positive integer to be fixed later. By the union bound, $\Pr[E] \geq \Pr[E \vee B] - \Pr[B]$ and so it suffices to lower bound $\Pr[E \vee B]$ and upper bound $\Pr[B]$.

    For $i \in [d]$, let $E_i$ denote the event that at the end of the $i^{th}$ iteration, no bin contains at least $t$ balls and similarly let $B_i$ denote the event corresponding to $B$ at the end of the $i^{th}$ iteration. Our goal is to estimate $\Pr[E \vee B] = \Pr[E_d \vee B_d]$. The event $E_1$ holds with probability $1$. We claim that for all $i \geq 2$, $\Pr[E_i \vee B_i] \geq (1 - 2(2D-2)/r)\Pr[E_{i-1} \vee B_{i-1}]$.

    Condition on a configuration of balls in bins after the $(i-1)^{th}$ iteration which satisfies $E_{i-1} \vee B_{i-1}$. Fix any matching $M$ in $G$ and throw balls according to $M$. There are two cases to consider:
    \begin{itemize}
        \item If $B_{i-1}$ holds, $B_i$ holds with probability $1$ since the number of balls in each bin never decreases.
        \item If $B_{i-1}$ does not hold, then $E_{i-1}$ holds. Since there are at most $2D-2$ bins with $t-1$ balls, the probability that we land into any of these bins (thereby creating a bin with $t$ balls) is at most $(2D-2) \cdot 2/r$, where we used that for each fixed bin, the probability that it receives a ball in this iteration is at most $2/r$. So the probability that $E_i$ holds in this case is at least $1- 2(2D-2)/r$.
    \end{itemize}
    So we have $\Pr[E_i \vee B_i] \geq (1 - 2(2D-2)/r)\Pr[E_{i-1} \vee B_{i-1}]$.
    By induction, this implies $\Pr[E \vee B] \geq (1-2(2D-2)/r)^d$. If $4(D-1)/r \leq 0.75$ (we will verify this later after picking $D$), then $1-4(D-1)/r \geq \exp(-2 \cdot 4(D-1)/r) \geq \exp(-8D/r)$. So $\Pr[E \vee B] \geq \exp(-8dD/r)$

    We now upper bound $\Pr[B]$. Let $B^X$ denote the event that at least $D$ bins in $X$ have at least $t-1$ balls. Define $B^Y$ analogously. By the union bound, $\Pr[B] \leq \Pr[B^X] + \Pr[B^Y]$.

    To upper bound $\Pr[B^X]$, note that the process restricted to bins in $X$ has the following properties. Since there are no edges in $X$, in each iteration, at most one bin in $X$ receives a ball. Moreover, the probability that in any given iteration, any fixed bin receives a ball is at most $2/r$ even conditioned on the balls thrown in previous iterations. So we can estimate $\Pr[B^X]$ by the union bound as follows. There are $\binom{|X|}{D}$ subsets of $X$ of size $D$. For a bin to have received at least $t-1$ balls, there must be some size $t-1$ subset of the $d$ iterations where the ball went into that bin. So for any collection of $D$ bins, there are at most $\binom{d}{t-1}^D$ such partial histories indicating for each of the $D$ bins, $t-1$ iterations where the ball went into that bin. Finally, since in each iteration, at most one ball is thrown and this ball has probability at most $2/r$ of going in any particular bin, any such partial history occurs with probability at most $(2/r)^{(t-1)D}$.
    Hence, by the union bound, we get
    \begin{align*}
        \Pr[B^X] \leq \binom{|X|}{D} \binom{d}{t-1}^D \left(\frac{2}{r}\right)^{(t-1)D}.
    \end{align*}
    By combining this with the analogous bound for $\Pr[B^Y]$, we obtain
    \begin{align*}
        \Pr[B] &\leq \binom{|X|}{D} \binom{d}{t-1}^D \left(\frac{2}{r}\right)^{(t-1)D} + \binom{|Y|}{D} \binom{d}{t-1}^D \left(\frac{2}{r}\right)^{(t-1)D} \\
        &\leq \binom{r}{D} \binom{d}{t-1}^D \left(\frac{2}{r}\right)^{(t-1)D} \tag{$\binom{a}{k}+\binom{b}{k} \leq \binom{a+b}{k}$} \\
        &\leq \left(\frac{er}{D}\right)^D \left(\frac{ed}{t-1}\right)^{(t-1)D}\left(\frac{2}{r}\right)^{(t-1)D} \tag{$\binom{n}{k} \leq \left(\frac{en}{k}\right)^k$}\\
        &\leq \left(\frac{er}{D}\left(\frac{2ed}{(t-1)r}\right)^{t-1}\right)^D \\
        &\leq \left(\frac{er}{D}\left(\frac{3ed}{tr}\right)^{t-1}\right)^D \tag{$t \geq 3$}
    \end{align*}
    
    Set $D = \ceil{e^3 r (3e^2d/tr)^{t-1}}$ so that the term above becomes sufficiently small. Plugging this into the above bound, we obtain
    \begin{align*}
        \Pr[B] &\leq \left(\frac{er}{\ceil{e^3 r (3e^2d/tr)^{t-1}}}\left(\frac{3ed}{tr}\right)^{t-1}\right)^D
        \leq \exp(-(t+1) D) \leq \exp(-(8d/r+1)D).
    \end{align*}
    The last inequality uses the assumption $d \leq tr/e^6$.

    Finally,
    \begin{align*}
        \Pr[E] &\geq \Pr[E \vee B] - \Pr[B] \\
        &\geq \exp(-8dD/r) - \exp(-(8d/r+1)D) \\        
        &\geq (1-\exp(-D))\exp(-8dD/r) \\
        &\geq (1-1/e)\exp\left(-\frac{8d}{r} \cdot 2e^3 r  \left(\frac{3e^2d}{tr}\right)^{t-1}\right) \\
        &\geq (1-1/e)\exp\left(-16e^3d \left(\frac{3e^2d}{tr}\right)^{t-1}\right).
    \end{align*}
    In the inequality where we substituted for $D$, we used that $D = \ceil{e^3 r (3e^2d/tr)^{t-1}} \leq 2e^3 r(3e^2d/tr)^{t-1}$ which follows from $d \geq tr^{1-\frac{1}{t-1}}/3e$.

    The only thing left to verify is that this choice of $D$ satisfies $4(D-1)/r \leq 0.75$ which was used to lower bound $\Pr[E \vee B]$.
    \begin{align*}
        \frac{4(D-1)}{r} \leq \frac{4}{r}\cdot e^3 r\left(\frac{3e^2d}{tr}\right)^{t-1} \leq 4e^3 \left(\frac{3e^2}{e^6}\right)^{t-1} \leq 4e^3 (3e^{-4})^2 \leq 0.75.
    \end{align*}
    The second inequality above uses $d \leq tr/e^6$ and the third inequality uses $t \geq 3$ and $3e^{-4} < 1$. This finishes the proof of the case where $d \geq tr^{1-\frac{1}{t-1}}/3e$.

    The other case $d < tr^{1-\frac{1}{t-1}}/3e$ follows from a simple union bound. For a fixed bin $v$, the probability that it contains at least $t$ balls at the end is at most $\binom{d}{t}\left(\frac{2}{r}\right)^t \leq \left(\frac{2ed}{tr}\right)^t$ where $\binom{d}{t}$ counts the choices for $t$ iterations in which $v$ received a ball and $\frac{2}{r}$ is an upper bound on the probability that $v$ receives a ball in a given iteration conditioned on previous iterations. So the probability that no bin has at least $t$ balls is at least
    \begin{align*}
        1-r\left(\frac{2ed}{tr}\right)^t 
        \geq \exp\left(-2r\left(\frac{2ed}{tr}\right)^t\right) \geq \exp\left(-\frac{4ed}{t}\left(\frac{2ed}{tr}\right)^{t-1}\right) 
        \geq \exp\left(-\frac{4ed}{3}\left(\frac{3ed}{tr}\right)^{t-1}\right).
    \end{align*}
    In the last inequality, we used $t \geq 3$. For the first inequality, we need $r\left(\frac{2ed}{tr}\right)^t \leq 0.75$ which holds since $d < tr^{1-\frac{1}{t-1}}/3e$ as verified below
    \[
    r\left(\frac{2ed}{tr}\right)^t < r\left(\frac{2r^{-\frac{1}{t-1}}}{3}\right)^t \leq \left(\frac{2}{3}\right)^t r^{-\frac{1}{t-1}} \leq \left(\frac{2}{3}\right)^3 \leq 0.75. \qedhere
    \]
\end{proof}

\begin{lemma} \label{lem:multicoll_fail_prob}
There exist constants $c_1, \delta > 0$ such that the following holds. Let $T$ be a PDT of depth $d$. If $|A| \geq 2n/3$ and $4 \leq d \leq \delta tn$, then the probability that Algorithm \ref{alg:sim_multicoll} when run on $T$ does not return FAIL is at least $(1-2/e)\exp(-c_1d(18e^2d/tn)^{t-1})$. 
\end{lemma}
\begin{proof}
    We need to lower bound the probability that there is no $t$-collision among the sets in $C$ when the simulation terminates. We will give a lower bound on the probability that there is no collision in $C$ and $iter \leq 4d$. It suffices to give an upper bound on $\Pr[iter > 4d]$ and a lower bound on the probability that there is no $t$-collision in $C$ in the first $4d$ iterations.

    To bound the probability that $iter > 4d$, we will use that in each of the first $4d$ iterations (if we have not already reached a leaf), after conditioning on the history, the probability that Line \ref{algline:trivial} or \ref{algline:notinPj} is executed is at least $1/2$. Whenever either of these lines is executed, we move down the PDT. Since the depth of the PDT is $d$, the only way we could not have reached a leaf after $4d$ iterations is if these lines were executed fewer than $d$ times during the first $4d$ iterations.

    Define for each $i \in [4d]$, the following random variables $X_i$ and $Y_i$ where $X_i$ encodes what happens to the state of the simulation ($L, F, v$ ) during the $i^{th}$ iteration and $Y_i$ is a binary random variable which is $1$ if the simulation ended before the $i^{th}$ iteration, or Line \ref{algline:trivial} or \ref{algline:notinPj} is executed and satisfied in the $i^{th}$ iteration. Note that $Y_i$ is determined by $X_1, \dots, X_i$.

    Conditioned on $X_1, X_2, \dots, X_{i-1}$ ($i \leq 4d$), we know (deterministically) that one of the following cases happens. 
    \begin{enumerate}
        \item There is no $i^{th}$ iteration
        \item The if clause following Line \ref{algline:trivial} is executed in the $i^{th}$ iteration
        \item The else clause Line \ref{algline:nontrivial} is executed in the $i^{th}$ iteration
    \end{enumerate}
    In the first two cases, we have $Y_i = 1$ by definition.
    
    So let us verify in the third case that $\Pr[Y_i = 1 \mid X_1, X_2, \dots, X_{i-1}] \geq 1/2$. Note that together $X_1, X_2, \dots, X_{i-1}$ determine $F, L$ in the beginning of the $i^{th}$ iteration. The probability that $y^{\oplus j} \notin A$ is at most $(n-|A|)/|A|$ since each such $y$ with $y^{\oplus j} \notin A$ corresponds to a unique used hole in $[n] \setminus A$. Since we assumed $|A| \geq 2n/3$, we get that this probability is at most $1/2$.     
    Thus we have $y^{\oplus j} \in A$ with probability at least $1/2$ in which case Line \ref{algline:notinPj} is executed.
    
    This implies $\E[\sum_{i = 1}^{4d} Y_i] \geq 2d$. We now use the Chernoff bound which applies to such  $Y_i$'s (see, for instance, \cite[Lemma 17.3]{mitzenmacher2017probability}) to conclude 
    \[
    \Pr[\sum_{i = 1}^{4d} Y_i < d ] \leq \exp(-2d^2/4d) = \exp(-d/2).
    \]

    We next give a lower bound on the probability that no $t$-collision is created in $C$ in the first $4d$ iterations. To do this, we will invoke Lemma \ref{lem:multi_balls_prob} with $|A|$ bins and $4d$ iterations. The graph $G$ will be the subgraph of the hypercube $\B^l$ induced by $A$. The condition $d \leq t|A|/e^6$ will be ensured by picking $\delta$ to be sufficiently small since $|A| \geq 2n/3$ and $d \leq \delta tn$.

    We now explain how to obtain a strategy for picking matchings in the random process of Lemma \ref{lem:multi_balls_prob}. We first see that the simulation gives a randomized strategy for picking matchings in the following way. If in an iteration, Line \ref{algline:trivial} is executed, then we do not throw any balls in that iteration. If Line \ref{algline:nontrivial} is executed, then we pick the matching corresponding to the $j^{th}$ direction in the subgraph of $\B^l$ induced by $A$. The set appended to $C$ corresponds to the collection of bins that received balls in this iteration and it is clear that the distribution of which bins receive balls in this iteration is the same as the distribution of the set appended to $C$. If two bins receive balls (corresponding to Line \ref{algline:notinPj}), then we flip a uniform random coin corresponding to $b$ in the simulation and branch accordingly (this is the only place where randomness is used by the strategy). We follow the simulation for the first $4d$ iterations in this way. If the simulation ends before $4d$ iterations, we do not throw any more balls in the random process. The probability that some bin receives at least $t$ balls is the probability that there is a $t$-collision in $C$ within the first $4d$ iterations of the simulation.

    By averaging, there is some way to fix the random choices in the above strategy to get a deterministic strategy for which the probability of not having an $t$-collision is at most the corresponding probability for the randomized strategy. Now Lemma \ref{lem:multi_balls_prob} implies that the probability that no bin contains at least $t$ balls is at least 
    \[(1-1/e)\exp(-c(4d)(3e^2\cdot 4d/t|A|)^{t-1}) \geq (1-1/e)\exp(-4cd(18e^2d/tn)^{t-1})\]
    where the inequality uses the assumption $|A| \geq 2n/3$. (Strictly speaking, the deterministic strategy above may have fewer than $4d$ iterations, but it is clear that this can only increase the probability of not having a $t$-collision.) Hence $ \Pr[\text{no } t\text{-collision in }C \text{ in first } 4d \text{ iterations}] \geq (1-1/e)\exp(-4cd(18e^2d/tn)^{t-1})$.
    
By the union bound, 
\begin{align*}
    &\Pr[\text{no $t$-collision in }C \text{ and } iter \leq 4d] \\
    & \geq \Pr[\text{no $t$-collision in }C \text{ in first } 4d \text{ iterations}]
    - \Pr[iter \geq 4d] \\
    & \geq (1-1/e)\exp(-4cd(18e^2d/tn)^{t-1}) - \exp(-d/2) \\
    & \geq (1-2/e)\exp(-4cd(18e^2d/tn)^{t-1})
\end{align*}
where the last inequality holds if we ensure that $d/2 \geq 4cd(18e^2d/tn)^{t-1} + 1$. This holds if $d/4 \geq 1$ and $d/4 \geq 4cd(18e^2d/tn)^{t-1}$. The condition $d/4 \geq 1$ holds by assumption. For the other condition, it suffices to ensure $1/4 \geq 4c(18e^2d/tn)^{t-1}$. For $\delta$ sufficiently small, specifically $\delta \leq \min\{1/18e^2, 1/400e^2c\}$, this holds
\[
4c(18e^2d/tn)^{t-1} \leq 4c (18e^2\delta)^{t-1} \leq 4c (18e^2 \delta) \leq 1/4.
\]
The first inequality used the assumption $d \leq \delta tn$. For the second inequality, we used $18e^2 \delta \leq 1$ and $t \geq 2$.
\end{proof}

\begin{lemma} \label{lem:not_sink}
    Let $v, C, L, F$ be returned by a successful run of Algorithm \ref{alg:sim_multicoll}. Suppose $|A| \geq 3$. For every set $I \subseteq [m]$ of size $t$, there exists $x \in A^m$ such that $x_{i_1} \neq x_{i_2}$ for some $i_1, i_2 \in I$ and $x$ reaches the leaf $v$.
\end{lemma}
\begin{proof}
    By Lemma \ref{lem:inv}, $L$ implies all the parity constraints on the path from the root to $v$. So it is sufficient to find $x \in A^m$ satisfying $L$ for which $x_{i_1} \neq x_{i_2}$ for some $i_1, i_2 \in I$. We consider cases depending on how many blocks in $I$ are free.

    \begin{enumerate}
        \item If $I$ contains at least two blocks in $F$, say $i_1$ and $i_2$, set $x_{i_1}$ and $x_{i_2}$ to distinct strings from $A$, set all other blocks in $F$ to arbitrary strings in $A$ and fix blocks outside $F$ according to $L$. By Lemma \ref{lem:inv}, such a string $x$ is unique and lies in $A^m$.

        \item If $I$ contains only one block in $F$, say $i_1$, we pick $i_2 \in I$ to be some fixed block. Consider the possible strings that $x_{i_2}$ can be according to $L$. Since $i_2 \notin F$, $L$ fixes at least $l-1$ bits of $x_i$. So there are at most two such strings and there is some other string in $A$ since $|A| \geq 3$. Set $x_{i_1}$ to such a string. Set all other blocks in $F$ to arbitrary strings in $A$ and fix blocks outside $F$ according to $L$.

        \item If no string in $I$ lies in $F$, we set blocks in $F$ to arbitrary strings in $A$ and fix blocks outside $F$ according to $L$. Since this was a successful run, we must have $x_{i_1} \neq x_{i_2}$ for some $i_1, i_2 \in I$. If this was not the case, then the corresponding sets in $C$ during the simulation would have non-empty intersection and the simulation would have failed. \qedhere
    \end{enumerate}
\end{proof}

\subsection{Affine DAGs for multicollision-finding}
\label{subsec:adags_tcoll}

The following simple lemma will let us upper bound the rank of a linear system in terms of the probability that it is satisfied under the distribution which is uniform on $A^m$. A slightly weaker statement can alternatively be obtained by using a related lemma for balanced distributions from \cite{bhat2024}.

\begin{lemma} \label{lem:rank_prob}
    Let $\Psi$ be a linear system on $(\B^l)^m$ whose rank is $r$. Let $A \subseteq \B^l$ be such that $|A| \geq 2n/3$. Let $\mu$ be the uniform distribution on $A^m$.
     Then 
    \[
    \Pr_{x \sim \mu}[x \text{ satisfies } \Psi] \leq \left(\frac{3}{4}\right)^{r}.
    \]
\end{lemma}
\begin{proof}
    Fix a collection of $r$ linearly independent columns of the coefficient matrix of $\Psi$. Let $B$ be the set of blocks containing the corresponding variables. We have $b\coloneqq|B| \leq r$. Condition on all blocks outside $B$. For any possible assignment to all blocks outside $B$, the number of solutions for the resulting system $\Psi'$ whose rank is $r$ and depends on $b$ blocks ($bl$ variables) is $\frac{n^{b}}{2^r}$. In particular, the number of such solutions from $A^b$ is at most $\frac{n^{b}}{2^r}$. By assumption, $|A|^b \geq (2n/3)^b$. So the conditional probability is at most $\frac{n^{b}}{2^r}/(2n/3)^b = 1.5^b /2^r \leq (3/4)^r$. Since this holds for all possible assignments to the blocks outside $B$, we have the bound $\Pr_{x \sim \mu}[x \text{ satisfies } \Psi] \leq \left(\frac{3}{4}\right)^{r}$ as desired. 
\end{proof}

The next lemma will be used between random walks to obtain an affine restriction which fixes few blocks to distinct strings and implies the linear system at the node reached at the end of the random walk.

\begin{lemma} \label{lem:aff_restrict}
    Let $v, C, L, F$ be returned by a successful run of Algorithm \ref{alg:sim_multicoll} when run on $T$ and $A \subseteq \B^l$. Let $\Psi$ be a linear system which is implied by $L$. Let $U = \B^l \setminus A$.
    Let $r$ be the rank of $\Psi$. If $|U| + 2r \leq n/2$, then there exists an affine restriction $\rho$ fixing blocks $[m]\setminus F'$ and satisfying the following conditions:
    \begin{enumerate}
        \item The number of blocks fixed by $\rho$, say $s$, is at most $r$.
        \item $\Psi$ is implied by $\rho$. 
        \item There exists a set $U' \subseteq A$ such that $|U'| \leq 2s$ 
        and for any assignment to $F'$, if we set the blocks $[m]\setminus F'$ according to $\rho$, all the strings assigned to $[m]\setminus F'$ lie in $U'$ and there is no $t$-collision among the blocks in $[m] \setminus F'$.
    \end{enumerate}
\end{lemma}
\begin{proof}
Consider the blocks in $\cl(\Psi) \cap F$. Let $\hat{U}$ be the set of all possible holes used by the blocks in $\cl(\Psi) \setminus F$ during this run of the simulation. Since we assign at most two holes to each fixed block during the simulation, $|\hat{U}| \leq 2|\cl(\Psi) \setminus F|$. Assign distinct strings from $\B^l \setminus (U \cup \hat{U})$ to the blocks in $\cl(\Psi) \cap F$. This can be done if $|\cl(\Psi) \cap F| \leq n - (|U| + 2|\cl(\Psi) \setminus F|)$. This holds since $|U| + 2|\cl(\Psi) \setminus F| + |\cl(\Psi) \cap F| \leq |U| + 2(|\cl(\Psi) \setminus F| + |\cl(\Psi) \cap F|) \leq |U| + 2r \leq n/2$ where the second inequality used Lemma \ref{lem:size_bound}. 

Now assign all blocks in $F \setminus \cl(\Psi)$ arbitrarily and extend to a full assignment $x$ by fixing all blocks in $[m] \setminus F$ according to $L$. By Lemma \ref{lem:inv}, there is a unique such extension which also ensures that all fixed blocks lie in $A$. Since $L$ implies $\Psi$, $x$ also satisfies $\Psi$.

This means that the partial assignment to blocks in $\cl(\Psi)$ which agrees with $x$ satisfies all equations in $\langle \Psi \rangle$ whose support is contained in $\cl(\Psi)$. Let $\sigma_1$ denote this bit-fixing restriction which sets blocks in $\cl(\Psi)$ according to $x$. Since we are considering a successful run of the algorithm, there is no $t$-collision among the blocks in $\cl(\Psi) \setminus F$, which are assigned strings from $\hat{U}$. Also note that $\hat{U}$ does not contain any strings from $U$ by Lemma \ref{lem:inv}. Moreover, we ensured above that each block in $\cl(\Psi) \cap F$ is not assigned a string from $U \cup \hat{U}$. So $\sigma_1$ fixes blocks in $\cl(\Psi)$ to strings in $A$ such that there is no $t$-collision among them. Let $U_1$ denote the set of strings assigned to $\cl(\Psi)$ by $\sigma_1$.

Now consider $\Psi' \coloneqq \Psi|_{\sigma_1}$. Since $\sigma_1$ fixes blocks in $\cl(\Psi)$ to bits, $\Psi'$ is safe. By Theorem \ref{lem:safe}, there exist $\rk(\Psi')$ variables lying in distinct blocks such that the corresponding columns in $\Psi'$ are linearly independent. Let $X$ denote the set of these variables and $B$ the set of blocks containing $X$. We will now define a restriction $\sigma_2$ which fixes all variables in blocks $B$ except those in $X$ to bits. We wish to ensure that irrespective of how the bits in $X$ are set, $\sigma_2$ guarantees that all blocks in $B$ are assigned distinct strings from $A \setminus U_1$.

We do this inductively, considering all blocks in $B$ in any order and for each block, picking an assignment to its $l-1$ variables that do not belong to $X$ in such a way that the two possible strings do not agree with any of the previously assigned strings. This can be done as long as $n/2 > |U| + |\cl(\Psi)| + 2(\rk(\Psi') - 1)$. To see this, note that the left side denotes the number of assignments to $l-1$ bits. The right hand side denotes the maximum number of forbidden assignments: there are $|U|$ forbidden holes from before, at most $|\cl(\Psi)|$ holes assigned to blocks in the closure by $\sigma_1$ and $2$ forbidden assignments for each block in $B$ which we have already fixed. This condition is satisfied since we have $|U| + |\cl(\Psi)| + 2(\rk(\Psi') - 1) < |U| + 2(|\cl(\Psi)| + \rk(\Psi')) \leq  |U| + 2r \leq n/2$ where the second inequality used Lemma \ref{lem:size_bound}. So such a restriction $\sigma_2$ exists. Let $U_2 \subseteq A \setminus U_1$ denote the strings assigned to the blocks $B$ by $\sigma_2$. Since we assign two strings to each block in $B$ and these are different for different blocks in $B$, we have $|U_2| = 2|B|$.

Finally consider the system $\Psi'|_{\sigma_2}$. Since the restriction $\sigma_2$ only fixed variables in blocks $B$ outside $X$ to bits, $X$ is still a collection of $\rk(\Psi')$ variables whose corresponding columns are linearly independent. So we can solve for each variable in $X$ to express it as an affine function depending only on variables in blocks outside $\cl(\Psi) \cup B$. This gives us the desired affine restriction $\rho$ once we combine with the restrictions $\sigma_1$ and $\sigma_2$.

Let us verify that this affine restriction satisfies the desired properties. The number of blocks fixed is $s = |\cl(\Psi)| + \dim(L(\Psi)[\setminus \cl(\Psi)]) \leq r$ by Lemma \ref{lem:size_bound}. 

To see that $\rho$ 
implies $\Psi$, first note that all equations supported on the closure of $\Psi$ are satisfied by the restriction $\sigma_1$ above and $\rho$ is an extension of $\sigma_1$. Next we need to check that $\Psi |_{\sigma_1}$ is implied by $\rho$. This was ensured since $\Psi|_{\sigma_1}$ is implied by $\sigma_2$ and $\Psi'|_{\sigma_2}$ above which are contained in $\rho$.

For the third point, we set $U' = U_1 \cup U_2$. 
So  $|U| = |U_1| + |U_2| \leq |\cl(\Psi)| + 2\rk(\Psi') \leq 2|\cl(\Psi)| + 2\rk(\Psi') = 2s$. 
We also ensured above that no matter how the bits in $X$ are assigned, there is no $t$-collision among the blocks $\cl(\Psi) \cup B$.
\end{proof}

Next we prove our main lemma which performs a random walk to find a node far from the root where the DAG has not made much progress.

\begin{lemma} \label{lem:dag_restrict}
    There exist constants $c, \delta > 0$ such that the following holds. Let $A \subseteq \B^l$ with $|A| \geq 2n/3$. Suppose there is an affine DAG $C$ solving $t$-$\coll_A^m$ whose depth is at most $D$ and size is at most $S$. Suppose $\log S \leq \delta n$. Set $d \coloneqq \floor{c tn^{1-1/t}(\log S)^{1/t}}$. There exist $m' \geq m - O(\log S)$ and $A' \subseteq A$ with $|A'| \geq |A| - O(\log S)$ such that there exists an affine DAG $C'$ solving $t$-$\coll_{A'}^{m'}$ whose depth is at most $D - d$ and size is at most $S$.
\end{lemma}
\begin{proof}
    Let $T$ be the depth $d$ PDT obtained by starting at the root of the DAG $C$, repeating nodes at depth at most $d$ if required and removing any nodes beyond depth $d$. 
    Run Algorithm \ref{alg:sim_multicoll} on $T$ and $A$. By Lemma \ref{lem:multicoll_fail_prob}, it succeeds with probability $p$ at least $S^{-O(1)}$ if $c$ is sufficiently small so that $t(18e^2c)^{t-1}$ is bounded from above by a constant. The assumptions on $d$ for Lemma \ref{lem:multicoll_fail_prob} are satisfied as explained next. 
    For $\delta \leq 1$, the assumption $\log S \leq \delta n$ implies $(\log S)^{1/t} \leq n^{1/t}$ which combined with $c$ being sufficiently small implies the required upper bound on $d$. The assumption $\log S \leq \delta n$ also implies $n \geq 1/\delta$ which in turn implies $d \geq 4$ if $\delta$ is sufficiently small. Since the affine DAG has size at most $S$, there must be a node $w$ in the DAG such that the leaves in $T$ corresponding to $w$ are successfully reached with probability  at least $p/S \geq S^{-O(1)}$. By Lemma \ref{lem:rpdt_dist}, this is also a lower bound on the probability that when $x$ is picked uniformly at random from $A^m$ and we follow the path taken by $x$ in the DAG $C$ starting at the root, the path reaches the node $w$.

    Let $\Phi$ denote the linear system at $w$ in the DAG $C$. By Lemma \ref{lem:rank_prob} and the $S^{-O(1)}$ lower bound on the probability of $\Phi$ being satisfied by a random $x \in A^m$, we get that the rank $r$ of $\Phi$ is at most $O(\log S)$. Now fix a successful run of the simulation which ends at a leaf $v$ in $T$ corresponding to $w$. Let $L$ be the linear system returned by this run. By Lemma \ref{lem:inv}, we have that $L$ implies all parity constraints on some path $p$ from the root to $w$. By Lemma \ref{lem:path_implies}, the system of parity constraints on path $p$ implies $\Phi$. Combining these, we get that $L$ implies $\Phi$. 
    Apply Lemma \ref{lem:aff_restrict} to this successful run and the linear system $\Phi$ to obtain $\rho$, $s$ and $U' \subseteq A$. Here $\rho$ is an affine restriction fixing $s$ blocks as affine functions of the other $m - s$ blocks. To apply Lemma \ref{lem:aff_restrict}, note that the condition $n-|A| + r \leq n/2$ holds  since $|A| \geq 2n/3$ and we showed above $r \leq O(\log S)$ which can be made at most $n/6$ by picking $\delta$ to be small enough.

    Set $m' = m - s$ and $A' = A \setminus U'$. Consider $C|_\rho$ which solves $t$-$\coll_A^m|_\rho$. Note that since $\rho$ 
    implies $\Phi$ (Lemma \ref{lem:aff_restrict}), $\Phi|_\rho$ is the empty system. So the node $w$ in $C|_\rho$ is now labeled by the empty system and we can consider the affine DAG $C'$ consisting only of nodes reachable from $w$. $C'$ still solves $t$-$\coll_A^m|_\rho$. 
    
    We claim that $C'$ solves $t$-$\coll_{A'}^{m'}$ after some minor modifications. We first modify $C'$ so that it solves $t$-$\coll_{A'}^{m'}$ when the input blocks are indexed using $F'$ instead of $[m']$. Fix any subset $J \subseteq F'$ of size $t$. For any sink node in $C'$ whose output label is not contained in $F'$, we replace it by $J$. 

    We now verify that for every sink in $C'$, the output label is correct for every input satisfying the linear system at the sink. We only need to check this for inputs in $(A')^{F'}$ since for other inputs, all outputs are considered correct. We consider cases on the original output label $I$: 
    \begin{enumerate}
        \item If $I \subseteq F'$, then the label stays unchanged. Consider any $x \in (A')^{F'}$ satisfying the linear system labeling this node. Then since $A' \subseteq A$, $x$'s extension $y$ according to $\rho$ satisfies $y \in A^m$ and $I$ must be a $t$-collision for $y$ since $C'$ solves $t$-$\coll_A^m|_\rho$. This implies $I$ is a $t$-collision for $x$ since for each $i \in I \subseteq F'$, $x_{i} = y_i$. So the output label is correct for all $x$ at this node.
        \item If $1 \leq |I \cap F'| < t$, the label is changed to $J$. Let $i_1 \in I \cap F'$ and $i_2 \in I \setminus F'$. Consider any $x \in (A')^{F'}$ satisfying the linear system at this node. Its extension $y$ according to $\rho$ lies in $A^m$. So $I$ is a $t$-collision for $y$ and in particular, $y_{i_1} = y_{i_2}$. However, $y_{i_2} \in U'$ by Lemma \ref{lem:aff_restrict} and $y_{i_1} = x_{i_1} \in A'$ which is a contradiction since $A' = A \setminus U'$. So there is no $x \in (A')^{F'}$ which satisfies the linear system at this node.
        \item If $I \cap F' = \emptyset$, the label is changed to $J$. By Lemma \ref{lem:aff_restrict}, for every $x \in (\B^l)^{F'}$, the extension $y$ does not have a $t$-collision outside $F'$. This means that the linear system at this node must be inconsistent after applying the restriction $\rho$.
    \end{enumerate}
    So $C'$ correctly solves $t$-$\coll_{A'}^{m'}$ on the index set $F'$. Finally we relabel according to any bijection between $F'$ and $[m']$ to obtain a DAG solving $t$-$\coll_{A'}^{m'}$.
    
    It is clear that the DAG $C'$ has size at most $S$. If the  DAG $C'$ has depth more than $D - d$, then the original DAG $C$ must have depth more than $D$. To see this, take a successful run of the algorithm ending at a leaf $v$ in $T$ where $v$ corresponds to $w$ in $C$. Observe that $v$ must be at depth $d$ since otherwise $w$ would be a sink in $C$ which is not possible by Lemma \ref{lem:not_sink}. So there is a path $p$ in $C$ from the source to $w$ of length $d$ corresponding to the root to $v$ path in $T$. If there is a path in $C'$ of length more than $D - d$, there is a corresponding path in $C$ which we can combine with $p$ to obtain a path of length more than $D$, which would be a contradiction.
\end{proof}

We now make repeated use of the above lemma to prove our result.

\begin{theorem}[Theorem \ref{thm:tbphp} rephrased]
For all $t \coloneqq t(n) \geq 2$, all $m > (t-1)n$, if there exists an affine DAG $C$ which solves $t$-$\coll_n^m$ of depth $D$ and size $S$, then $D (\log S)^{1-1/t} = \Omega(tn^{2-1/t})$.
\end{theorem}
\begin{proof}
We first consider the case $\log S > \delta n$  where $\delta$ is the constant in Lemma \ref{lem:dag_restrict}. In this case, we can combine this bound with a general $\Omega(tn)$ lower bound on the depth of any affine DAG for $t$-$\coll_n^m$. Such a lower bound can be proved by extending any of the proofs for $t = 2$ \cite{egi24, bi25}. Alternatively, we can use the simulation from the previous subsection. Run Algorithm \ref{alg:sim_multicoll} on any depth $d$ PDT $T$ and $A = \B^l$, where $d = \Theta(tn)$ is the largest that is allowed by the statement of Lemma \ref{lem:multicoll_fail_prob}. Then combining Lemmas \ref{lem:multicoll_fail_prob} and \ref{lem:not_sink} implies that such a PDT must fail to solve $t$-$\coll_n^m$ on some input.

Now consider the other case $\log S \leq \delta n$.
Set $A = \B^l$. We repeatedly invoke Lemma \ref{lem:dag_restrict} with $A$, $m$ and $C$, updating  $A$, $m$ and $C$ according to $A'$, $m'$, $C'$ guaranteed by the statement. We can do this as long as $|A| \geq 2n/3$. The other conditions required by the lemma continue to hold by the conclusion of the lemma statement. In each iteration, $|A|$ decreases by at most $O(\log S)$. So we can use Lemma \ref{lem:aff_restrict} at least $(n/3)/O(\log S) = \Omega(n/\log S)$ many times. This finally gives a DAG of depth at most $D - \Omega(nd/\log S)$ where $d = \floor{c tn^{1-1/t}(\log S)^{1/t}}$. Since depth must be nonnegative, we have $D \geq \Omega(nd/\log S) = \Omega(tn^{2-1/t}/(\log S)^{1-1/t})$ as desired.
\end{proof}

\section{Size-depth lower bound for collision-finding when 
$m = n+o(n)$}
\label{sec:sbphp}

In this section, we prove the size-depth lower bound for affine DAGs solving $\coll_n^m$ when $m = n + o(n)$ (which would imply Theorem \ref{thm:s_bphp})
 and the associated parity decision tree lower bound (Theorem \ref{thm:rpdt_coll}). 

For a parity $P = \sum_{i \in [m], j \in [l]} c_{i, j} x_{i, j} + b$ and $h \in [l]$, we will use $\proj(P, h)$ to denote the linear form $\sum_{i \in [m]} c_{i, h} x_{i, h}$. We will use $|\proj(P, h)|$ to denote the number of variables appearing in $\proj(P, h)$, i.e.~$|\proj(P, h)|$ is the number of $i \in [m]$ such that $c_{i, h} = 1$.


\subsection{PDT lower bound for collision-finding}

Let $\cS \subseteq \B^l$ be a \emph{multiset} containing $m$ strings such that each string appears at most twice in $\cS$. Our goal in this subsection is to prove a distributional lower bound for PDTs that solve collision-finding when the input is promised to be a permutation of the multiset $\cS$. We use $\coll_\cS^m$ to denote this problem. 

The hard distribution $\mu$ will be the uniform distribution on permutations of $\cS$. For the analysis, it will be convenient to refine $\cS$ to obtain a set $S'$ in the following way. We assign a color (red $\tr$ or blue $\tb$) to each element of $\cS$ as follows. If a string $y$ appears only once in $\cS$, then we add $(y, \tb)$ to $S'$. If $y$ appears twice in $\cS$, then we add both $(y, \tb)$ and $(y, \tr)$ to $S'$. This ensures $|S'| = |\cS|$. Now if we define $\mu'$ to be the uniform distribution on the permutations of $S'$, then the marginal distribution on  the first components of the $m$ blocks is identical to $\mu$. So it suffices to analyze $\mu'$ instead of $\mu$. 

If $A \subseteq S' \subseteq \B^l \times \{\tr, \tb\}$, we may refer to $x \in (\B^l)^{m'}$ (where $m' = |A|$) as a permutation of $A$ when we really mean that $x$ is a permutation of the multiset obtained by ignoring the colors of the elements in $A$.

Let $T$ be a deterministic parity decision tree on $(\B^l)^m$.
We describe a simulation for $T$, Algorithm \ref{alg:2coll_sim},  which will let us give a lower bound on the probability that $T$ has not found a collision when run on the distribution $\mu$.

\begin{algorithm} \label{alg:2coll_sim}
\caption{$\coll_\cS^m$ simulation on $\mu$}

\KwIn{ $\cS \subseteq \B^l$, PDT $T$}
Set $S'$ to be the refinement of $\cS$ as described above\;
$A = S'$ \tcp*{Available holes}
$col$ is a length-$m$ array with each entry set to \texttt{none} \tcp*{Colors assigned to blocks}
$d \gets$ depth of $T$ \;
$F \gets [m]$ \tcp*{Free blocks} 
$L \gets \emptyset$ \tcp*{List of equations}
$L_F \gets \{\sum_{i \in [m]} x_{i, j} = \sum_{(y, c) \in A} y_j \mid j \in [l]\}$ \tcp*{Permutation equations}
$v \gets$ root of $T$\;
$iter \gets 0$ \;

\While{$v$ is not a leaf}{
$P' \gets$ query at $v$ \;
$P \gets P'$ after substituting according to $L$\; 
\tcp{$P$ now depends only on blocks in $F$ and is equivalent to $P'$ under $L$}
\eIf{$P$ is determined by $L_F$}{ \label{algline:coll_trivial}
Update $v$ using the constraints in $L_F$\;
}
{ \label{algline:coll_nontrivial}
$j \gets \min\{h \in [l] \mid \proj(P, h) \notin \{0, \sum_{i \in F} x_{i, h} \}\}$ \; \label{algline:pick_j}
\If{$|\proj(P, j)| > |F|/2$}{
$P \gets P + \sum_{i \in F} x_{i, j} + \sum_{(y, c) \in A} y_j$\; \label{algline:proj_P_smaller}
}
$i \gets \min\{i \in F \mid x_{i, j} \text{ appears in } P\}$ \;
Pick $(y, c)$ uniformly at random from $A$\;
\eIf{$c = \tr$ or $(y^{\oplus j}, \tb) \notin A$}{ \label{algline:coll_notinA}
$L \gets L \cup \{x_{i, h} = y_h \mid h \in [l]\}$\;
$col_i \gets c$\;
$A \gets A \setminus \{(y, c)\}$;
$F \gets F \setminus \{i\}$;
}
{ \label{algline:coll_blue_pair}
Pick $g \in F\setminus\{i\}$ uniformly at random\;
$L \gets L \cup \{x_{i, h} = y_h, x_{g, h} = y_h  \mid h \in [l]\setminus\{j\}\} \cup \{x_{i, j}+x_{g, j} = 1\}$\;
$A \gets A \setminus \{(y, \tb), (y^{\oplus j}, \tb)\}$;
$F \gets F \setminus \{i, g\}$\;
$col_i \gets \tb$; $col_g \gets \tb$\;
\eIf{$x_{g, j}$ does not appear in $\proj(P, j)$}{ \label{algline:coll_notinPj}
Pick $b \in \bF_2$ uniformly at random\;
$L \gets L \cup \{P = b\}$\;
Update $v$ to the appropriate child according to $b$\;
}{
Pick $b \in \bF_2$ uniformly at random\;
$L \gets L \cup \{x_{i, j} = b\}$\;
}
}
$L_F \gets \{\sum_{i \in F} x_{i, j} = \sum_{(y, c) \in A} y_j \mid j \in [l]\}$;
}
$iter \gets iter+1$\; \label{algline:coll_iter}
}
\eIf{$S' \setminus A$ contains both $(y, \tb)$ and $(y, \tr)$ for some $y \in \B^l$ or $iter > 8d$}{
\KwRet FAIL \tcp*{Found collision or too many iterations}
}
{
\KwRet $v, A, L, F, L_F$
}
\end{algorithm}


\begin{lemma} \label{lem:coll_inv}
    The simulation maintains the following invariants in the beginning of each iteration of the while loop and at termination:
    \begin{enumerate}
        \item The collection of equations $L$ uniquely determines $x_i, i \in [m] \setminus F$ as affine functions of $x_i, i \in F$. 
        \item For each assignment to $x_{i'} \; (i' \in F)$, if we assign values to $x_i\; (i \in [m] \setminus F)$ according to $L$, the blocks $x_i \; (i \in [m] \setminus F)$ form a permutation of $S' \setminus A$. In particular, if we assign a permutation of $A$ to $x_{i'} \; (i' \in F)$, then by setting all other blocks according to $L$, $x$ will be a permutation of $S'$. 
        \item $L$ and $L_F$ together imply all the parity constraints on the path from the root to the current node.        
        \item $|A| \geq |S'| -2\;iter$.
    \end{enumerate}
\end{lemma}
\begin{proof} The proof is similar to that of Lemma \ref{lem:inv}, by induction on the number of iterations, $iter$. 
At the beginning of the first iteration, $iter = 0$ and all the statements are true.

    For the induction step, assume the statement holds at the beginning of the iteration when $iter = i$, $i \geq 0$. We will show that all the conditions also hold at the end of the iteration. 

    When the condition in Line \ref{algline:coll_trivial} is satisfied, we do not change $L, L_F$, $A$. The parity constraint labeling the edge traversed in this iteration is implied by the condition in Line \ref{algline:coll_trivial}. So all the invariants are satisfied at the end of the iteration.

    Now suppose the else clause on Line \ref{algline:coll_nontrivial} is executed. By the condition that $P$ is not determined by $L_F$, the set considered in Line \ref{algline:pick_j} is nonempty. The updated parity in Line \ref{algline:proj_P_smaller} is equivalent to the original $P$ since $\sum_{i \in F} x_{i, j} = \sum_{(y, c) \in A} y_j$ is in $L_F$.

    The case $c = \tr$ or $(y^{\oplus j}, \tb) \notin A$ is clear as we explicitly set $x_i = y$ which decreases the size of $A$ by $1$. We do not move down the PDT in this case.

    In the other case, we have $c = \tb$ and $(y^{\oplus j}, \tb) \in A$. We now assign the pair $y$ and $y^{\oplus j}$ to $x_i$ and $x_g$ but we do not initially fix which is which. Note that this decreases the size of $A$ by $2$. The condition $x_{i} \neq x_g$ is enforced by $x_{i, j} + x_{g, j} = 1$. If $g$ lies inside $\proj(P, j)$, then we explicitly sample $x_{i, j}$, thereby completely determining $x_i$ and $x_g$ similar to the previous case ($c = \tr$ or $(y^{\oplus j}, \tb) \notin A$).

    If $x_{g, j}$ does not appear in $\proj(P, j)$, then $x_{i, j}$ is determined by the equation $P = b$ which is added to $L$. In this case, we update the node according to $b$. Since under $L$ and $L_F$, the constraint $P=b$ is equivalent to the original parity constraint at the edge just crossed, we have ensured that each parity constraint on the path from the root to the current node is implied by $L$ and $L_F$.

    For the first point, note that $L$ is of a block upper triangular form if we rearrange the columns so that all the fixed blocks appear first and in the order that they were fixed. Moreover, the fixed bits in these blocks together with constraints of the form $x_{i, j} + x_{g, j} = 1$ ensure that for any assignment to the free blocks (not necessarily a permutation), if we assign values to $x_i \; (i \in [m] \setminus F)$ according to $L$, the blocks $x_i \; (i \in [m] \setminus F)$ form a permutation of $S' \setminus A$.
\end{proof}

\begin{lemma} \label{lem:coll_rpdt_dist} Let $W(v)$ be the event that node $v$ of $T$ is visited by Algorithm \ref{alg:2coll_sim}. Let $V(v)$ be the event that for a random $x \sim \mu$, running $T$ on $x$ reaches $v$.
Then for every $v \in T$, we have $\Pr[W(v)] = \Pr[V(v)]$.
\end{lemma}
\begin{proof}
By Lemma \ref{lem:inv}, a node is visited only if the system $L$ combined with $L_F$ implies all the constraints on the path leading to that node. So it suffices to show that the linear systems $L$ and $L_F$ are modified with the correct probabilities. To do this, we will prove the invariant that at the beginning of each iteration, conditioned on $(x, col) \sim \mu'$ satisfying $L$, $L_F$ and the revealed entries of $col$ during the simulation, the blocks of $(x, col)$ in $F$ form a uniform random permutation of $A$. This clearly holds in the beginning since $L$ is empty, all entries of $col$ are $\texttt{none}$ and each permutation $(x, col)$ of $A$ satisfies $L_F$.

So suppose at the beginning of an iteration, the linear system $L$ fixes blocks $[m] \setminus F$ and the only entries of $col$ which are not $\texttt{none}$ are those in $[m]\setminus F$. Moreover conditioned on $(x, col) \sim \mu'$ satisfying $L$ and the revealed entries of $col$, $(x, col)$ projected to the free blocks $F$ is a uniform random permutation of $A$. We will show that the simulation updates $L, L_F, col$ with the correct probability in this iteration and the resulting systems $L, L_F$ and array $col$ also satisfy the above condition. 


If the parity we are trying to simulate is already determined by $L$ and $L_F$, then $L, L_F, col$ stay unchanged and so the distribution on the free blocks remains the uniform distribution on permutations of $A$.

Now consider the case where Line \ref{algline:coll_nontrivial} is executed.
    Fix $i \in F, j \in [l]$. We condition on the color $col_i$ and all bits of $x_i$ except $x_{i, j}$ according to $A$.     
    If the sampled color $col_i = \tr$, or, $col_i = \tb$ and the revealed bits of $x_i$ uniquely determine $x_{i, j}$ according to $A$, then we sample $x_{i, j}$ according to the resulting conditional distribution. It is not hard to see that this is equivalent to how in Algorithm \ref{alg:2coll_sim}, we sample $(y, c) \sim A$ completely but we set $x_{i,j} = y_j$ only when $c = \tr$ or $(y^{\oplus j}, \tb) \notin A$.
    
    The if clause at Line \ref{algline:coll_notinA} handles the case where $col_i = \tr$, or, $col_i = \tb$ and the revealed bits of $x_i$ uniquely determine $x_{i, j}$ according to $A$. In this case, we have fixed $(x_i, col_i)$ to $(y, c)$ by modifying $L$ and $col$ appropriately, and so the resulting conditional distribution on all the other blocks in $F$ is uniform on all $A \setminus \{(y, c)\}$. So $A$ and $F$ are updated appropriately during the simulation.

    If $col_i = \tb$ and conditioned on all bits of $x_i$ except $x_{i, j}$, the bit $x_{i, j}$ is not uniquely determined according to $A$, then $x_{i, j}$ is a uniform random bit.  Condition on the unique $g \in F \setminus \{i\}$ such that $x_{i, h} = x_{g, h}$ for all $h \neq j$ and $col_g = \tb$. Note that since $x$ must be a permutation, we necessarily have $x_{i, j} = x_{g, j} +1$. Now $x_{i, j}$ is a uniform random bit which is independent of all blocks $x_{i'}$ with $i' \in F \setminus \{i, g\}$. 
    
    So if a parity $P$ (which we can assume only depend on blocks in $F$) contains $x_{i, j}$ but not $x_{g, j}$, then this parity is uniformly distributed even if we were to condition on all $x_{i', h'}$ for $i' \in F \setminus \{i, g\} $. This means that this parity is independent of all $x_{i', h'}$ for $i' \in F \setminus \{i, g\} $. The independence implies that conditioned on $P = b$ for any $b$, the distribution on all blocks in $F \setminus \{i, g\}$ remains the uniform distribution on all permutations of $A\setminus \{(y, \tb), (y^{\oplus j}, \tb)\}$.

    The other case where we simply condition on the value of $x_{i, j}$ is similar to the case where $col_i = \tr$, or, $col_i = \tb$ and the revealed bits of $x_i$ uniquely determine $x_{i, j}$ according to $A$.

    In all cases, note that the equations in $L_F$ are satisfied by all permutations of $A$ at the end of an iteration.
\end{proof}

The following lemma will be used to estimate the probability that the simulation does not find a potential collision.

\begin{lemma} \label{lem:coll_balls_bins_prob}
    There exist constants $\delta, \delta' > 0$ such that the following holds. Let $k$ and $s$ be positive integers such that $k \leq \delta's$.
    Let $G = (V = X \cup Y, E)$ be a bipartite graph on $V = [s+k]$ such that the vertices in $\{s+1, s+2, \dots, s+k\}$ are isolated.   Consider the following random process for throwing balls into $s+k$ bins. Each bin corresponds to a vertex of $G$. For the $i^{th}$ iteration, let $G_i$ denote the induced subgraph of $G$ containing only non-empty bins at the beginning of the $i^{th}$ iteration. Choose a (not necessarily perfect) matching $M$ in $G_i$ (depending on past choices and outcomes), where the matching is interpreted as a spanning subgraph of $G_i$. Pick a bin $v \in G_i$ uniformly at random. Throw a ball into each bin in the connected component of $M$ containing $v$. 

    Consider $d$ iterations of the above process where $d \leq s/4$. For every strategy of picking matchings in this process, the probability that at the end there is no $i \in [k]$ such that bins $i$ and $s+i$ both contain a ball is at least $(1-1/e)\exp(-\delta d^2k/s^2)$.
\end{lemma}
\begin{proof}
For brevity, let us say that there is a collision if there is some $i \in [k]$ such that bins $i$ and $s+i$ both contain a ball.

    First consider the case $d \geq s/e^{10}k$. Let $E$ denote the event that at the end there is no collision. Let $B$ be the bad event that there are at least $2D-1$ indices $i \in [k]$ such that bin $i$ or bin $s+i$ contains a ball, where $D = \ceil{e^{10}dk/s}$. By the union bound, $\Pr[E] \geq \Pr[E \vee B] - \Pr[B]$.
    
    We now lower bound $\Pr[E \vee B]$. For $i \in [d]$, let $E_i$ denote the event that there is no collision at the end of the $i^{th}$ iteration. Similarly define $B_i$ to be the event corresponding to $B$ at the end of the $i^{th}$ iteration. We claim that $\Pr[E_i \vee B_i] \geq (1-(2D-2)\cdot 2/(s+k - 2(i-1)))\Pr[E_{i-1} \vee B_{i-1}]$. Indeed if $B_{i-1}$ holds, then $B_i$ holds with probability $1$. Otherwise there are at most $2D-2$ indices $i \in [k]$ such that bin $i$ or bin $s+i$ contains a ball. So the probability that we throw a ball into one of these at most $2D-2$ indices (thereby creating a collision) is at most $(2D-2)\cdot 2/(s+k-2(i-1))$. Here $2/(s+k-2(i-1))$ is an upper bound on the probability that any particular bin receives a ball in the $i^{th}$ iteration conditioned on the previous iterations. The numerator $2$ comes from the largest size of any component in the matching. The denominator is a lower bound on the number of empty bins after the $(i-1)^{th}$ iteration, since in each iteration at most two bins get a ball.
    
    By induction,
    \begin{align*}
        \Pr[E \vee B] &\geq \prod_{i = 2}^d \left(1 - \frac{4(D-1)}{s+k-2(i-1)}\right) \\
        &\geq \left(1 - \frac{4(D-1)}{s - 2d}\right)^d \\
        &\geq \exp\left(-\frac{8(D-1)d}{s-2d}\right) \\
        &\geq \exp\left(-\frac{16Dd}{s}\right) \tag{$s-2d \geq s/2$} 
    \end{align*}
    The third inequality uses $1-x \geq \exp(-2x)$ if $x \leq 0.75$ which we verify now for $x = \frac{4(D-1)}{s - 2d}$.
    \begin{align*}
        \frac{4(D-1)}{s - 2d} \leq \frac{4}{s-2d} \cdot \frac{e^{10}dk}{s} \leq 2e^{10}\delta' 
    \end{align*}
    Here we used the assumptions $d \leq s/4$ and $k \leq \delta' s$. So the desired bound holds if $\delta'$ is sufficiently small, $\delta' \leq 3/8e^{10}$.

    We now upper bound $\Pr[B]$. Since the vertices $\{s+1, s+2, \dots, s+k\}$ are isolated, we can move each such $s+i$ to the part ($X$ or $Y$) containing $i$ while still ensuring that there are no edges within $X$ and within $Y$. Let $B^X$ denote the event that there at least $D$ indices $i \in [k] \cap X$ such that bin $i$ or bin $s+i$ has a ball. Define $B^Y$ analogously. By the union bound, $\Pr[B] \leq \Pr[B^X] + \Pr[B^Y]$.

    When considering only bins in $X$, note that in any iteration at most one bin receives a ball since there are no edges in $X$. We also know that the probability that a bin receives a ball in the $i^{th}$ iteration ($i \leq d \leq s/4$) is at most $2/(k+s - 2(i-1)) \leq 2/(s-2d) \leq 4/s$ even conditioned on the events of previous iterations. Hence by the union bound, 
    \begin{align*}
        \Pr[B^X] &\leq \binom{|[k] \cap X|}{D}d^D\left(\frac{8}{s}\right)^D
    \end{align*}
    where the binomial coefficient counts which $D$ indices witness the occurrence of $B^X$, and for each such index $i$, there is a factor of $d$ indicating in which iteration, a ball was sent into bin $i$ or $s+i$.
    Combining the above with the analogous bound for $B^Y$ by the union bound, we obtain
    \begin{align*}
        \Pr[B] &\leq \binom{|[k] \cap X|}{D}d^D\left(\frac{8}{s}\right)^D + \binom{|[k] \cap Y|}{D}d^D\left(\frac{8}{s}\right)^D \\
        &\leq \binom{k}{D}d^D\left(\frac{8}{s}\right)^D \\
        &\leq \left(\frac{ek}{D}\right)^D d^D \left(\frac{8}{s}\right)^D \\
        &\leq \left(\frac{8ekd}{sD}\right)^D \\
        &\leq \left(\frac{8ekd}{s\ceil{e^{10}dk/s}}\right)^D \\
        &\leq \left(\frac{8}{e^9}\right)^D \\
        &\leq \left(\frac{1}{\exp(1+16d/s)}\right)^D
    \end{align*}
    where the last inequality follows from $d \leq s/4$.
    Now combining the bounds on $\Pr[E \vee B]$ and $\Pr[B]$,
    \begin{align*}
        \Pr[E] &\geq \exp\left(-\frac{16Dd}{s}\right) - \exp\left(-\left(1+\frac{16d}{s}\right)D\right) \\
        &\geq \left(1- \exp(-D)\right)\exp\left(-\frac{16Dd}{s}\right) \\
        &\geq (1-1/e)\exp\left(-\frac{16d}{s}\cdot  \frac{2e^{10}dk}{s}\right) \\
        &\geq (1-1/e)\exp\left(-\frac{32e^{10}d^2k}{s^2}\right).
    \end{align*}
    In the third inequality, where we substituted for $D$, we used that $D = \ceil{e^{10}dk/s} \leq 2e^{10}dk/s$ which follows from $d \geq s/e^{10}k$.
    This finishes the proof of the case $d \geq  s/e^{10}k$.

    Next consider the case $d < s/e^{10}k$. In this case, we can just union bound over the $k$ possible indices where a collision can occur. For an index $i \in [k]$, for bins $i$ and $s+i$ to have a ball at the end, there must be two distinct iterations where they received balls since $s+i$ is isolated in $G$. So by the union bound the probability of a collision is at most $kd(d-1)\left(\frac{2}{s-2d}\right)^2 \leq 16kd^2/s^2$. So the probability that there is no collision is at least
    \[
    1- 16kd^2/s^2 \geq \exp(-32kd^2/s^2)
    \]
    where to use the inequality $1-x \geq \exp(-2x)$, we need $16kd^2/s^2 \leq 0.75$ which holds by the assumption $d < s/e^{10}k$.
\end{proof}

\begin{lemma} \label{lem:coll_fail_prob}
There exist constants $\delta, \delta'$ such that the following holds. Let $m = |S'|$ and $k$ denote the number of strings which appear twice in $\cS$ (which is the same as the number of red elements in $S'$). 
    If $m - 16d \geq 3n/4$, $d \geq 8$ and $k \leq \delta' n$, then the probability that Algorithm \ref{alg:2coll_sim} succeeds is at least $(1-2/e)\exp(-\delta d^2 k/n^2)$. 
\end{lemma} 
\begin{proof}
We need to lower bound the probability that $iter \leq 8d$ and there is no $y \in \B^l$ such that $S'\setminus A$ contains both $(y, \tb)$ and $(y, \tr)$. Let $E$ denote this event. Set $D = 8d$.  Let $E_D$ denote the event that immediately after $D$ iterations (or at termination, if there were fewer than $D$ iterations), there is no $y \in \B^l$ such that $S'\setminus A$ contains both $(y, \tb)$ and $(y, \tr)$. By the union bound, $\Pr[E_{D}] \leq \Pr[E] + \Pr[iter > D]$. So it suffices to give a lower bound on $\Pr[E_D]$ and an upper bound for $\Pr[iter > D]$.

    To bound the probability that $iter > D$, we will first show that in each of the first $D$ iterations (if we have not already reached a leaf), after conditioning on the history, the probability that Line \ref{algline:coll_trivial} or \ref{algline:coll_notinPj} is executed is at least $1/4$. Whenever either of these lines is executed, we move down the PDT. Since the depth of the PDT is $d$, the only way we could not have reached a leaf after $D$ iterations is if these lines were executed fewer than $d$ times during the first $D$ iterations.

    Formally, define for each $i \in [D]$, the following random variables $X_i$ and $Y_i$ where $X_i$ encodes the state of the simulation ($L, A, F, v$ ) in the $i^{th}$ iteration and $Y_i$ is a binary random variable which is $1$ if the simulation ended before the $i^{th}$ iteration, or in the $i^{th}$ iteration, Line \ref{algline:coll_trivial} or \ref{algline:coll_notinPj} is executed and satisfied. $Y_i$ is determined by $X_1, \dots, X_i$.

    Conditioned on $X_1, X_2, \dots, X_{i-1}$ ($i \leq D$), one of the following cases happens. 
    \begin{enumerate}
        \item There is no $i^{th}$ iteration
        \item The if clause following Line \ref{algline:coll_trivial} is executed in the $i^{th}$ iteration
        \item The else clause Line \ref{algline:coll_nontrivial} is executed in the $i^{th}$ iteration
    \end{enumerate}
    In the first two cases, we have $Y_i = 1$ by definition.
    
    So let us verify in the third case that $\Pr[Y_i = 1 \mid X_1, X_2, \dots, X_{i-1}] \geq 1/4$. 
    Let $m' \coloneqq |A| \geq m - 2D$ by Lemma \ref{lem:coll_inv} and $i \leq D$. 
Let us give an upper bound on the probability that Line \ref{algline:coll_notinA} is executed. Let $k'$ denote the number of red elements in $A$ and $s' = m' - k'$ the number of blue elements. The number of elements $(y, \tb)$ in $A$ such that $(y^{\oplus j}, \tb) \notin A$ is at most the number of $z \in \B^l$ such that $(z, \tb) \notin A$. This is because each such $y$ with $(y^{\oplus j}, \tb) \notin A$ corresponds to a unique blue $z$ missing from $A$ by flipping the $j^{th}$ bit. So the probability that Line \ref{algline:coll_notinA} is executed is at most $(k' + (n - s'))/m'$, since $k'$ is the number of red elements in $A$ and $n-s'$ is the number of $z$ such that $(z, \tb) \notin A$. We can upper bound this quantity as follows.
\begin{align*}
    \frac{k'+(n-s')}{m'} \leq \frac{k' + n - (m' - k')}{m'} \leq \frac{2k'+n}{m'} - 1 \leq \frac{(2\delta'+1)n}{m - 16d} - 1 \leq \frac{9n/8}{m - 16d} - 1 \leq \frac{1}{2}
\end{align*}
The third inequality used $k' \leq k \leq \delta'n$ and $m' \geq m - 16d$. The fourth inequality holds if $\delta'$ is chosen to be small enough. The last inequality holds by the assumption $m - 16d \geq 3n/4$.
    
    So the else clause in Line \ref{algline:coll_blue_pair} is executed with probability at least $1/2$. Conditioned on this event, consider the unique $k \in F \setminus \{i\}$ such that $x_{i, h} = x_{k, h}$ for $h \neq j$ and is colored blue. Then $k$ lies outside $\proj(P, j)$ with probability at least $1/2$ since we ensured that $|\proj(P, j)| \leq m'/2$. Therefore overall with probability at least $1/4$, the Line \ref{algline:coll_notinPj} is executed and the condition is satisfied.

    So $\E[\sum_{i = 1}^{8d} Y_i] \geq 2d$. We now use the Chernoff bound which applies to such $Y_i$'s to conclude 
    \[
    \Pr[\sum_{i = 1}^{8d} Y_i < d ] \leq \exp(-d/4).
    \]

    We now lower bound $\Pr[E_D]$ by invoking Lemma \ref{lem:coll_balls_bins_prob} on the following graph. The graph $G$ consists of the subgraph of the hypercube $\B^l$ induced by the blue strings in $S'$ (these vertices are indexed by $[s]$) along with $k$ isolated vertices corresponding to the red strings in $S'$ (indexed by $\{s+1, s+2, \dots, s+k\}$). The indices for the last $k$ vertices are chosen so that the index corresponding to a red string $(y, \tr)$ is exactly $s$ more than the index corresponding to $(y, \tb)$. The condition on $k$ being sufficiently small with respect to $s$ holds if $\delta'$ is sufficiently small since $k \leq \delta' n$ and $s = m - k \geq 3n/4-k \geq (3/4-\delta')n$.

    We will use the simulation to give a randomized strategy for picking matchings in the random process of Lemma \ref{lem:coll_balls_bins_prob} such that the probability of not having a collision is $\Pr[E_D]$. By averaging, we would get a deterministic strategy for which the probability of not having a collision is at most $\Pr[E_D]$. On the other hand, Lemma \ref{lem:coll_balls_bins_prob} implies that for any strategy, the probability of not having a collision is at least $(1-1/e)\exp(-\delta d^2k/s^2)$.

    The randomized strategy works as follows. If in an iteration, Line \ref{algline:coll_trivial} occurs, we do not pick a matching to throw balls and proceed to the next iteration of the simulation. If in an iteration, Line \ref{algline:coll_nontrivial} occurs, we pick the matching corresponding to flipping the  $j^{th}$ bit in the hypercube. The bin chosen uniformly corresponds to picking $(y, c)$ in the simulation. It is easy to see that each bin has the same probability of receiving a ball as the probability with which the corresponding string is removed from $A$.  Based on the outcome of the throw, proceed with the appropriate part of the simulation. In the case where the component has size $2$, then we make random choices as described in the else clause in Line \ref{algline:coll_blue_pair} and branch accordingly. We continue in this way for the first $D$ iterations of the simulation, stopping early if the simulation terminates. The set $S' \setminus A$ corresponds to all the bins that have received a ball and so the probability of there being a collision in the random matching process is the same as the probability that $S' \setminus A$ contains both $(y, \tb)$ and $(y, \tr)$ for some $y$. As explained earlier, we can now apply Lemma \ref{lem:coll_balls_bins_prob} to a deterministic strategy obtained by averaging to get the desired lower bound on $\Pr[E_D]$.

    Now combining the bounds for $\Pr[E_D]$ and $\Pr[iter > D]$, we get
    \begin{align*}
        \Pr[E] &\geq \Pr[E_D] - \Pr[iter > D] \\
        &\geq (1-1/e)\exp(-\delta d^2k/s^2) - \exp(-d/4) \\
        &\geq (1-2/e)\exp(-\delta d^2 k/s^2) \\
        &\geq (1-2/e)\exp(-\delta''d^2 k/n^2) \tag{$s = m - k \geq (3/4 - \delta')n$}
    \end{align*}
    for some constant $\delta''$.

    For the third inequality above, it suffices to have $d/8 \geq 1$ and $d/8 \geq \delta d^2 k/s^2$. The former holds by assumption. For $d/8 \geq \delta d^2 k/s^2$, note that
    \[
    \frac{\delta d k}{s^2} \leq \frac{\delta \cdot (1/4 + \delta')n/16 \cdot \delta' n}{(3/4 - \delta')^2 n^2} = \frac{\delta (1/4 + \delta')  \delta' }{16(3/4 - \delta')^2 } \leq \frac{1}{8}
    \]
    where the last inequality holds by choosing $\delta'$ to be sufficiently small. The first inequality above uses the bounds $d \leq (m - 3n/4)/16 \leq (n+k - 3n/4)/16 \leq (1/4 + \delta')n/16$, $k \leq \delta' n$ and $s = m - k \geq (3/4 - \delta')n$.
\end{proof}

\begin{lemma} \label{lem:coll_not_sink}
    Let $v, A, L, F, L_F$ be returned by a successful run of Algorithm \ref{alg:2coll_sim} when run on $T$ and $\cS$. Suppose $|A| \geq 3$. For every pair of distinct elements $i_1, i_2 \subseteq [m]$, there exists $x$ which is a permutation of $\cS$ such that $x_{i_1} \neq x_{i_2}$ and $x$ reaches the leaf $v$.
\end{lemma}
\begin{proof}
    By Lemma \ref{lem:inv}, $L$ combined with $L_F$ implies all the parity constraints on the path from the root to $v$. So it is sufficient to find $x$ which is a permutation of $\cS$ and satisfies $L, L_F$ for which $x_{i_1} \neq x_{i_2}$. Moreover, by Lemma \ref{lem:inv}, it suffices to set all the free blocks to a permutation of $A$ since extending this according to $L$ gives a permutation of $\cS$, and $L_F$ holds since it is implied by the free blocks being a permutation of $A$. We consider cases depending on how many blocks in $\{i_1, i_2\}$ are free.

    \begin{enumerate}
        \item If $i_1, i_2 \in F$, let $(y_1, c_1), (y_2, c_2) \in A$ be such that $y_1 \neq y_2$. Such elements of $A$ exist since $|A| \geq 3$ and for each $y \in \B^l$, there are at most two colored versions of $y$ in $A$.
        Set $x_{i_1}$ and $x_{i_2}$ to $y_1$ and $y_2$ respectively. Set all other blocks in $F$ according to an arbitrary permutation of $A \setminus \{(y_1, c_1), (y_2, c_2)\}$ and fix blocks outside $F$ according to $L$. By Lemma \ref{lem:inv}, such a string $x$ is unique and is a permutation of $\cS$.

        \item Suppose $i_1 \in F$ and $i_2 \notin F$ (the case $i_1 \notin F$ and $i_2 \in F$ is analogous). Consider the possible strings that $x_{i_2}$ can be according to $L$. Since $i_2 \notin F$, at least $l-1$ bits of $x_{i_2}$ are fixed by $L$ and so there are at most two such strings. Let $(y, c) \in A$ be such that $y$ is different from the strings that $x_{i_2}$ can be. Such a $(y, c)$ exists since $|A| \geq 3$ and any string that has already been assigned to $x_{i_2}$ can only appear once in $A$ (and must have the opposite color). Set $x_{i_1} = y$. Set all the remaining blocks in $F$ to an arbitrary permutation of $A \setminus \{(y, c)\}$. Fix blocks outside $F$ according to $L$ to obtain a permutation of $\cS$.

        \item If both $i_1, i_2$ do not lie in $F$, we set blocks in $F$ to an arbitrary permutation of $A$ and fix blocks outside $F$ according to $L$. Since this was a successful run, we must have $x_{i_1} \neq x_{i_2}$. If this was not the case, say $x_{i_1} = x_{i_2} = y$, then $S' \setminus A$ would contain both $(y, \tr)$ and $(y, \tb)$ at termination and the simulation would have failed. \qedhere
    \end{enumerate}

\end{proof}

We now have all the ingredients for the randomized PDT lower bound for collision-finding.
\begin{theorem}(Theorem \ref{thm:rpdt_coll} restated) \label{thm:rpdt_coll_again}
    There exists a constant $\delta > 0$ such that the following holds. Let $k \leq \delta n$ and $m \coloneqq n+k$. Let $\eps \in [2^{-k}, 1/2]$. Any zero-error randomized parity decision tree solving $\coll_n^m$ with abort probability $\eps$ must have depth $\Omega(n\sqrt{\log(1/\eps)/k})$.
\end{theorem}
\begin{proof}
We pick $\delta$ to be the constant $\delta'$ from the statement of Lemma \ref{lem:coll_fail_prob}. 
We will prove the lower bound for $\eps \leq (1-2/e)^2$. This also suffices for larger constant abort probability since repeating a constant number of times is enough to bring down the abort probability to any desired constant.

First note that $\log n$ is a lower bound on the depth of any RPDT solving $\coll_n^m$ which aborts with probability $\eps < 1$, since 
any parity certificate indicating $x_{i_1} = x_{i_2}$ must have size at least $\log n$.

If for the given $\eps$, $k$ and $n$, every randomized PDT requires $n/100$ queries to solve $\coll_n^m$ with probability $1-\eps$, then the desired statement holds by the assumption $\eps \geq 2^{-k}$.

So consider the case where there is some randomized PDT of depth $d \leq n/100$ solving $\coll_n^m$ which aborts with probability at most $\eps$. By Yao's principle, this implies that for every distribution, there is a deterministic PDT of depth $d$ solving $\coll_n^m$ which aborts with probability at most $\eps$ on that distribution. Let $\cS$ denote a multiset containing $m = n+k$ strings in $\B^l$ such that each string in $\B^l$ appears at least once and at most twice. Let $\mu$ denote the uniform distribution on permutations of $\cS$.

Let $T$ denote a PDT which solves $\coll_n^m$ while aborting with probability at most $\eps$ on the distribution $\mu$. Run Algorithm \ref{alg:2coll_sim} on $T$ and $\cS$. Lemma \ref{lem:coll_not_sink} implies that whenever the simulation succeeds, the leaf reached must be labeled with $\bot$. Lemma \ref{lem:coll_rpdt_dist} implies that the probability that the simulation ends at a $\bot$ leaf is also the abort probability on an input $x \sim \mu$. So we can use Lemma \ref{lem:coll_fail_prob} to give a lower bound on the abort probability. The required conditions can be seen to be satisfied since $m \geq n$, $d \leq n/100$, $k \leq \delta n$ and $d \geq \log n \geq 8$ for $n$ large enough. Using Lemma \ref{lem:coll_fail_prob}, we obtain
\[
\eps \geq \left(1-2/e\right)\exp(-O( d^2 k/n^2)).
\]
Using $(1-2/e) \geq \sqrt{\eps}$ and then rearranging, we obtain $d \geq \Omega(n\sqrt{\log(1/\eps)/k})$ as desired.
\end{proof}

\subsection{Random permutations satisfying a safe system}

The next lemma gives an upper bound on the probability that a safe system is satisfied by a uniform random permutation of a multiset $\cS \subseteq \B^l$ which does not contain any string more than twice. It will again be convenient to work with the colored version $S'$ of $\cS$ used in the previous subsection.

\begin{lemma} \label{lem:perm_prob}
Let $\Psi$ be a safe linear system on $(\B^l)^m$. Let $\cS \subseteq \B^l$ be a multiset such that $|\cS| = m$ and let $S'$ denote the corresponding set obtained by coloring as above. Suppose $\langle \Psi \rangle$ contains $\langle \{\sum_{i \in [m]} x_{i, j} = \sum_{(y, c) \in S'} y_j \mid j \in [l]\} \rangle$. Let $k$ denote the number of red elements in $S'$. Suppose $k \leq n/20$.
Let the rank of $\Psi$ be $l+r$ ($r \geq 0$). Suppose $m - 2r \geq 3n/4$. Let $\mu$ be the uniform distribution on permutations of $\cS$. Then 
    \[
    \Pr_{x \sim \mu}[x \text{ satisfies } \Psi] \leq \left(\frac{7}{8}\right)^{\ceil{r/4}}.
    \]
\end{lemma}
\begin{proof}
    Instead of directly working with $\mu$, it will be convenient to think of sampling a uniform random permutation of $S'$. Let $\mu'$ be the uniform distribution on permutations of $S'$. Abusing notation, we write $(x, col) \sim \mu'$ to denote such a random permutation of $S'$ where $x \in (\B^l)^m$ is a permutation of $\cS$ and $col \in \{\tb, \tr\}^m$ stores the corresponding colors. 

    The proof is by induction on $r$. If $r = 0$, the statement is trivial. 

    For the induction step, suppose $r \geq 1$.
    Pick an equation $P = b$ in $ \Psi$ such that $P \notin \langle \{\sum_{i \in [m]} x_{i, j} \mid j \in [l]\} \rangle$. Such an equation exists since $\dim(\{\sum_{i \in [m]} x_{i, j} \mid j \in [l]\}) = l$ and $\rk(\Psi) > l$. 
    Let $j \in [l]$ be such that $\proj(P, j) \notin \{0, \sum_{i \in [m]} x_{i, j}\}$. 
    If $|\proj(P, j)| > m/2$, we add $\sum_{i \in [m]} x_{i, j} = \sum_{(y, c) \in S'} y_{j}$ to this equation. This does not change the linear system. 
    After this operation (if required), we now have $|\proj(P, j)| \leq m/2$. Pick some $i \in [m]$ such that $x_{i, j}$ appears in $P$.

We now condition according to $x_{i, h}, h \neq j$ and the color $col_i \in \{\tb, \tr\}$ similar to Algorithm \ref{alg:2coll_sim}. Let $(y, c) \in S'$.
\begin{enumerate}
    \item $c = \tr$ or $(y^{\oplus j}, \tb) \notin S'$.
    
In this case, after conditioning on $x_{i, h} = y_h\; (h \neq j)$ and $col_i = c$, we also condition on $x_{i, j}$.
    Conditioned on $x_i = y$ and $col_i = c$, we know $(x, col) \sim \mu'$ satisfies $\Psi$ if and only if it satisfies $\Psi \cup \{x_{i, j} = y_j \mid j \in [l]\}$. Let $\Psi' = \{Q = b' \in \langle \Psi \cup \{x_{i, h} = y_h \mid h \in [l]\}) \mid i \notin \supp(Q)\} \rangle$. We will show that $\Psi'$ contains a safe system of dimension at least $l+r - 1$.

    Let $L$ and $L'$ denote the sets of linear forms of $\Psi$ and $\Psi'$ respectively. Observe that $L' = L[\setminus \{i\}]$. Indeed since $L'$ is obtained by essentially substituting constants for $x_{i, j}$ in the linear forms of $L$, this is the same as removing the variables $x_{i, j}$ from $L$.
    Let $H = \{Q \in \langle L \rangle \mid \supp(Q) \subseteq \cl(L') \cup \{i\}\}$. 
    Since $L$ is safe, $H$ is safe and $\dim(H) \leq |\cl(L')|+1$. We also have $L'[\setminus \cl(L')] = L[\setminus \cl(L') \cup \{i\}]$. Combining these, we get 
    $|\cl(L')| + \dim(L'[\setminus \cl(L')]) \geq \dim(H) - 1 + \dim(L[\setminus \cl(L') \cup \{i\}]) = \dim(L) - 1$.

    Let $U = \{\sum_{i' \in [m] \setminus \{i\}} x_{i', j} \mid j \in [l]\}$. By Lemma \ref{lem:size_bound}, $\cl(L') \leq r+l \leq m - 3n/4 + l < m$. So $\dim(U[\setminus \cl(L')]) = l$ since all forms are on disjoint sets of variables and no form in $U[\setminus \cl(L')]$ simplifies to $0$ by the bound on $\cl(L')$. By Lemma \ref{lem:sdim_bound} and the lower bound above on $|\cl(L')| + \dim(L'[\setminus \cl(L')])$, there exists a safe subset of $\langle L'\rangle$ whose dimension is at least $r+l-1$ and such that $L'$ contains $U$. Let $\Phi$ be the corresponding subset of equations of $\Psi'$. Let $\tilde{S} = S' \setminus \{(y, c)\}$. 
    So $\langle \Phi \rangle $ contains $\langle \{\sum_{i' \in [m] \setminus \{i\}} x_{i', j} = \sum_{(y', c') \in \tilde{S}} y_i \mid j \in [l]\} \rangle $. Moreover, conditioned on $x_i = y$, $col_i = c$, the remaining blocks form a uniform random permutation $(x', col')$ of $\tilde{S}$. To use the induction hypothesis, we will also require that $m' - 2r' \geq 3n/4$ where $m' \coloneqq |\tilde{S}| = m-1$ and the new system has rank $l+r'$. To ensure this, if required, we remove an equation from the system $\Phi$ so that the rank of $\Phi$ is exactly $l+r -1$ while still ensuring that $\langle \Phi \rangle$ contains $\langle \{\sum_{i' \in [m] \setminus \{i\}} x_{i', j} = \sum_{(y', c') \in \tilde{S}} y_i \mid j \in [l]\} \rangle $. Having done this (if required), we can now use the induction hypothesis to conclude that 
    \[
    \Pr_{(x, col) \sim \mu'}[x \text{ satisfies } \Psi \mid x_i = y, col_i = c] \leq \Pr_{({x'}, col') \sim \tilde{\mu}}[\tilde{x} \text{ satisfies } \Phi] \leq \left(\frac{7}{8}\right)^{\ceil{(r-1)/4}}
    \]
    where $\tilde{\mu}$ is the uniform distribution on $\tilde{S}$.
    
\item $c = \tb$ and $(y^{\oplus j}, \tb) \in A$.

Condition on $col_i = \tb$ and $x_{i, h} = y_h$ for all $h \in [m] \setminus \{j\}$.
Since $(y^{\oplus j}, \tb) \in A$, in any permutation such that $col_i = \tb$ and $x_{i, h} = y_h$ for all $h \in [l] \setminus \{j\}$, there must be a unique $p \in [m]\setminus \{i\}$ such that $x_{p, h} = y_h$ for $h \in [l] \setminus \{j\}$ and $col_p = \tb$. Condition on the value of this $p$ which is uniformly distributed in $[m] \setminus \{i\}$.

Let $\Psi_0$ consist of all equations in the span of $\Psi$ which do not involve the variables $x_{i, j}, x_{p, j}$. Note that the rank of $\Psi_0$ is at least $l+r - 2$ since $\Psi$ has rank $l+r$ and the dimension of the linear forms of $\Psi$ projected down to $x_{i, j}, x_{p, j}$ is at most $2$. Add the equation $\sum_{i' \in [m] } x_{i', j} = \sum_{(y', c) \in S'} y'_j$ to $\Psi_0$. Note that $\Psi_0$ is safe since it is implied by $\Psi$ which is safe. Moreover, note that each equation in $\langle \Psi_0 \rangle $ which involves either $x_{i, j}$ or $x_{p, j}$ must involve both. Let $\Psi' = \{Q = b' \in \langle \Psi_0 \cup \{x_{i, h} = x_{p, h} = y_h \mid h \in [l]\setminus \{j\}\} \cup \{ x_{i, j} + x_{p, j} = 1\} \rangle \mid i, p \notin \supp(Q)\}$. We will show that $\Psi'$ contains a safe system of rank at least $t - 2$ where $t$ is the rank of $\Psi_0$ by arguing exactly as in the previous case.

Let $L$ and $L'$ denote the sets of linear forms of $\Psi_0$ and $\Psi'$ respectively. We have that $L' = L[\setminus\{i, p\}]$. For this, we note as before that linear forms obtained after substituting constants for $x_{i, h}, x_{p, h} \; (h \neq j)$ and removing any resulting constant terms are the same as those obtained by just removing $x_{i, h}, x_{p, h} \; (h \neq j)$. In addition, we are using that every linear form $v$ in $L$ which involves one of $x_{i, j}, x_{p, j}$ must actually use both and so $\Psi'$ must also contain some equation whose linear form is $v+x_{i, j} + x_{p, j}$ (which does not depend on either $x_{i, j}, x_{p, j}$) as these can be cleared from the corresponding equation in $\Psi$ by substituting $x_{i, j} + x_{p, j} = 1$.

Let $H = \{Q \in L \mid \supp(Q) \subseteq \cl(L') \cup \{i, p\}\}$. 
Since $L$ is safe, $H$ is safe which implies $\dim(H) \leq |\cl(L')|+2$. We also have $L'[\setminus \cl(L')] = L[\setminus \cl(L') \cup \{i, p\}]$. Combining these, we get 
    $|\cl(L')| + \dim(L'[\setminus \cl(L')]) \geq \dim(H) - 2 + \dim(L[\setminus \cl(L') \cup \{i, p\}]) = \dim(L) - 2$.
    By arguing as in the previous case, using Lemma \ref{lem:sdim_bound}, there exists a safe subset of $L'$ whose dimension is at least $t-2$ and which contains $\{\sum_{i' \in [m] \setminus \{i, k\}} x_{i', j'} \mid j' \in [l]\setminus \{j\}\}$. Let $\Phi$ be the corresponding subset of equations of $\Psi_0$. Since $t \geq l+r - 2$, we get that $\Phi$ has rank at least $l+r - 4$.
    
    Moreover, conditioned on $\{x_i, x_p\} = \{y, y^{\oplus j}\}$ and $col_i = col_p = \tb$, the remaining blocks form a uniform random permutation $(x', col')$ of $\tilde{S} \coloneqq S' \setminus \{(y, \tb), (y^{\oplus j}, \tb)\}$. To apply the induction hypothesis, we will also require that $m-2 - 2r' \geq 3n/4$ where the new system $\Phi$ has rank $l+r'$. 
    To do this, we remove some equation from $\Phi$ if required so that $\rk(\Phi) \leq \rk(\Psi) - 1$ while ensuring that the system still contains $\{\sum_{i' \in [m] \setminus \{i, k\}} x_{i', j} = \sum_{(y', c') \in \tilde{S}} y_i \mid j \in [l]\}$. 

We now consider cases depending on whether $x_{p, j}$ appears in $\proj(P, j)$. Let $\tilde{\mu}$ be the uniform distribution on $\tilde{S}$.
\begin{itemize}
\item If $x_{p, j}$ appears in $\proj(P, j)$, we just use that the permutation $x'$ on the remaining $m-2$ blocks must satisfy $\Phi$ which is implied by $\Psi$ and $x_{i, h} = x_{k, h} = y_h$ for $h \in [l] \setminus \{j\}$.
\begin{align*}
    &\Pr_{(x, col) \sim \mu'}[x \text{ satisfies } \Psi \mid x_{i, h} = x_{p, h} = y_h, h \in [l]\setminus \{j\}, col_i = col_p = \tb]  \\
    &\leq \Pr_{(x', col') \sim \tilde{\mu}}[x' \text{ satisfies } \Phi] \leq \left(\frac{7}{8}\right)^{\ceil{(r-4)/4}}.
\end{align*}
    
(When $r \leq 3$, we are just bounding the probability by $1$ and not using the induction hypothesis.)

\item If $x_{p, j}$ does not appear in $\proj(P, j)$, we also consider the probability that $P = b$ is satisfied.
\begin{align*}
    &\Pr_{(x, col) \sim \mu'}[x \text{ satisfies } \Psi \mid col_i = col_p = \tb, x_{i, h} = x_{p, h} = y_h, h \in [l]\setminus \{j\}] \\
    &\leq \Pr_{(x, col) \sim \mu'}[x \text{ satisfies } P = b \text{ and } \Phi \mid col_i = col_p = \tb, x_{i, h} = x_{p, h} = y_h, h \in [l]\setminus \{j\}] \\
    &\leq \Pr_{(x, col) \sim \mu'}[x \text{ satisfies } P = b \mid x \text{ satisfies } \Phi, col_i = col_p = \tb, x_{i, h} = x_{p, h} = y_h, h \in [l]\setminus \{j\}] \cdot \\
    &\;\;\;\;\;\;\Pr_{(x, col) \sim \mu'}[x \text{ satisfies } \Phi \mid col_i = col_p = \tb, x_{i, h} = x_{p, h} = y_h, h \in [l]\setminus \{j\}] \\
    &= \frac{1}{2} \Pr_{(x', col') \sim \tilde{\mu}}[x' \text{ satisfies } \Phi] \\
    &\leq \frac{1}{2}\left(\frac{7}{8}\right)^{\ceil{(r-4)/4}}.
\end{align*}
For the equality in the above chain, note that conditioned on $x_i$ and $x_p$ satisfying the above conditions, $x_{i, j}$ is a uniform random bit which is independent of all $x_{i', h'}, i' \notin \{i, p\}, h' \in [l]$. Therefore since the parity $P$ contains $x_{i, j}$ but not $x_{p, j}$, $P$ is uniformly distributed after conditioning on all $x_{i', h'}, i' \notin \{i, k\}, h' \in [l]$.
\end{itemize}
\end{enumerate} 

We now combine the above cases. The only case in which we do better than what follows from induction is when $(x_i, col_i) = (y, \tb)$ is such that  $(y^{\oplus j}, \tb) \in S'$ and for the unique $p$ satisfying $x_{i} = x_p^{\oplus j}$ and $col_p = \tb$, we have $x_{p, j}$ lying outside $P$. The probability that $col_i = \tr$ or $(x_i^{\oplus j}, \tb) \notin A$ is at most $\frac{k+n-(m-k)}{m}$ where $k$ is the number of red strings in $S'$ and $n-(m-k)$ is an upper bound on the the number of strings $z$ such that $(z, \tb) \notin S'$. This can be upper bounded by $(2k+n)/m - 1 \leq (n/8 + n)/(3n/4) - 1 = 1/2$ where we used the assumptions $m - 2r \geq 3n/4$ and $k \leq n/20$. 
Conditioned on $(x_i, col_i) = (y, \tb)$ satisfying $(y^{\oplus j}, \tb) \in S'$, we know that $x_{p, k}$ lies outside $P$ with probability at least $1/2$ since we ensured that $|\proj(P, j)| \leq m/2$. 

We combine the above to obtain the desired bound. All probabilities below are over $(x, col) \sim \mu'$.
\begin{align*}
    &\Pr[x \text{ satisfies } \Psi] \\
    &=\Pr[col_i = \tr \text{ or } (x_i^{\oplus j}, \tb) \notin S'] \cdot \Pr[x \text{ satisfies } \Psi \mid col_i = \tr \text{ or } (x_i^{\oplus j}, \tb) \notin S'] + \\ 
    &\;\;\;\;\; \Pr[col_i = \tb, (x_i^{\oplus j}, \tb) \in S', x_{k, j} \text{ inside } P] \cdot \Pr[x \text{ satisfies } \Psi \mid col_i = \tb, (x_i^{\oplus j}, \tb) \in S', x_{k, j} \text{ inside } P] + \\
    &\;\;\;\;\; \Pr[col_i = \tb, (x_i^{\oplus j}, \tb) \in S', x_{k, j} \text{ outside } P] \cdot \Pr[x \text{ satisfies } \Psi \mid col_i = \tb, (x_i^{\oplus j}, \tb) \in S', x_{k, j} \text{ outside } P] \\
    &\leq \Pr[col_i = \tr \text{ or } (x_i^{\oplus j}, \tb) \notin S'] \cdot \left(\frac{7}{8}\right)^{\ceil{(r-1)/4}} +  \Pr[col_i = \tb, (x_i^{\oplus j}, \tb) \in S', x_{k, j} \text{ inside } P] \cdot \left(\frac{7}{8}\right)^{\ceil{(r-4)/4}} + \\
    &\;\;\;\;\; \Pr[col_i = \tb, (x_i^{\oplus j}, \tb) \in S', x_{k, j} \text{ outside } P] \cdot \frac{1}{2}\left(\frac{7}{8}\right)^{\ceil{(r-4)/4}} \\
    &\leq \bigg(\Pr[col_i = \tr \text{ or } (x_i^{\oplus j}, \tb) \notin S'] + \Pr[col_i = \tb, (x_i^{\oplus j}, \tb) \in S', x_{k, j} \text{ inside } P] + \\
    &\;\;\;\;\; \Pr[col_i = \tb, (x_i^{\oplus j}, \tb) \in S', x_{k, j} \text{ outside } P]\cdot \frac{1}{2} \bigg) \left(\frac{7}{8}\right)^{\ceil{(r-4)/4}} \\
    &\leq \left(1-\Pr[col_i = \tb, (x_i^{\oplus j}, \tb) \in S', x_{k, j} \text{ outside } P]\cdot \frac{1}{2}\right) \left(\frac{7}{8}\right)^{\ceil{(r-4)/4}} \\
    &\leq \left(1 - \frac{1}{2}\cdot \frac{1}{2} \cdot \frac{1}{2}\right) \left(\frac{7}{8}\right)^{\ceil{(r-4)/4}} = \left(\frac{7}{8}\right)^{\ceil{r/4}} \qedhere
\end{align*}
\end{proof}

\subsection{Finding an affine restriction}

To find a suitable affine restriction for restarting the random walk, we will need the following lemma.

\begin{lemma}\label{lem:find_pivots}
Let $\Psi$ be a safe linear system on $(\B^l)^m$. Suppose $\Psi$ contains $\{\sum_{i \in [m]} x_{i, j} = b_j \mid j \in [l]\}$ for some $b_j \in \bF_2 \; (j \in [l])$. Let the rank of $\Psi$ be $l+r$ ($r \geq 0$). Suppose $m > 2r+l$. Then there exist affine restrictions $\rho_1$ and $\rho_2$, and $F, I \subseteq [m]$ such that the following hold.
\begin{enumerate}
\item $|F| = m - (2r+1)$.
\item There exists $G = \{(i_1, j_1), (i'_1, j_1), (i_2, j_2), (i'_2, j_2), \dots, (i_r, j_r), (i'_r, j_r)\} \subseteq ([m]\setminus F) \times [l]$ such that $I \coloneqq \{i_1, i'_1, i_2, i_2', \dots, i_r, i'_r\}$ has size $2r$ where $\rho_1$ fixes variables in $G$ as affine functions of variables in $((F \cup I) \times [l]) \setminus G$.
\item $\rho_1$ implies $x_{i_t, j_t}+ x_{i'_t, j_t} = 1$ for all $t \in [r]$.
\item $\rho_2$ fixes all the variables in the unique block $i_0 \in [m] \setminus (F \cup I)$ as affine functions of variables in $((F \cup I) \times [l]) \setminus G$. More specifically, for each $j \in [l]$, 
\[
\rho_2(x_{i_0, j}) = b_j + \sum_{t \in [r]} (\llbracket j_t \neq j \rrbracket (x_{i_t, j_t} + x_{i'_t, j_t})+\llbracket j_t = j \rrbracket \cdot 1) + \sum_{i \in F} x_{i, j}
\](where $\llbracket j_t = j \rrbracket$ is the $0/1$ indicator for $j_t = j$ interpreted over $\bF_2$).
\item 
$\rho_1 \cup \rho_2$ implies $\Psi$.
\end{enumerate}
\end{lemma}

\begin{proof}
    The proof is by induction on $r$. For $r = 0$, take $\rho_1$ to be the empty restriction and $G = \emptyset$. Let $F = [m] \setminus \{1\}$ and $\rho_2(x_{1, j}) = b + \sum_{i = 2}^m x_{i, j}$ for all $j \in [l]$. It is easy to see that these have the desired properties.

    Now suppose $r \geq 1$. We consider two cases depending on whether $L(\Psi)$ has any important sets, which were defined \cite{ei25}. A non-empty set $S \subseteq [m]$ is said to be important for a collection of linear forms $V$ if $\dim(\{v  \in  \langle V  \rangle \mid \supp(v) \subseteq S\}) \geq |S|$. We will use the notation $V_S$ to denote $\{v \in  \langle V \rangle \mid \supp(v) \subseteq S\}$. Note that if $V$ is safe, then a set $S$ is important for $V$ if $\dim(V_S) = |S|$ since $\langle V \rangle$ does not contain any dangerous sets.
    
    \begin{itemize}
        \item Suppose there is some important set for $L(\Psi)$. 
        
        Let $S \subseteq [m]$ be a minimal important set for $V \coloneqq L(\Psi)$. Let $i_1 \in S, j_1 \in [l]$ be such that $x_{i_1, j_1}$ appears in $V_S$. Let $T \supseteq S$ be a \emph{maximal} important set for $V$. 
        By the definition of $T$ being important, $|T| = \dim(V_{T}) \leq \dim(V) \leq r+l < m$. So $T \neq [m]$. Moreover, since $L(\Psi)$ includes the forms $\{\sum_{i \in [m]} x_{i, j} \mid j \in [l]\}$, we have $\dim(V[\setminus T]) \geq l$ as these $l$ forms are on disjoint sets of variables and none of these $l$ forms simplifies to $0$ after removing blocks in $T$. So the bound on $|T|$ can be improved to $|T| = \dim(V) - \dim(V[\setminus T]) \leq r$.

        Now pick any $i'_1 \in [m] \setminus T$ which exists since $m > r$. Define $V'$ to be $V$ after substituting $x_{i_1', j_1} = x_{i_1, j_1}$. For a linear form $v$, this substitution is equivalent to replacing it by $v + (x_{i_1, j_1} + x_{i_1', j_1})$ if $v$ contains $x_{i_1', j_1}$ and otherwise leaving it unchanged. We claim that $\dim(V'[\setminus \{i_1'\}]) = r+l$ and, moreover, $V'[\setminus \{i_1'\}]$ is safe. We first check that $\dim(V'[\setminus \{i_1'\}]) = r+l$. If this was not the case, then there is some non-zero linear form $v \in \langle V  \rangle$ which lies in the span of $\{x_{i_1, j_1} + x_{i_1', j_1}\} \cup \{x_{i_1', j' }\mid j' \neq j_1\}$. But then $T \cup \{i'\}$ is important for $V$ since $\dim(V_{T \cup \{i'\}}) \geq \dim(V_T \cup \{v\}) = |T| + 1$, where the equality follows from noting that $v$ involves $i'_1 \notin T$ and $T$ being important. So $\dim(V'[\setminus \{i_1'\}]) = r+l$.

        \begin{claim*}
            $V'[\setminus \{i_1'\}]$ is safe.
        \end{claim*}
        \begin{proof}
            Suppose $V'[\setminus \{i_1'\}]$ is not safe. Let $D \subseteq  \langle V[\setminus \{i_1'\}]  \rangle$ be a minimal dangerous set. By minimality of $D$, $|\supp(D)| = \dim(D) - 1$. 

        Let $D^0$ be a collection of $|D|$ linear forms in $V$ corresponding to $D$, so  that $D$ can be obtained from $D^0$ by substituting $x_{i_1', j_1} = x_{i_1, j_1}$ and then setting all $x_{i_1', j'} = 0$ for $j' \neq j_1$. 
        Since $\Psi$ is safe, we have $\dim(D) = \dim(D^0) \leq |\supp(D^0)|$. 
        Moreover, $\supp(D^0) \subseteq \supp(D) \cup \{i_1'\}$ which implies $|\supp(D^0)| \leq |\supp(D)| + 1$. By using $|\supp(D)| = \dim(D) - 1$, these bounds together imply that $\supp(D^0) = \supp(D) \cup \{i_1'\}$ and so $T'\coloneqq \supp(D^0)$ is important for $V$. 
        
        Since $i_1' \in \supp(D^0)$ , we have $T' \nsubseteq T$. We claim that $T \cup T'$ is important for $V$. This would contradict the maximality of $T$. Indeed
        \begin{align*}
            &\dim(V_{T \cup T'}) \geq \dim(V_T \cup V_{T'}) \geq \dim(V_T) + \dim(V_{T'}) - \dim(V_T \cap V_{T'}) \\
            &\geq \dim(V_T) + \dim(V_{T'}) - \dim(V_{T \cap {T'}}) \geq |T| + |T'| - |T \cap T'| = |T \cup T'|
        \end{align*}
        where in the last inequality, we used that $T$ and $T'$ are important for $V$ and $\dim(V_{T \cap T'}) \leq |T \cap T'|$ since $V$ is safe.
        Hence, $V'[\setminus \{i_1'\}]$ is safe.
        \end{proof}
        
         Let $\Psi'$ denote the linear system obtained from $\Psi$ by substituting $x_{i'_1, j_1} = x_{i_1, j_1}+1$ and then moving $x_{i'_1, j'} \; (j' \neq j)$ to the right side, treating them as fixed from now on. In this way, the variables in $\Psi'$ we consider now are only from the blocks $[m] \setminus \{i_1'\}$. Note that $L(\Psi') = V'[\setminus \{i'\}]$ which we showed above has dimension $r+l$ and is safe. In particular, $\Psi'$ as a system on variables from $[m] \setminus \{i_1'\}$ is consistent. Moreover, since $i_1' \notin T$ and $S \subseteq T$, we have  $V_S \subseteq L(\Psi')_S$. Consider a linear equation $v = b$ from $\Psi'_S$ involving $x_{i_1, j_1}$. Suppose $v = x_{i_1, j_1} + w$ for some linear form $w$ (possibly $0$). 
        
        Let $\Psi_1$ be a linear system whose span contains all those equations in the span of $\Psi'$ which do not involve the variable $x_{i_1, j_1}$. This is equivalent to substituting $x_{i_1, j_1} = w + b$ in $\Psi'$ and removing $0=0$ equations. $L(\Psi_1)$ has dimension $r+l-1$ since $x_{i_1, j_1} + w = b$ was an equation in $\langle \Psi' \rangle$. Let $\Psi_2$ be the system obtained from $\Psi_1$ by moving all variables $x_{i_1, j'} (j' \neq j)$ to the right side of the $=$ signs, so that we now think of only the blocks $[m] \setminus \{i_1, i_1'\}$ as variables. We will show that $\Psi_2$ must be consistent, in the sense that there is no equation $v' = b'$ (where $b'$ can involve $\{x_{i_1, j'}, x_{i_1', j'} \mid j' \in [l] \setminus \{j_1\}\}$) in the span of $\Psi_2$ for which the linear form $v'$ is $0$ but $b' \neq 0$. 
        
        Indeed if there was some such equation $0 = b'$, then the corresponding equation in $\Psi_1$ would have $\{i_1\}$ as its support. But we will show that this cannot be the case since $S$ (which includes $i_1$) was chosen to be a minimal important set for $V$ (recall that the span of $V$ also includes the span of $L(\Psi_1)$). If $S \neq \{i_1\}$, this is an immediate contradiction to the minimality of $S$. On the other hand, if $S = \{i_1\}$, then we have a contradiction to $V$ being safe since we have two distinct forms in $\langle V \rangle$ whose support is  $\{i_1\}$.
        So $\Psi_2$ is a consistent system of rank $(r-1) + l$. 

        \begin{claim*}
            $\Psi_2$ is safe.
        \end{claim*}
        \begin{proof}
            Note that $L(\Psi_2) = L(\Psi_1) [\setminus \{i_1\}]$. Consider any set $D = \{v_1, v_2, \dots, v_p, w_{p+1}, \dots, w_{q}\}$ of linearly independent forms in $\langle L(\Psi_2) \rangle$. By applying suitable invertible operations, we may assume that $v_1, v_2, \dots, v_p$ are supported in $S \setminus \{i_1\}$ and that $w_{p+1}[\setminus S], \dots, w_{q}[\setminus S]$ are linearly independent. We have $p \leq |S| - 1$ since $\dim(L(\Psi_2)_S) \leq \dim(L(\Psi_1)_S) = |S| - 1$. We know that $|\supp(\{w_{p+1}[\setminus S], \dots, w_{q}[\setminus S]\})| \geq q-p$ since $S$ is important for $L(\Psi')$, $L(\Psi')$ is safe and \cite[Lemma 4.8]{egi24} gives that $L(\Psi')[\setminus S]$ must therefore be safe, where we are additionally using that $\langle L(\Psi') \rangle[\setminus S]$ must contain all $w_{p+1}[\setminus S], \dots, w_{q}[\setminus S]$. 
        
        So now we need to show that $|\supp(v_1, v_2, \dots, v_p)| \geq p$. Consider some corresponding linear forms $v_1', v_2', \dots, v_p'$ in $\langle L(\Psi_1) \rangle$ such that applying $[\setminus \{i_1\}]$ to them gives $v_1, v_2, \dots, v_p$ respectively. Since $L(\Psi_1)$ is safe, we have $|\supp(v_1', v_2', \dots, v_p')| \geq p$. Also $\supp(v_1', v_2', \dots, v_p') \subseteq \supp(v_1, v_2, \dots, v_p) \cup \{i_1\}$. If $D$ is dangerous, then we would have $|\supp(v_1, v_2, \dots, v_p)| < p$. This would imply $|\supp(v_1', v_2', \dots, v_p')| \leq |\supp(v_1, v_2, \dots, v_p) \cup \{i_1\}| \leq |\supp(v_1, v_2, \dots, v_p)| + 1 < p+1$. By combining this with $|\supp(v_1', v_2', \dots, v_p')| \geq p$, we get $|\supp(v_1', v_2', \dots, v_p')| = p < |S|$. But then $\supp(v_1', v_2', \dots, v_p')$ is important for $L(\Psi_1)$ contradicting the minimality of $S$ (recall that $\langle V\rangle$ contains $\langle L(\Psi_1)\rangle $). So $\Psi_2$ is safe.
        \end{proof}        

        We can now apply the induction hypothesis to $\Psi_2$ which has rank $r' +l$, where $r' \coloneqq r-1$, and over $m'$ blocks, where $m' \coloneqq m- 2$. It is easy to verify that the required conditions to apply the induction hypothesis are satisfied. Suppose we obtain affine restrictions $\rho_1', \rho_2'$ and sets $F', I' \subseteq [m] \setminus \{i_1, i_1'\}$ by applying the induction hypothesis. Define $\rho_2 \coloneqq \rho_2'$ and $F \coloneqq F'$. The affine restriction $\rho_1$ is obtained by extending $\rho_1'$ by setting $\rho_1(x_{i_1, j_1}) = w|_{\rho_1' \cup \rho_2'} + b$ and $\rho_1(x_{i_1', j_1}) = w|_{\rho_1' \cup \rho_2'} + b+1$. So $I = \{i_1, i_1'\}\cup I'$. 
        
        Most of the desired properties are immediate. Let us verify that $\rho_1 \cup \rho_2$ implies $\Psi$. By construction, each equation in $\Psi_1$ (which is rearranged to give $\Psi_2$) is obtained by taking an equation from $\Psi$ and adding to it some linear combination of $x_{i_1, j_1} + x_{i_1', j_1} = 1$ and $x_{i_1, j_1} + w = b$. Turning this around, every equation in $\Psi$ can be expressed as the sum of some equation from $\langle \Psi_1 \rangle$ (possibly $0 = 0$) and some linear combination of $x_{i_1, j_1} + x_{i_1', j_1} = 1$ and $x_{i_1, j_1} + w = b$. By the induction hypothesis, each equation of $\langle \Psi_1 \rangle $ is implied by $\rho_1' \cup \rho_2'$ which is contained in $\rho_1 \cup \rho_2$. Each of $x_{i_1, j_1} + x_{i_1', j_1} = 1$ and $x_{i_1, j_1} + w = b$ is also implied by $\rho_1 \cup \rho_2$ by the way we defined $\rho_1$. This finishes the proof in this case.
\vspace{10pt}
        \item Suppose there are no important sets for $L(\Psi)$.

        Let $v = b$ be an equation in $\langle \Psi \rangle  \setminus \langle \{\sum_{i \in [m]} x_{i, j} = b_j \mid j \in [l]\} \rangle$. Let $j_1 \in [l]$ be such that $\proj(v, j_1) \notin \{0, \sum_{i \in [m]} x_{i, j_1}\}$. By adding $\sum_{i \in [m]} x_{i, j_1} = b_{j_1}$ to $v=b$ if required, we may assume that $|\proj(v, j_1)| \leq m/2$. Pick any $i_1 \in [m]$ such that the variable $x_{i_1, j_1}$ appears in $v$.

        Since $L(\Psi)$ does not contain any important sets, we will consider \emph{near-important} sets in $L(\Psi)$. A non-empty set $S \subseteq [m]$ is near-important for  $V\coloneqq L(\Psi)$ if $\dim(V_S) \geq |S| - 1$. Since $V$ does not contain any important sets, a set $S$ being near-important for $V$ implies $\dim(V_S) = |S| - 1$. Let $T$ be a maximal near-important set for $V$ which contains $i_1$. Note that $\{i_1\}$ is near-important, so such a maximal near-important set exists. By definition, $|T| = \dim(V_T) + 1 \leq \dim(V) + 1 \leq r+l+1 < m$. So $T \neq [m]$. By reasoning as in the previous case, we obtain the better bound $|T| \leq r+l+1 - l = r+1$.

        We next note that there exists $i_1' \in [m] \setminus T$ such that $x_{i_1', j_1}$ does not appear in $v$. This follows from $m > 2r+l$, $|\proj(v, j_1)| \leq m/2$ and $|T \setminus \{i_1\}| \leq r$ since $T$ always contains $i_1$. 
        
        Let $v = x_{i_1, j_1} + w$ for some linear form $w$. Define $\Psi_1$ to be the linear system obtained from $\Psi$ by first substituting $x_{i_1', j_1} = x_{i_1, j_1} + 1$, then substituting $x_{i_1, j_1} = w + b$ and removing $0=0$ equations. So $\Psi_1$ does not depend on $x_{i_1, j_1}, x_{i_1', j_1}$ and has rank at most $r+l - 1$ since $x_{i_1, j_1} + w = b$ is an equation in $\Psi$. Now let $\Psi_2$ be the linear system obtained from $\Psi_1$ by moving all variables in blocks $\{i_1, i_1'\}$ to the right side of $=$, so that we now think of only the blocks $[m] \setminus \{i_1, i_1'\}$ as variables. 
        
        We claim that the rank of $\Psi_2$ is $r+l-1$ and, in particular, it is consistent as a system on $[m] \setminus \{i_1, i_1'\}$. Note that $L(\Psi_2) = L(\Psi_1)[\setminus \{i_1, i_1'\}]$. So if the rank of $\Psi_2$ is less than $r+l-1$, then there must exist at least two distinct non-zero vectors $v_1, v_2 \in \langle V \rangle$ such that for some $c_i, d_i \in \bF_2$, $v_i + c_{i}(x_{i_1, j_1} + w) + d_i(x_{i_1, j_1} + x_{i_1', j_1})$ only contains variables in $\{x_{i_1, j'}, x_{i_1', j'} \mid j' \neq j\}$. We may assume that $c_1 = 0$. (If $c_2 = 0$, then we can swap $v_1$ and $v_2$. Otherwise, if both $c_1 = c_2 = 1$, we can consider $v_1 + v_2 \in \langle V \rangle$ which is non-zero and 
        \begin{align*}
            & (v_1+ v_2) + (d_1+d_2)(x_{i_1, j_1} + x_{i_1', j_1}) \\
            &= (v_1 + (x_{i_1, j_1} + w) + d_1(x_{i_1, j_1} + x_{i_1', j_1})) + (v_2 + (x_{i_1, j_1} + w) + d_2(x_{i_1, j_1} + x_{i_1', j_1}))
        \end{align*}
        only contains variables in $\{x_{i_1, j'}, x_{i_1', j'} \mid j' \neq j\}$.) 

        Since $v_1 + d_1(x_{i_1, j_1} + x_{i_1', j_1})$ only contains variables in $\{x_{i_1, j'}, x_{i_1', j'} \mid j' \neq j\}$, $\supp(v_1) \subseteq \{i_1, i_1'\}$.
        It cannot be that $|\supp(v_1)| = 1$, since there are no important sets in $V$. So $\supp(v_1) = \{i_1, i_1'\}$. This means that $\{i_1, i_1'\}$ is near-important in $V$. 
        
        We claim that $T \cup \{i_1, i_1'\}$ is near-important which contradicts the maximality of $T$ since $i_1' \notin T$. 
        \begin{claim*}
            If $A, B \subseteq [m]$ are near-important in $V$ and $A \cap B \neq \emptyset$, then $A \cup B$ is near-important.
        \end{claim*}
        \begin{proof}
        To upper bound $\dim(V_{A \cap B})$ below, we will use the assumption $A \cap B \neq \emptyset$ and that there are no important sets for $V$.
             \begin{align*} \vspace{-25pt}
            &\dim(V_{A \cup B}) \geq \dim(V_A \cup V_B) = \dim(V_A) + \dim(V_B) - \dim(V_A \cap V_B) \\
            &\geq \dim(V_A) + \dim(V_B) - \dim(V_{A \cap B}) \geq (|A| - 1) + (|B|-1) - (|A \cap B| - 1) \geq |A \cup B| - 1 
        \end{align*}
        Since  $A \cup B \neq \emptyset$, this implies that  $A \cup B$ is near-important in $V$.
        \end{proof}

        The above claim implies that $T \cup \{i_1, i_1'\}$ is near-important (where $T \cap \{i_1, i_1'\} \neq \emptyset$ since $i_1$ belongs to the intersection), which contradicts the maximality of $T$. So it must be that the rank of $\Psi_2$ is $r+l-1$.

\begin{claim*}
    $\Psi_2$ is safe. 
\end{claim*}
\begin{proof}
    Suppose $\Psi_2$ is not safe. Then there exists a minimal dangerous set $D$ in $\langle L(\Psi_2) \rangle$. Let $S \coloneqq \supp(D)$. By the minimality of $D$, $|S| = \dim(D) + 1$. Let $D^0$ denote a set of $|D|$ corresponding linear forms in $\langle L(\Psi) \rangle$ so that substituting $x_{i_1', j_1} = x_{i_1, j_1}$, then $x_{i_1, j_1} = w$ and applying $[\setminus \{i_1, i_1'\}]$ gives $D$. Let $D^0 = \{v_1, v_2, \dots, v_{|D|}\}$. For each $p \in [|D|]$, we may assume that either both $x_{i_1, j_1}, x_{i_1', j_1}$ appear in $v_p$ or none of $x_{i_1, j_1}, x_{i_1', j_1}$ appear in $v_p$. If this was not the case, then we can consider $v_p + x_{i_1, j_1} + w$ which lies in $\langle L(\Psi) \rangle$ and satisfies the desired condition since $w$ does not involve $x_{i_1', j_1}$ by the choice of $i_1'$. 

        Under this condition that each $v_p$ contains both $x_{i_1, j_1}, x_{i_1', j_1}$ or neither, when applying the substitution $x_{i_1', j_1} = x_{i_1, j_1}$ followed by $x_{i_1, j_1} = w$, we never need to use $x_{i_1, j_1} = w$. So the overall effect of substituting $x_{i_1', j_1} = x_{i_1, j_1}$, then $x_{i_1, j_1} = w$, and applying $[\setminus \{i_1, i_1'\}]$ to $D^0$ is equivalent to just applying $[\setminus \{i_1, i_1'\}]$. So $D^0[\setminus \{i_1, i_1'\}] = D$. This implies $\supp(D^0) \subseteq \supp(D) \cup \{i_1, i_1'\}$. On the other hand, since $D^0 \subseteq \langle L(\Psi) \rangle$ and there are no important sets for $\langle L(\Psi) \rangle$, $\dim(D) \leq \dim(D^0) \leq |\supp(D^0)| - 1$. Combining these with $\dim(D) = |\supp(D)| + 1$, we obtain $\supp(D^0) = \supp(D) \cup \{i_1, i_1'\}$ and $\dim(D^0) \geq \dim(D) = |\supp(D)|+1 = |\supp(D^0)| -1$. So $\supp(D^0)$ is near-important for $L(\Psi)$. By the above claim, we obtain that $T \cup \supp(D^0)$ is near-important for $L(\Psi)$ which contradicts the maximality of $T$ since $i_1' \in \supp(D^0) \setminus T$. So $\Psi_2$ is safe.
\end{proof}

        We can now apply the induction hypothesis to $\Psi_2$ to obtain affine restrictions $\rho_1', \rho_2'$ and sets $F', I' \subseteq [m] \setminus \{i_1, i_1'\}$. The affine restrictions $\rho_1, \rho_2$ and the sets $F, I$ are defined just like they were in the other case and they can be verified to have the desired properties in exactly the same way. \qedhere
    \end{itemize}
\end{proof}

\begin{lemma} \label{lem:coll_aff_restrict}
    Let $A_0 \subseteq \B^l$. Let $v, A, L, F, L_F$ be returned by a successful run of Algorithm \ref{alg:sim_multicoll} when run on $T$ and $\cS$ where $\cS$ is a multiset containing only elements from $A_0$, such that $k$ elements appear twice in $\cS$ and no element appears more than twice. Let $\Psi$ be a linear system which is implied by the constraints of $L$ together with $L_F$. Extend $\Psi$ so that its span includes $\{\sum_{i \in [m]} x_{i, j} = \sum_{y \in \cS} y_j \mid j \in [l]\}$. 
    Let $l+r$ be the safe dimension of $\Psi$. If $n/2 + 2r \leq |A_0|$, $|F| > r$ and $m > 2r+l$, then there exists an affine restriction $\rho$ fixing blocks $[m]\setminus \hat{F}$ and a collection of equations $L_{\hat{F}} = \{\sum_{i \in \hat{F}} x_{i, j} = b_j \mid j \in [l]\}$ for some $b_j \in \bF_2 (j \in [l])$ satisfying the following conditions:
    \begin{enumerate}
        \item The number of blocks fixed by $\rho$, say $s$, is at most $2r$.
        \item $\Psi$ is implied by $\rho$ together with $L_{\hat{F}}$.
        \item There exists a set $U \subseteq A_0$ such that $|U| = s$ 
        and for any assignment to $\hat{F}$ satisfying $L_{\hat{F}}$, if we set the blocks $[m]\setminus \hat{F}$ according to $\rho$, all the strings assigned to $[m]\setminus \hat{F}$ lie in $U$ and are distinct.
        \item For all $j \in [l]$, $b_j = (\sum_{y \in \cS} y_j) + (\sum_{y \in U} y_j)$.
    \end{enumerate}
\end{lemma}
\begin{proof}
We will use below that $|\cl(\Psi)| \leq r$. To see this, first note that $|\cl(\Psi)| \leq \rk(\Psi) \leq r+l < m$ by Lemma \ref{lem:size_bound} and the assumption $m > 2r+l$. So $\dim(L(\Psi)[\setminus \cl(\Psi)]) \geq l$ since $\langle L(\Psi) \rangle$ contains $\{\sum_{i \in [m]} x_{i, j}  \mid j \in [l]\}$, these $l$ linear forms are on disjoint sets of variables and none of these linear forms simplifies to $0$ after applying $[\setminus \cl(\Psi)]$. So we have $|\cl(\Psi)| \leq \rk(\Psi) - \dim(L(\Psi)[\setminus \cl(\Psi)]) \leq r$.

Now consider the blocks in $\cl(\Psi) \cap F$. Let $\hat{U}$ be the set of all possible holes used by the blocks in $\cl(\Psi) \setminus F$ during this run of the simulation. Since we assign at most two holes to each fixed block during the simulation, $|\hat{U}| \leq 2|\cl(\Psi) \setminus F|$. Assign distinct strings from $A_0 \setminus \hat{U}$ to the blocks in $\cl(\Psi) \cap F$. This can be done if $|\cl(\Psi) \cap F| \leq |A_0| -  2|\cl(\Psi) \setminus F|$. This holds since $2|\cl(\Psi) \setminus F| + |\cl(\Psi) \cap F| \leq 2(|\cl(\Psi) \setminus F| + |\cl(\Psi) \cap F|) = 2|\cl(\Psi)| \leq 2r \leq |A_0|$ where the second inequality was shown above and the third inequality uses the assumption $|A_0| \geq 2r+n/2$. 

Now assign all blocks in $F \setminus \cl(\Psi)$ arbitrarily in such a way that $L_F$ is satisfied. This can be done since $|F| > r \geq |\cl(\Psi)|$, so there is some block which has not already been fixed and this can be chosen appropriately to ensure that $L_F$ is satisfied. Extend to a full assignment $x$ by fixing all blocks in $[m] \setminus F$ according to $L$. By Lemma \ref{lem:inv}, there is a unique such extension which also ensures that all fixed blocks lie in $A_0$ since $\cS$ only contains elements from $A_0$. By assumption, $L$ and $L_F$  together imply $\Psi$ and in particular, the full assignment $x$ defined above satisfies $\Psi$. 

This means that the partial assignment to blocks in $\cl(\Psi)$ which agrees with $x$ satisfies all equations in $\langle \Psi \rangle$ whose support is contained in $\cl(\Psi)$. Let $\sigma_1$ denote this (bit-fixing) restriction which sets blocks in $\cl(\Psi)$ according to $x$. Since we are considering a successful run of the algorithm, the blocks in $\cl(\Psi) \setminus F$ are assigned distinct strings from $\hat{U}$ which only contains strings  from $A_0$. Moreover, we ensured above that each block in $\cl(\Psi) \cap F$ is not assigned a string from $\hat{U}$. So $\sigma_1$ fixes blocks in $\cl(\Psi)$ to distinct strings in $A_0$. Let $U_1$ denote the set of strings assigned by $\sigma_1$ to blocks in $\cl(\Psi)$.

Now consider $\Psi' \coloneqq \Psi|_{\sigma_1}$. Since $\sigma_1$ fixes blocks in $\cl(\Psi)$ to bits, $\Psi'$ is safe. 
$\langle \Psi' \rangle$ contains the equations $\sum_{i \in [m] \setminus \cl(\Psi)} x_{i, j} = (\sum_{y \in \cS} y_j) + (\sum_{y \in U_1} y_j)$ for all $j \in [l]$. Let $m' = m - |\cl(\Psi)|$ be the number of blocks involved in $\Psi'$ and let $r' \coloneqq \rk(\Psi') - l$. Since the dimension of the linear forms in $\langle L(\Psi) \rangle$ supported on $\cl(\Psi)$ is at least $|\cl(\Psi)|$ (Lemma \ref{lem:size_bound}), we obtain $m' > 2r' + l$. By applying Lemma \ref{lem:find_pivots} to $\Psi'$, we get affine restrictions $\rho_1'$, $\rho_2'$ and sets $F', I' \subseteq [m] \setminus \cl(\Psi)$ satisfying the following properties. The restriction $\rho_1'$ fixes indices $G' = \{(i_1, j_1), (i'_1, j_1), (i_2, j_2), (i'_2, j_2), \dots, (i_{r'}, j_{r'}), (i'_{r'}, j_{r'})\} \subseteq ([m] \setminus (\cl(\Psi) \cup F')) \times [l]$. The restriction $\rho_2'$ fixes for all $j \in [l]$,
\[
\rho_2'(x_{i_0, j}) = c_j + \sum_{t \in [r']} (\llbracket j_t \neq j \rrbracket (x_{i_t, j_t} + x_{i'_t, j_t})+\llbracket j_t = j \rrbracket \cdot 1) + \sum_{i \in F'} x_{i, j}
\]
where $i_0$ is the unique block in $[m] \setminus (\cl(\Psi) \cup F' \cup I')$ and $c_j =  (\sum_{y \in \cS} y_j) + (\sum_{y \in U_1} y_j)$.

We will now define a bit-fixing restriction $\sigma_2$ which fixes all the bits $x_{i, j}$ for $(i, j) \in (I \times [l]) \setminus G$. The restriction $\sigma_2$ will have the property that $x_{i_t, j} = x_{i'_t, j}$ for all $t \in [r'], j \neq j_t$. Moreover, $\sigma_2$ will ensure that all $x_{i_t}, x_{i'_t}$ lie in $A_0 \setminus U_1$ and for distinct $s, t \in [r']$, $x_{i_{s}} \neq x_{i_t}$ irrespective of what $x_{i_s, j_s}$ and $x_{i_t, j_t}$ are. We construct $\sigma_2$ inductively in a greedy fashion. Suppose we have defined $\sigma_2$ on the blocks $\{i_1, i_1', \dots, i_{t-1}, i'_{t-1}\}$ for $t \leq r'$. Consider pairs of holes (strings in $\B^l$) which differ on the $j_t^{th}$ bit. Out of these $n/2$ pairs, a pair cannot be used if some hole in the pair lies outside $A_0$ or if some hole in the pair has already been assigned to a block in $\cl(\Psi)$ by $\sigma_1$ or to a block in $\{i_1, i_1', \dots, i_{t-1}, i'_{t-1}\}$ by $\sigma_2$. So the number of such forbidden pairs is at most $(n - |A_0|) + |\cl(\Psi)| + 2(t-1) \leq n - |A_0| + 2(|\cl(\Psi)| + t-1) \leq n - |A_0| + 2(r-1) < n/2$ where the second inequality uses Lemma \ref{lem:sdim_bound} and the last inequality uses the assumption $|A_0| \geq n/2 + 2r$. Thus, we can find some pair of holes differing on the $j_t^{th}$ bit which we can assign to the blocks $\{i_t, i'_t\}$. So we can find $\sigma_2$ with the desired properties. 
Let $U_2$ denote the collection of holes assigned by $\sigma_2$ to $I$. By construction, $|U_2| = |I| = 2r'$.

Finally, let $\rho$ be the affine restriction obtained by combining $\sigma_1, \sigma_2$ and $\rho_1'$ where we simplify $\rho_1'$ according to $\sigma_2$ so that it fixes $I'$ as affine functions of only the blocks in $F'$. Define $\hat{F} = F' \cup \{i_0\}$ which is the same as $[m] \setminus (\cl(\Psi) \cup I')$. $L_{\hat{F}}$ is obtained by rewriting $\rho_2$ as a system of linear equations after simplifying according to the restriction $\sigma_2$. Define $U = U_1 \cup U_2$.

We now verify that these have the desired properties. The number of blocks fixed by $\rho$ is $|\cl(\Psi)| + 2r' = |\cl(\Psi)| + 2(\rk(\Psi') - l) \leq 2(\rk(\Psi) - l) = 2r$ where the inequality uses Lemma \ref{lem:sdim_bound}. 

$L_{\hat{F}}$ together with $\rho$ implies $\Psi$ since $\sigma_1$ implies all the equations in $\langle \Psi \rangle$ which are supported on $\cl(\Psi)$, $\rho_1' \cup \rho_2'$ implies $\Psi'$ and $\rho$ essentially contains $\sigma_1, \sigma_2, \rho_1'$. 

We have $|U| = |U_1| + |U_2| = |\cl(\Psi)| + |I|$ which is exactly the number of blocks fixed by $\rho$. It is also not hard to see that the construction of $\rho$ ensures that all the blocks in $[m] \setminus \hat{F}$ are always assigned distinct strings. Here we are using that the affine restriction $\rho_1'$ given by Lemma \ref{lem:find_pivots} ensures that for each $t \in [r']$, $\rho_1'$ implies $x_{i_t, j_t} + x_{i_t', j_t} = 1$ which means $x_{i_t, j_t} \neq x_{i_t', j_t}$. It is also easy to verify that $b_j = (\sum_{y \in \cS} y_j) + (\sum_{y \in U} y_j)$ for all $j$ as claimed.
\end{proof}

\subsection{Putting everything together}

For a set $A \subseteq \B^l$ of allowed holes and $z \in \B^l$, we define $\coll_{A, z}^m \subseteq (\B^l)^m \times \binom{[m]}{2}$ to be the collision-finding problem with the additional promise that the input $x \in (\B^l)^m$ satisfies $x_i \in A$ for all $i \in [m]$ and $\sum_{i \in [m]} x_i = z$ where the sum is vector addition in $\bF_2^l$. The latter condition equivalently states that for all $j \in [l]$, $\sum_{i \in [m]} x_{i, j} = z_j$.

The next lemma shows that if $A$ is sufficiently large, then we can always find a multiset $\cS \subseteq A$ of size $m$ which additionally satisfies the condition that the bitwise sum of all strings in $\cS$ is $z$ and which has close to the smallest number of collisions possible.

\begin{lemma} \label{lem:find_multiset}
    Let $A \subseteq \B^l, y \in \B^l$ and $m \coloneqq |A| + k$. Suppose $|A| > n/2 + k$. Then there exists a multiset $\cS$ of size $m$ satisfying the following properties. $\cS$ contains only strings from $A$ and no string appears more than twice in $\cS$. Moreover, $\sum_{y \in \cS} y = z$ and at most $k+1$ strings appear twice in $\cS$.
\end{lemma} 
\begin{proof}
    Fix an arbitrary multiset $\cS$ of size $m$ containing strings from $A$ such that exactly $k$ strings appear twice in $\cS$. So each string from $A$ appears at least once in $\cS$. Define $z' \coloneqq \sum_{y \in \cS} y$. If $z' = z$, we are done. So suppose $z' \neq z$.

    We wish to exchange an element from $\cS$ with another element so that the resulting multiset has the desired sum. For this, we will show that there exists a pair of distinct elements $\{y_1, y_2\}$ such that $y_1 + y_2 = z + z'$ for which both $y_1, y_2$ are in $A$ and neither $y_1$ nor $y_2$ appears twice in $\cS$. There are $n/2$ pairs $\{y_1, y_2\}$ such that $y_1 + y_2 = z + z'$. Out of these, the number of bad pairs such that at least one of the strings in the pair is outside $A$ is at most $n - |A|$. The number of pairs such that some element from it appears twice in $\cS$ is at most $k$. So the total number of bad pairs is at most $(n-|A|) + k < n/2$ and there exists some pair $\{y_1, y_2\}$ satisfying the desired conditions.

    Now consider the multiset $\cS'$ obtained from $\cS$ by removing an occurrence of $y_1$ and replacing it by $y_2$. We have $\sum_{y \in \cS'} y = (\sum_{y \in \cS} y) + y_1 + y_2 = z' + (z' + z) = z$. Moreover, by the choice of $y_1, y_2$, the number of strings that appear twice in $\cS'$ is $k+1$.
\end{proof}

\begin{lemma} \label{lem:coll_dag_restrict}
    There exists a constant $\delta > 0$ such that the following holds. Let $k \leq \delta n$. Let $A \subseteq \B^l$ with $|A| \geq 7n/8$. Let $m \coloneqq |A| + k$. Let $z \in \B^l$. Suppose there is an affine DAG $C$ solving $\coll_{A, z}^m$ whose depth is at most $D$ and size is at most $S$. 
    Let $d$ be an integer satisfying $8 \leq d \leq n/200$. There exist $s \leq O(\log S + d^2k/n^2)$ and $b \in \B^l$ such that the following holds. Let $m' \coloneqq m - s$.
    There is some $A' \subseteq A$ with $|A'| = |A| - s$ such that there exists an affine DAG $C''$ solving $\coll_{A', b}^{m'}$ whose depth is at most $D - d$ and size is at most $S$.
\end{lemma}
\begin{proof}
Use Lemma \ref{lem:find_multiset} to obtain a multiset $\cS$ containing $m$ strings from $A$ such that no element appears more than twice, $\sum_{y \in \cS} y = z$ and there are $k'$ strings which appear twice in $\cS$ where $k' \in \{k, k+1\}$. 

Let $T$ be the depth $d$ PDT obtained by starting from the root of the DAG $C$, repeating nodes at depth at most $d$ and removing nodes at depth greater than depth $d$. Any leaves in $T$ at depth less than $d$ correspond to leaves of $C$. So if an input $x$ which satisfies the promise of $\coll_{A, y}^m$ reaches a leaf in $T$ at depth less than $d$, then $T$ must find a collision on $x$. In particular, this holds for any input $x$ which is a permutation of the multiset $\cS$. Run Algorithm \ref{alg:2coll_sim} on $T$ and the multiset $\cS$. By Lemma \ref{lem:coll_fail_prob}, it succeeds with probability $p$ at least $(1-2/e) \exp(-O(d^2k'/n^2)) \geq (1-2/e)\exp(-O(d^2k/n^2))$. The conditions for applying Lemma \ref{lem:coll_fail_prob} hold if we ensure $\delta \leq \delta'/2$, so that $k' \leq k+1 \leq 2\delta n \leq \delta'n$.

Since the affine DAG only has at most $S$ nodes, there is some node $w$ in $C$ such that the corresponding leaves in $T$ are reached with probability at least $p/S$ by successful runs of the simulation. By Lemma \ref{lem:coll_rpdt_dist}, this is also a lower bound on the probability that when $x$ is picked uniformly at random from among all permutations of $\cS$, the path taken by $x$ in $C$, according to the queries made starting at the root, reaches the node $w$.

Let $\Phi$ denote the linear system labeling the node $w$. Let $\Phi'$ be the linear system obtained from $\Phi$ by also including the equations in $\{\sum_{i \in [m]} x_{i, j} = z_j \mid j \in [l]\}$. $\Phi'$ is still satisfiable since every permutation of $\cS$ satisfies $\{\sum_{i \in [m]} x_{i, j} = z_j \mid j \in [l]\}$ and some such permutation reaches $w$, which implies that it satisfies $\Phi$. Let $\Psi$ be the safe system obtained by using 
Lemma \ref{lem:sdim_bound} on the forms $L(\Phi')$ with $U \coloneqq \{\sum_{i \in [m]} x_{i, j} \mid j \in [l]\}$ and then considering the corresponding equations in $\langle \Phi' \rangle$. Let $r+l$ be the rank of $\Psi$, which is the safe dimension of $\Phi'$. The probability that a random permutation of $\cS$ satisfies $\Phi$ is at most the probability that it satisfies $\Psi$, since each permutation of $\cS$ satisfies $\Phi$ if and only if it satisfies $\Phi'$, and $\Phi'$ implies $\Psi$. So we can use Lemma \ref{lem:perm_prob} and the above lower bound of $p/S$ on the probability of $\Phi$ to be satisfied to get that $r \leq O(\log S + d^2k/n^2)$. The condition on $k$ required for applying Lemma \ref{lem:perm_prob} holds as long as we choose $\delta \leq 1/20$. To see $m - 2r \geq 3n/4$, note that $r \leq \rk(\Phi) \leq d \leq n/200$ by Lemma \ref{lem:path_implies}. Now combining this with the assumption $m \geq |A| \geq 7n/8$ gives $m - 2r \geq 3n/4$.

Now fix a successful run of the simulation which ends at a leaf $v$ in $T$ corresponding to $w$. To apply Lemma \ref{lem:coll_aff_restrict}, let us verify that the required conditions are satisfied. We have $|A| \geq 7n/8$ by assumption and $2r \leq 2d \leq n/100$ by Lemma \ref{lem:path_implies}. This gives $|A| \geq n/2 + 2r$. To see $|F| > r$, where $F$ is the set of free blocks at the end of this successful run of the simulation, note that $|F| \geq m - 16d \geq 7n/8 - 16n/200 = 159n/200$ and $r \leq d \leq n/200$. Similarly $m \geq 7n/8$ and $2r+l \leq n/100 + \log_2 n$ which imply $m > 2r+l$.
    Having verified all the required conditions, apply Lemma \ref{lem:coll_aff_restrict} to this successful run and the linear system $\Phi$ to obtain $\rho$, $s$, $\hat{F}$, $L_{\hat{F}}$ and $U \subseteq A$. Here $\rho$ is an affine restriction fixing $s$ blocks as affine functions of the other $m - s$ blocks. Moreover, $s \leq 2r \leq O(\log S + d^2k/n^2)$. We have $L_{\hat{F}} = \{\sum_{i \in \hat{F}} x_{i, j} = b_j \mid j \in [l]\}$ where $b = z + \sum_{y \in U} y$ by Lemma \ref{lem:coll_aff_restrict} and $\hat{F} \neq \emptyset$ since $m > 2r+l$.

    Set $m' = m - s$,  and $A' = A \setminus U$. Consider $C|_\rho$ which solves $\coll_{A, z}^m|_\rho$. We claim that $C|_\rho$ solves $\coll_{A', b}^{m'}$ on the index set $\hat{F}$ after some slight changes. For any sink node whose output label lies outside $\binom{\hat{F}}{2}$, we change it to any pair from $\binom{\hat{F}}{2}$. It is not hard to see that this does not affect any strings $x$ satisfying the promise of $\coll_{A', b}^{m'}$. Indeed for any $x \in (A')^{\hat{F}}$ such that $\sum_{i \in [m]} x_i = b$, its unique extension $y$ according to $\rho$ is a valid input for $\coll_{A, z}^m$. Moreover, Lemma \ref{lem:coll_aff_restrict} ensures that all blocks in $[m] \setminus \hat{F}$ are assigned distinct strings from $U$ which is disjoint from $A'$. So the only allowed solutions for $y$ must be pairs from $\binom{\hat{F}}{2}$ which we have kept unchanged in the DAG. So $C|_\rho$ solves $\coll_{A', b}^{m'}$ on the index set $\hat{F}$.

    Pick any $i_0 \in \hat{F}$. Consider an affine restriction $\rho'$ which sets $\rho'(x_{i_0, j}) = b_j + \sum_{i \in \hat{F} \setminus \{i_0\}} x_{i, j}$ for all $j \in [l]$. Note that $\rho'$ is equivalent to $L_{\hat{F}}$. We now apply the affine restriction $\rho'$ to $C|_\rho$ to obtain another DAG $C'$ which we will continue to interpret as being over the blocks $\hat{F}$ instead of $\hat{F} \setminus \{i_0\}$. It is easy to see that $C'$ still solves $\coll_{A', b}^{m'}$ on the index set $\hat{F}$ since all the inputs satisfying the promise of $\coll_{A', b}^{m'}$ are consistent with $\rho'$.

    Recalling that $\rho$ together with $L_{\hat{F}}$ implies $\Phi$, we see that the node $w$ in $C'$ is labeled by the empty system. So we may consider the DAG $C''$ which only contains the nodes in $C'$ which are reachable from $w$. The DAG $C''$ solves $\coll_{A', b}^{m'}$ on the index set $\hat{F}$. Finally we relabel according to any bijection between $F'$ and $[m']$ to obtain a DAG solving $\coll_{A', b}^{m'}$.

    We only need to verify that the size of $C''$ is at most $S$ (which is immediate) and that the depth of $C''$ is at most $D-d$.  To see the depth bound, note that the node $w$ in $C$ cannot be at depth less than $d$ by Lemma \ref{lem:not_sink}. So any path starting from $w$ can have length at most $D-d$ since the depth of $C$ is only at most $D$.
\end{proof}

\begin{theorem}[Theorem \ref{thm:s_bphp} rephrased]
There exists a constant $\delta > 0$ such that the following holds. Let $k \leq \delta n$ and $m \coloneqq n+k$. If there exists an affine DAG solving $\coll_n^m$ whose size is $S$ and depth is $D$, then 
\[
\log S = \Omega \left(\min\left( \frac{n^2}{D}, \frac{n^4}{kD^2} \right)\right).
\]    
\end{theorem}
\begin{proof}
    We assume that $n$ is large enough. Set $A = \B^l$ and $z = 0^l$. Let $C$ be an affine DAG solving $\coll_n^m$ of size $S$ and depth $D$. Since $\coll_{A, z}^m$ is a promise version of $\coll_n^m$, $C$ also solves $\coll_{A, z}^m$.
    
    We will consider cases depending on $S$.
    \begin{itemize}
        \item Suppose $S \geq 2^{k/10^6}$. Set $d = n/200$. Apply Lemma \ref{lem:coll_dag_restrict} repeatedly to $A, z, C$ and $d$, updating $A, z, C$ according to $A', b, C''$ given by the lemma statement. This can be done as long as $|A| \geq 7n/8$. In each iteration, $|A|$ decreases by at most $O(\log S + d^2k/n^2) = O(\log S)$ by the assumption on $S$. So there are at least $(n/8)/O(\log S)$ iterations. In each iteration, the depth of the DAG $C$ reduces by $d$. This implies that the depth $D$ of the original DAG must be at least $\Omega(nd/\log S) = \Omega(n^2/\log S)$.
        \item Suppose $S \leq 2^{k/10^6}$. Set $d = \floor{\frac{n\sqrt{\log S}}{\sqrt{k}}}$. Then by arguing exactly as in the other case and using that $d^2k/n^2 \leq \log S$, we obtain that $D \geq \Omega(nd/\log S) = \Omega(n^2/\sqrt{k\log S})$.
    \end{itemize}
    So we have $D \log S \geq \Omega(n^2)$ or $D \sqrt{\log S} = \Omega(n^2/\sqrt{k})$. By rearranging, we obtain
    \[
\log S = \Omega \left(\min\left( \frac{n^2}{D}, \frac{n^4}{kD^2} \right)\right). \qedhere
\]  
\end{proof}

\section{Lifting DT-hardness to size-depth lower bounds for $\res(\oplus)$}
\label{sec:lifting}

In this section, we prove the lifting theorem for bounded-depth $\res(\oplus)$. 


\subsection{Block-lifted formulas, balanced gadgets and DT-hardness}
\label{subsec:lift_prel}
In this subsection, we collect  preliminaries related to lifting.

We first recall how the formula $\lift(\varphi)$ is constructed in \cite{de2021automating}. 
Let $\varphi$ be a CNF formula on $n$ blocks of $l$ variables. Suppose the variables of $\varphi$ are $z_{i, j}, \; i \in [n], j \in [l]$. The formula $\lift(\varphi)$ is defined in the following way. For each block $i \in [n]$ of $\varphi$, we have a block of $2l+1$ variables, $s_i, x_i^0, x_i^1$ in $\lift(\varphi)$ where $x_i^0$ and $x_i^1$ each contain $l$ variables. 

$\lift(\varphi)$ essentially encodes $\varphi \circ \bind$ as a CNF formula in the natural way and then removes all the trivial clauses containing some $s_i$ and $\neg s_i$.
Formally, we do the following. For each clause $C$ in $\varphi$, we define $\lift(C)$ as the following CNF formula. For each $i \in [n]$ such that some variable of block $i$ appears in $C$, let $C_i$ denote the clause containing only those literals from $C$ which belong to block $i$. For each such $i$, if $C_i$ is given by $z_{i, j_1} \vee z_{i, j_2} \vee \dots \vee \neg z_{i, j_t}$, we define 
\[
\lift(C_i) = (s_i \vee x_{i, j_1}^0 \vee x_{i, j_2}^0 \vee \dots \vee \neg x_{i, j_t}^0) \wedge (\neg s_i \vee x_{i, j_1}^1 \vee x_{i, j_2}^1 \vee \dots \vee \neg x_{i, j_t}^1).\]
Now $\lift(C)$ is the natural CNF formula obtained by expanding $\vee_i \lift(C_i)$ where the disjunction is only over $i$ such that some variable from block $i$ appears in $C$. So if $C$ has block-width $b$, $\lift(C)$ contains $2^b$ clauses. Finally, $\lift(\varphi) = \wedge_{C \in \varphi} \lift(C)$. So $\lift(\varphi)$ contains $2^b m$ clauses of width $b + w$ where $m$ is the number of clauses of $\varphi$, $w$ is its width and $b$ its block-width.

We used the block-indexing function above to ensure that the number of clauses in the lift of a clause only depends on the block-width of the clause. In the rest of this section, we will allow for any gadget $g: \B^{m} \rightarrow \B^l$ (over input bits $u_1, \dots, u_m$) which is $\delta$-balanced \cite{bhat2024, bi25} defined next\footnote{In \cite{bi25}, the term `affine balanced' was used, which generalized the definition of `balanced, stifled' in \cite{bhat2024}. Both these works only dealt with Boolean-valued functions, but what we define here is the natural extension to multi-output functions.}. For $\delta \in (0, 1)$, a function $g: \B^m \rightarrow \B^l$ is $\delta$-balanced if there exist distributions $\mu_y$ for all $y \in \B^l$ such that the following holds. The distribution $\mu_y$ is supported only on $g^{-1}(y)$ and for each $j \in [m]$, the distribution $\mu_y$ is equal to a mixture $(1-\delta)\;A^j_y + \delta \;B^j_y \times U_1$ for some distributions $A^j_y$ on $\B^m$ and $B^j_y$ on $\B^{m-1}$. Here $U_1$ is a uniform random bit independent of $B^j_y$ and we think of $U_1$ as the bit $u_j$ and $B^j_y$ as a distribution for $u_{[m] \setminus \{i\}}$. For probability distributions $X, Y$ on the same sample space, the notation $(1-\delta)X + \delta Y$ means that with $\delta$ probability, we sample from $Y$ and with the remaining $1-\delta$ probability from $X$.

Note that a $\delta$-balanced function $g$ is necessarily stifling \cite{cmss23}, which means that for each $y \in \B^l, j \in [m]$, there is some way to fix all bits except $i$ in the input for $g$ to force the output of $g$ to be $y$. It is also easy to see that every stifling function $g$ is $1/m$-balanced as observed in \cite{bi25}. Specifically, we may take $\mu_y$ to be the following distribution. Pick $j \in [m]$ uniformly at random. Fix $u_{[m]\setminus \{j\}}$ to an assignment forcing $g$ to output $y$. Set $u_j$ to a uniform random bit. This shows that $g$ is $1/m$-balanced since for any $j$, when it is picked in the first step, the bit $u_j$ is uniform and independent of the other bits.

However, we would like $\delta$ to be a constant even if $m$ is growing.
We observe that the block-indexing function $\bind_l$ is $1/3$-balanced (independent of $l$). The distribution is similar to the distribution above for $\ind_{1+2}$. We pick uniformly at random one of the three parts $s, x^0, x^1$. The chosen part is now set uniformly at random while the other two parts are fixed to force $\bind_l$ to output $y$. Formally, for each $y \in \B^l$, $\mu_y$ is the following distribution: 
\[\mu_y = \frac{1}{3} U_1 \times \{y\} \times \{y\} + \frac{1}{3} \{0\} \times \{y\} \times U_l + \frac{1}{3} \{1\} \times U_l \times \{y\}\]
where $U_l$ is the uniform distribution on $\B^l$ and $\{y\}$ is the distribution which gives $y$ with probability $1$.
So $\bind_l$ is $1/3$-balanced.

We next consider DT-hardness with another parameter for the success probability. 
A block-respecting partial assignment $\rho \in (\B^l \cup \{*\})^n$ is a certificate for a search problem $\cR \subseteq (\B^l)^n \times \cO$ if there exists $o \in \cO$ such that for all $z \in (\B^l)^n$ consistent with $\rho$, we have $(z, o) \in \cR$. We use the notation $\free(\rho)$ to denote the set of free blocks of $\rho$ and $\fix(\rho)$ to denote the set of fixed blocks of $\rho$.

$\cP \subseteq (\B^l \cup \{*\})^n$ is $(p, q, \eps)$-block-DT-hard for a search problem $\cR \subseteq (\B^l)^n \times \cO$ if the following hold:
\begin{itemize}
    \item No partial assignment in $\cP$ is a certificate for $\cR$.
    \item $\cP$ contains the empty partial assignment $*^n$.
    \item $\cP$ is downward closed: if $\rho'$ can be obtained from $\rho$ by forgetting the value of some fixed block and $\rho \in \cP$, then $\rho' \in \cP$.
    \item For every assignment $\rho \in \cP$ which fixes at most $p$ blocks, there is some distribution $\nu_\rho$ over assignments to the blocks not fixed by $\rho$ such that the following holds. If $T$ is any block decision tree of depth $q$ which only queries blocks not fixed by $\rho$, then $\Pr[\rho \; \cup \; \sigma_{T, \nu_{\rho}} \in \cP ] \geq \eps$.
\end{itemize}
Recall that $\sigma_{T, \nu}$ is the partial assignment corresponding to the random leaf reached when $T$ is run on the distribution $\nu$.
We will always assume that $p, q, \eps > 0$.

In what follows, we will deal with lifts of arbitrary search problems. For a search problem $\cR \subseteq (\B^l)^n \times \cO$ and a gadget $g: \B^m \rightarrow \B^l$, we define $\cR \circ g \subseteq (\B^m)^n \times \cO$ by $(x, o) \in \cR \circ g$ iff $(g^n(x), o) \in \cR$ where $g^n(x)$ is $g(x_1) g(x_2) \dots g(x_n)$.


\subsection{PDT lower bound for $\cR|_{\rho} \circ g$}
In this subsection, we give a simulation to analyze parity decision trees for $\cR|_{\rho} \circ g$ on a lifted distribution where $\rho$ is some partial assignment and we have some distribution $\nu_\rho$ which is hard for depth $q$ block decision trees with respect to $\cP$.
The simulation will use the distributions $\{\mu_y\}_{y \in \B^l}$ given by $g$ being $\delta$-balanced and how these distributions can be expressed as suitable mixtures. 

The simulation in Algorithm \ref{alg:sim_lift} is essentially from \cite{bi25} though it does not appear there in exactly the form we need. The main notable change is that when Algorithm \ref{alg:sim_lift} succeeds, it also returns some of the internal variables used during the simulation. These are used in later parts of the proof. 

We provide proofs of all properties we need from the simulation, primarily for completeness, but also because we need the simulation to terminate quickly with  probability very close to $1$ (instead of just constant probability), which was not observed in \cite{bi25}. Such high probability statements for simulation theorems have appeared earlier in the context of communication complexity \cite{goos2020query, riazanov2025searching}, though the proof here is much simpler.

\begin{algorithm} \label{alg:sim_lift}
\caption{}
\KwIn{ $\rho \in (\B^l \cup \{*\})^n$, PDT $T$ on $(\B^m)^{\free(\rho)}$}
$F \gets \free(\rho)$ \; 
$L \gets \emptyset$ \tcp*{Collection of equations}
$v \gets$ root of $T$\;
$\sigma \gets *^{\free(\rho)}$ \tcp*{Partial assignment storing queries to $z$}
$iter \gets 0$ \;

\While{$v$ is not a leaf}{
$P' \gets$ query at $v$ \;
$P \gets P'$ after substituting according to $L$\; 
\tcp{$P$ now depends only on blocks in $F$ and is equivalent to $P'$ under $L$}
\eIf{$P$ is a constant, $b \in \bF_2$}{ \label{algline:sim_trivial}
Update $v$ according to $b$\;
}
{ \label{algline:sim_nontrivial}
$(i, j) \gets \min\{(i, j) \in F \times [m] \mid x_{i, j }\text{ appears in } P\}\}$ \; 
Query $z_i$ and update $\sigma$ accordingly\;
Sample $c \sim \Bern(\delta)$\;
\eIf{$c$ is $1$}{ \label{algline:sim_good}
Sample $u' \sim B_{z_i}^j$\;
$L \gets L \cup \{x_{i, h} = u'_h \mid h \in [m]\setminus \{j\}\}$\;
Pick $b \in \bF_2$ uniformly at random\;
$L \gets L \cup \{P = b\}$\;
Update $v$ according to $b$\;
}{
Sample $u \sim A_{z_i}^j$\;
$L \gets L \cup \{x_{i, h} = u_h \mid h \in [m]\}$\;
}
$F \gets F \setminus \{i\}$\;
}

$iter \gets iter+1$\; \label{algline:sim_iter}
}
\eIf{$\rho \;\cup \sigma \notin \cP$}{
\KwRet FAIL \;
}
{
\KwRet $v, \sigma, L, F$
}
\end{algorithm}

\begin{lemma} \label{lem:lift_inv}
    The simulation maintains the following invariants in the beginning of each iteration of the while loop and at termination:
    \begin{enumerate}
        \item The collection of equations $L$ uniquely determines $x_i\; (i \in [m] \setminus F)$ as affine functions of $x_{i'} \;(i' \in F)$. 
        \item $L$ implies all the parity constraints on the path from the root to the current node $v$.    
        \item For each assignment to all blocks in $F$, if we assign values to $x_i \; (i \in \free(\rho) \setminus F)$ according to $L$, we have $g(x_i) = z_i$ for each  $i \in \free(\rho) \setminus F$.
        \item The partial assignment $\sigma$ fixes exactly those blocks in $z$ that have been queried so far, and $\sigma$ is consistent with $z$.
    \end{enumerate}
\end{lemma}
\begin{proof}
    The first two invariants are similar to those in Lemma \ref{lem:inv} so we omit them. The invariant about $\sigma$ is also easy to see. 
    
    So we only need to verify that $g(x_i) = z_i$ for each $i \in \free(\rho) \setminus F$ irrespective of how blocks in $F$ are assigned. This follows from how the distributions $A_{z_i}^j$ and $B^j_{z_i}$ are defined.     
    For all $j$, the distribution $A_{z_i}^j$ is supported only on $g^{-1}(z_i)$.
    Since $B^j_{z_i} \times U_1$ is supported only on $g^{-1}(z_i)$, any string in the support of $B^j_{z_i}$ always forces the output of $g$ to $z_i$.
\end{proof}

For $z \in (\B^l)^{\free(\rho)}$, we use $\mu_z$ to denote the distribution $ \prod_{i \in \free(\rho)} \mu_{z_{i}}$. For a subset $F \subseteq \free(\rho)$, we write $\mu_z[F]$ to denote the distribution $\prod_{i \in F} \mu_{z_i}$.

Let $\nu$ be a distribution on $(\B^l)^I$ where $I$ is some index set. 
Then $\nu \circ \mu$ denotes the following distribution on $(\B^m)^I$. Sample $z \sim \nu$. Then return a sample from the distribution $\prod_{i \in I} \mu_{z_i}$. 

\begin{lemma} \label{lem:lift_rpdt_dist} 
Let $\nu$ be a distribution on $(\B^l)^{\free(\rho)}$. 
Let $W(v)$ be the event that node $v$ of $T$ is visited by Algorithm \ref{alg:sim_lift} when run on $\rho,T$ and $z$ drawn from the distribution $\nu$. Let $V(v)$ be the event that for $x \sim \nu \circ \mu$, running $T$ on $x$ reaches $v$.
Then for every $v \in T$, we have $\Pr[W(v)] = \Pr[V(v)]$.
\end{lemma}
\begin{proof}
We will show that the statement holds pointwise:

Let $W_z(v)$ be the event that node $v$ of $T$ is visited by Algorithm \ref{alg:sim_lift} on the string $z \in (\B^l)^{\free(\rho)}$. Let $V_z(v)$ be the event that for a random $x \sim \mu_z$, running $T$ on $x$ reaches $v$.
Then for every $v \in T$, we have $\Pr[W_z(v)] = \Pr[V_z(v)]$.
The desired statement then follows by averaging.

Fix $z \in (\B^l)^{\free(\rho)}$.
Like in Lemma \ref{lem:rpdt_dist}, because of the invariants in Lemma \ref{lem:lift_inv}, it is sufficient to show that the simulation adds equations to $L$ with the correct probabilities. We will maintain the following invariant.
    At the beginning of an iteration, we have a linear system $L$ which fixes some subset $\free(\rho) \setminus F$ of blocks as affine functions of the blocks $F$ such that conditioned on $x \sim \mu_z$ satisfying $L$, the distribution on $x$ projected to the free blocks $F$ is $\mu_{z}[F]$. 

    The case that the parity we are trying to simulate is already determined by $L$ is clear.

Suppose the parity query is not determined by $L$. Let $P$ denote the equivalent parity query under $L$ which only depends on the blocks in $F$.
    Fix $i \in F, j \in [m]$ occurring in $P$. Having queried $z_i$, we now sample $u \sim (1-\delta)A_{z_i}^j + \delta B_{z_i}^j \times U_1$ according to how the mixture is defined. In the case that we are sampling from $A_{z_i}^j$, it is clear that the resulting conditional distribution on the free blocks stays $\mu_z[F]$ after having removed $i$ from $F$ since we have independence across blocks.
    
    When sampling from $B_{z_i}^j \times U_1$, we do not directly assign a uniform random bit to $x_{i, j}$ but instead assign the random bit to $P$. This parity is indeed uniformly distributed in this case since the bit $x_{i, j}$ is uniform random and independent of all other variables appearing in the parity. Moreover, the independence of the blocks also implies that conditioned on the value of this parity and the other bits in $x_i$, the distribution on the free blocks remains $\mu_z[F]$.
\end{proof}

\begin{lemma} \label{lem:lift_fail_prob}
Suppose $\rho$ fixes at most $p$ blocks. Let $T$ be a PDT on  $(\B^m)^{\free(\rho)}$ whose depth $d$ is at most $\delta q/2$. Then the probability that Algorithm \ref{alg:sim_multicoll} when run on $T$ and $z \sim \nu_\rho$ does not return FAIL is at least $\eps - \exp(-\delta^2q/2)$. 
\end{lemma}
\begin{proof}
Algorithm \ref{alg:sim_lift} when run on $T$ and $\rho$ naturally gives a randomized decision tree on $(\B^l)^{\free(\rho)}$. This tree may have depth more than $q$, so we truncate it to obtain a randomized decision tree $T'$ of depth $q$.

Consider the natural coupling between the random leaf reached when $T'$ is run on $\mu$ and the full execution of the simulation on $\mu$. Observe that if the simulation was not cut short because of $q$ queries being made, the partial assignment at the leaf of $T'$ reached is the same as $\sigma$ in Algorithm \ref{alg:sim_lift} when it terminates. So we can use the union bound to get that $\Pr[\rho \cup \sigma_{T', \nu_\rho} \in \cP] \leq \Pr[\rho \; \cup \sigma \in \cP] + \Pr[iter > q]$, where we used that the number of iterations of the simulation is at least the number of queries made.

By the definition of $(p, q)$-DT-hardness and averaging, $\Pr[\rho \cup \sigma_{T', \nu_\rho} \in \cP] \geq \eps$ where the probability is over the distribution $\nu_{\rho}$ and the randomness of $T'$. On the other hand, to estimate $\Pr[iter > q]$, note that in each iteration, conditioned on the history, with probability at least $\delta$, we take a step down the PDT. So when the number of iterations exceeds $q$, the number of such good iterations where we move down the PDT out of the first $q$ iterations is less than $d$. Let $X$ denote the number of such good iterations. The Chernoff bound now implies that
\[
\Pr[X < d] = \Pr[X - \delta q < d - \delta q] \leq \exp\left(-2\left(\delta - \frac{d}{q}\right)^2 q\right) \leq \exp(-\delta^2q/2).
\]
Putting the two bounds together gives the desired lower bound $\eps - \exp(-\delta^2 q/2)$.
\end{proof}

\begin{lemma} \label{lem:lift_not_sink}
    Let $v, \sigma, L, F$ be returned by a successful run of Algorithm \ref{alg:sim_lift} when run on $\rho$ and $T$. For every $o \in \cO$, there exists $x \in (\B^m)^{\free(\rho)}$ such that $x$ reaches the leaf $v$ and $(x, o) \notin \cR|_\rho \circ g$.
\end{lemma}
\begin{proof}
    By Lemma \ref{lem:lift_inv}, we know that $L$ implies all the parity constraints on the path from the root to $v$. So it suffices to find $x$ satisfying $L$ for which $(x, o) \notin \cR$.

    Since this was a successful run, $\rho \;\cup \sigma \in \cP$. By definition of $\cP$, $\rho \; \cup \sigma$ is not a certificate for $\cR$. So there exists some $y \in (\B^l)^n$ which is consistent with $\rho \; \cup \sigma$ such that $(y, o) \notin \cR$. For each $i \in F$, we now assign $x_i$ so that $g(x_i) = y_i$. Then extend to all blocks in $\free(\rho) \setminus F$ according to the constraints in $L$. By Lemma \ref{lem:lift_inv}, for each $i \in \free(\rho) \setminus F$, we have $g(x_i) = y_i$ since $y$ agrees with $\sigma$ wherever it is defined. So $g^{\free(\rho)}(x)$ agrees with $y$ on all the blocks in $\free(\rho)$. Hence $(x, o) \notin \cR|_\rho \circ g$.
\end{proof}

\subsection{Affine DAGs for $\cR \circ g$}

We will need a simple variant of the `lifted distributions fool rank' lemma of Bhattacharya, Chattopadhyay and Dvo\v{r}\'{a}k \cite{bhat2024} for safe systems that avoids the loss depending on the gadget size.

\begin{lemma} \label{lem:lift_safe_prob}
    Let $\nu$ be a distribution on $(\B^l)^U$. Let $\{\mu_y\}_{y \in \B^l}$ be a collection of distributions witnessing that $g$ is $\delta$-balanced. Let $\Phi$ be a linear system on the lifted variables $x_{i, j}\; (i \in U, j \in [m])$. Let $r = \sdim(\Phi)$. Then
    \[
    \Pr_{x \sim \nu \circ \mu}[x \text{ satisfies } \Phi] \leq (1-\delta/2)^r.
    \]
\end{lemma}
\begin{proof}
    By definition, there exists a safe linear system $\hat{\Phi}$ implied by $\Phi$ and the rank of $\hat{\Phi}$ is $r$. We will show that $\Pr_{x \sim \nu \circ \mu}[x \text{ satisfies } \hat{\Phi}] \leq (1-\delta/2)^r$ which will imply the claim.

    Condition on $z = g^U(x)$. After this conditioning, $\nu \circ \mu$ is the product distribution $\prod_{i \in U} \mu_{z_i}$. Since $\hat{\Phi}$ is safe, there exists a way to rewrite $\hat{\Phi}$ as fixing $r$ variables $x_{i_t, j_t}\; (t \in [r])$ as affine functions of all other variables, where $i_t \; (t \in [r])$ are all distinct (Theorem \ref{lem:safe}). Condition on all blocks not in $T \coloneqq \{i_t \mid t \in [r]\}$.

    For each $i\coloneqq i_t$ and $j\coloneqq j_t$, we will now consider sampling $x_i$ according to $\mu_{z_i}$ in the following way (each of these is done independently since our distribution is a product distribution). Consider a Bernoulli random variable with probability $\delta$ to pick a component of the mixture $\mu_{z_i} = (1-p) A_{z_{i}}^{j} + p B_{z_{i}}^{j} \times U_1$. If it is $1$, then we assign $x_{i, [m] \setminus \{j\}} \sim B_{z_i}^{j}$. Otherwise $x_i \sim A_{z_i}^{j}$. 

    Having done this, the only unfixed variables are some variables from $x_{i_t, j_t},\; t \in [r]$. If $s$ of the Bernoulli variables above are $1$, then $\hat{\Phi}$ simplifies to $s$ equations each involving a different variable among  $x_{i_t, j_t}\; t \in [r]$ along with $r-s$ other equations involving only constants which we ignore. Now each of these $s$ surviving variables is a uniform random bit independent of the others. So conditioned on the events so far, the probability that $\hat{\Phi}$ is satisfied is at most $1/2^s$.

    Hence the overall probability that $\hat{\Phi}$ is satisfied is at most
    \[
    \sum_{s = 0}^r \binom{r}{s}(1-\delta)^{r-s} \delta^s \cdot \frac{1}{2^s} = \left(1- \delta + \frac{\delta}{2}\right)^r = \left(1-\frac{\delta}{2}\right)^r. \qedhere
    \]
\end{proof}

The next lemma gives a suitable affine restriction to restart the random walk. 

\begin{lemma} \label{lem:lift_aff_restrict}
    Let $\rho \in \cP$. Let $v, \sigma, L, F$ be returned by a successful run of Algorithm \ref{alg:sim_lift} on $\rho$ and $T$. Let $\Phi$ be a linear system implied by $L$. Let $r = \sdim(\Phi)$. If $|\fix(\rho)| + r \leq p$ and $r \leq q$, then there exists an affine restriction $\tau$ on the lifted variables fixing blocks $\free(\rho)\setminus F'$ and satisfying the following conditions:
    \begin{enumerate}
        \item The number of blocks fixed by $\tau$ is at most $r$.
        \item $\Phi$ is implied by $\tau$. 
        
        \item There exists a partial assignment $\rho'$ on $(\B^l)^{\free(\rho)}$ fixing blocks in $\free(\rho)\setminus F'$ such that $\rho \cup \rho' \in \cP$ and the following holds. For every assignment to the lifted variables in the blocks $F'$, if we set the blocks $[m] \setminus F'$ according to $\tau$, 
        then applying the gadget $g$ to all the blocks in $\fix(\rho') = \free(\rho) \setminus F'$ gives the partial assignment $\rho'$.
    \end{enumerate}
\end{lemma}
\begin{proof}
Let $\sigma'$ denote the partial assignment obtained by from $\sigma$ by only keeping blocks in $\cl(\Phi) \setminus F$ and forgetting all other blocks. So $|\fix(\rho \cup \sigma')| = |\fix(\rho)| + |\cl(\Phi) \setminus F| \leq |\fix(\rho)| + r \leq p$ where we used Lemma \ref{lem:sdim_bound}. 
Since $\cP$ is downward closed and $\rho \cup \sigma \in \cP$ as this was a successful run, we have $\rho \cup \sigma' \in \cP$.

Let $I_1 = \cl(\Phi) \cap F$. Let $d = \dim(L(\Phi)[\setminus \cl(\Phi)])$. Since $L(\Phi)[\setminus \cl(\Phi)]$ is safe, we can pick linearly independent columns $(i_1, j_1), (i_2, j_2), \dots, (i_d, j_d)$ for the corresponding matrix, where $i_1, i_2, \dots, i_d$ are all distinct (Theorem \ref{lem:safe}). Let $I_2 = \{i_1, i_2, \dots, i_d\}$. 
Consider a non-adaptive decision tree $T'$ on the unlifted variables $(\B^l)^{\free(\rho \cup \sigma')}$ which queries blocks in $I_1 \cup I_2$. This decision tree has depth $|I_1| + |I_2| \leq |\cl(\Phi)| + d \leq r \leq q$ where the second inequality uses Lemma \ref{lem:sdim_bound}. So consider running $T'$ on the distribution $\nu_{\rho \cup \sigma'}$ given by $\rho \cup \sigma' \in \cP$ with $|\fix(\rho \cup \sigma')| \leq p$. Since $\Pr[\rho \cup \sigma' \cup \sigma_{T', \nu_{\rho \cup \sigma'}} \in \cP ] \geq \eps > 0$, there exists some partial assignment $\hat{\sigma}$ fixing blocks in $I_1 \cup I_2$ such that $\rho \cup \sigma' \cup \hat{\sigma} \in \cP$.

Now consider a total assignment $x$ for the lifted variables in blocks $\free(\rho)$ obtained in the following way. For each block $i$ in $I_1$, set $x_i$ such that $g(x_i)= \hat{\sigma}(z_i)$. All the other blocks in $F \setminus I_1$ are set arbitrarily. We extend to the remaining blocks $\free(\rho) \setminus F$ by assigning according to $L$. Now by Lemma \ref{lem:lift_inv}, $g(x_i) = \sigma(z_i)$ for all $i \in \free(\rho) \setminus F$. In particular, this holds for all $i \in \cl(\Phi) \setminus F$ which implies  $g(x_i)= \sigma'(z_i)$. 
Let $\tau_1$ denote the bit-fixing restriction fixing the variables in blocks $\cl(\Phi)$ according to $x$.

Define another bit-fixing restriction $\tau_2$ fixing all variables in $\{x_{i, j} \mid i \in I_2, j \in [m]\}$ except $(i_1, j_1), (i_2, j_2), \dots, (i_d, j_d)$ in the following way. For each such $i \in I_2$ such that $(i, j)$ is among the $d$ linearly independent columns, $\tau_2$ fixes all $x_{i, j'}, \; j' \neq j$ to force the gadget $g$ to output $\hat{\sigma}(z_i)$ when applied to block $x_i$. Such an assignment exists since $g$ is stifling.

Now consider $\Phi|_{\tau_1 \cup \tau_2}$. By the choice of $(i_1, j_1), (i_2, j_2), \dots, (i_d, j_d)$ above, $\Phi|_{\tau_1 \cup \tau_2}$ still has rank $d$ and can be viewed as fixing these $d$ variables in terms of all the variables in blocks $\free(\rho) \setminus (I_1 \cup I_2)$. Let $\tau_3$ denote this affine restriction. Finally $\tau \coloneqq \tau_1 \cup \tau_2 \cup \tau_3$ and $F' = F \setminus (\cl(\Phi) \cup I_2)$.

We now verify the desired properties. $\tau_1$ fixes $\cl(\Phi)$ and $\tau_2, \tau_3$ fix $I_2$. So the number of blocks fixed by $\tau$ is $|\cl(\Phi)| + \dim(L(\Phi)[\setminus \cl(\Phi)]) \leq r$ by Lemma \ref{lem:sdim_bound}. By construction, $\tau_1$ implies the equations in $\langle \Phi\rangle$ supported on the closure of $\Phi$. The other equations are implied by how we defined $\tau_3$ from $\Phi|_{\tau_1 \cup \tau_2}$.

Define $\rho' = \sigma' \cup \hat{\sigma}$. Now consider any assignment to the lifted variables in $F'$. Extend to a full assignment $x$ for $(\B^m)^{\free(\rho)}$ according to $\tau$. Then we saw above that $\tau_1$ ensures $g(x_i) = \sigma'(z_i)$ for all $i \in \cl(\Phi) \setminus F$. For $i \in \cl(\Phi) \cap F$, $\tau_1$ ensures $g(x_i) = \hat{\sigma}(z_i)$. For $i \in I_2$, $\tau_2$ ensured that $g(x_i) = \hat{\sigma}(z_i)$ irrespective of the bit $x_{i, j}$ not fixed by $\tau_2$.
\end{proof}

We now prove the main lemma which does the random walk.

\begin{lemma} \label{lem:lift_dag_restrict}
    Let $\rho \in (\B^l \cup \{*\})^n$. Suppose there is an affine DAG $C$ solving $\cR|_\rho \circ g$ whose depth is at most $D$ and size is at most $S$.    
    Let $d = \floor{\delta q/2}$. Set $\hat{\eps} = \eps - \exp(-\delta^2 q/2)$ and suppose $\hat{\eps} > 0$. If $|\fix(\rho)| \leq p$, then $D \geq d$.
    
    Moreover if $|\fix(\rho)| + \frac{2}{\delta}\ln\left(\frac{S}{\hat{\eps}}\right) \leq p$, then there exists $\rho'$ on $(\B^l)^{\free(\rho)}$ such that $|\fix(\rho')| \leq O(\log(S/\hat{\eps})/\delta)$ and there exists an an affine DAG $C'$ solving $\cR|_{\rho \cup \rho'} \circ g$ whose depth is at most $D - d$ and size is at most $S$.
\end{lemma}
\begin{proof}
We start by assuming only $|\fix(\rho)| \leq p$.
Consider the depth $d$ PDT $T$ given by the DAG $C$ starting at its root and keeping only nodes that are at depth at most $d$, repeating as necessary.
Run Algorithm \ref{alg:sim_lift} on $\rho$ and $T$ where $z$ is drawn from $\nu_{\rho}$. By Lemma \ref{lem:lift_fail_prob}, it succeeds with probability $\hat{\eps} > 0$. Consider any node $w$ in the DAG $C$ which corresponds to a leaf $v$ reached when the simulation terminated successfully. By Lemma \ref{lem:lift_not_sink}, $w$ cannot be a sink. So $v$ must be at depth $d$ in $T$ implying $D \geq d$.

Now let us assume $|\fix(\rho)| + \frac{2}{\delta}\ln\left(\frac{S}{\hat{\eps}}\right) \leq p$. As noted above, the simulation succeeds with probability $\hat{\eps}$. Since there are only $S$ nodes in $C$, there exists some node $w$ in the DAG $C$ such that the corresponding leaves in $T$ are successfully reached with probability at least $\hat{\eps}/S$. By Lemma \ref{lem:lift_rpdt_dist}, the node $w$ is reached with probability at least $\hat{\eps}/S$ when we pick $x \sim \nu_{\rho} \circ \mu$ and follow the path in the DAG $C$ taken by $x$ starting at the root.

Let $\Phi$ be the linear system at node $w$. Let $r\coloneqq \sdim(\Phi)$. By Lemma \ref{lem:lift_safe_prob}, the probability that $x \sim \nu_{\rho} \circ \mu$ satisfies $\Phi$ is at most $(1-\delta/2)^r \leq \exp(-\delta r/2)$. Combining this with the lower bound above, we obtain $r \leq \frac{2}{\delta}\ln\left(\frac{S}{\hat{\eps}}\right)$.
Fix a successful run of the simulation which outputs $v, \sigma, L, F$. We now use Lemma \ref{lem:lift_aff_restrict} to obtain an affine restriction $\tau$ on the lifted variables fixing blocks in $\free(\rho)\setminus F'$ and a partial assignment $\rho'$ fixing the same set of blocks. The condition $|\fix(\rho)| + r \leq p$ holds by the bound above. For the other condition, by Lemma \ref{lem:path_implies}, $r \leq d = \delta q/2 \leq q$ since $\delta \leq 1$.

Now apply the affine restriction $\tau$ to $C$. This DAG solves $(\cR|_{\rho} \circ g)|_{\tau}$. The node $w$ in the resulting DAG is now labeled by the empty system so the DAG $C'$ obtained by considering only the nodes reachable from $C'$ also solves $(\cR|_{\rho} \circ g)|_{\tau}$.

We claim that $(\cR|_{\rho} \circ g)|_{\tau} \subseteq \cR|_{\rho \cup \rho'} \circ g$. Consider any $x' \in (\B^l)^{F'}$ and $o \in \cO$ such that $(x', o) \in (\cR|_{\rho} \circ g)|_{\tau}$. By definition, if $x$ denotes the unique extension of $x'$ according to $\tau$, $(x, o) \in \cR|_{\rho} \circ g$. This means $(g^{\free(\rho)}(x), o) \in \cR|_{\rho}$. By Lemma \ref{lem:lift_aff_restrict}, when $g$ is applied to each of the blocks of $x$ in $\fix(\rho')$, we get $\rho'$. On the blocks in $F'$, $x$ agrees with $x'$. So $(g^{F'}(x'), o) \in \cR|_{\rho \cup \rho'}$ which implies $(x', o) \in \cR|_{\rho \cup \rho'}\circ g$. Hence $(\cR|_{\rho} \circ g)|_{\tau} \subseteq \cR|_{\rho \cup \rho'} \circ g$. Since $C'$ solves $(\cR|_{\rho} \circ g)|_{\tau}$, it also solves $\cR|_{\rho \cup \rho'} \circ g$.

The bound on the size of $C'$ is obvious. The bound on depth follows from using that there is some length $d$ path in $C$ from the root to $w$ by Lemma \ref{lem:lift_not_sink}. The bound on $|\fix(\rho')|$ comes from the guarantee on $\tau$ given by Lemma \ref{lem:lift_aff_restrict} and $r \leq O(\log(S/\hat{\eps})/\delta)$.
\end{proof}

We now make repeated use of the above lemma. The condition connecting $\eps, \delta$ and $q$ below is fairly mild and is given just to keep the final bound in terms of $\eps$.
\begin{theorem} \label{thm:lift_gen}
    Let $\cR$ have a $(p, q, \eps)$-block-DT-hard set $\cP$. Suppose $\eps \geq 2\exp(-\delta^2 q/2)$. If $C$ is an affine DAG solving $\cR \circ g$ whose depth is $D$ and size is $S$, then $D (\log S + \log(1/\eps)) = \Omega(\delta p\floor{\delta q/2})$.
\end{theorem}
\begin{proof}
    Set $\rho = *^n, d = \floor{\delta q/2}, \hat{\eps} = \eps - \exp(-\delta^2 q/2)$. By assumption, $\hat{\eps} \geq \eps/2$.
    
    Now repeatedly apply Lemma \ref{lem:lift_aff_restrict} to $\rho$ and $C$, updating them to $\rho \cup \rho'$ and $C'$ respectively given by the statement, until $|\fix(\rho)| > p$. Since in each phase, we fix at most $O(\log(S/\hat{\eps})/\delta) \leq O(\log(S/\eps)/\delta)$ more blocks in $\rho$, there are at least $\ceil{\frac{p}{O(\log(S/\eps)/\delta)}}$ phases. This implies $D \geq \frac{\delta p}{O(\log(S/\eps))} \cdot d = \Omega\left(\frac{\delta p\floor{\delta q/2}}{\log(S/\eps)}\right)$ as desired.
\end{proof}


Theorem \ref{thm:lift} is a simple corollary of Theorem \ref{thm:lift_gen}.
\begin{theorem}[Theorem \ref{thm:lift} restated] \label{thm:lift_restate}
    Let $\varphi$ be a formula over $nl$ variables, partitioned into $n$ blocks of $l$ variables each. Let $m$ be the number of clauses of $\varphi$, $w$ its width and $b$ its block-width. Suppose there exists a $(p, q)$-block-DT-hard set for $\varphi$. 
    Then the formula $\lift(\varphi)$ over $n(2l+1)$ variables contains at most $2^bm$ clauses of width $w + b$ and every depth $D$ $\res(\oplus)$ proof of $\lift(\varphi)$ must have size $\exp(\Omega(pq/D))$.
\end{theorem}
\begin{proof}
The $(p, q)$-block-DT-hard set for $\varphi$ is a $(p, q, \eps = 1/3)$-block-DT-hard set for the associated false clause search problem $\cR^\varphi$.  The gadget $g$ is $\bind_l$ for which we can take $\delta = 1/3$ as explained in Subsection \ref{subsec:lift_prel}.

Any $\res(\oplus)$ proof of $\lift(\varphi)$ of depth $D$ and size $S$ can be converted to an affine DAG for the associated false clause search problem $\cR^{\lift(\varphi)}$ with the same size and depth. It is easy to see that by changing the output labels of the affine DAG for $\cR^{\lift(\varphi)}$ suitably we can obtain an affine DAG for $\cR^\varphi \circ g$.
    
    If $\eps < 2\exp(-q/18)$, then $q = O(1)$. So the desired bound is just $D \log S = \Omega(p)$. We will show that $D \geq p$. It is sufficient to show that any ordinary decision tree for $\cR^\varphi$ must have depth at least $p$ since then we can use the PDT lifting theorem using stifling gadgets from \cite{cmss23}. We use the hard distributions from $\cP$ to give an adversary strategy. At each query, the corresponding hard distribution for the current partial assignment gives a response which ensures that the resulting partial assignment remains in $\cP$. This can be done until more than $p$ queries are made.

    Now suppose $\eps \geq 2\exp(-q/18)$. This also implies $\delta q/2 \geq 1$, so $\floor{\delta q/2} \geq \delta q/4$. Then Theorem \ref{thm:lift_gen} implies $D \log S = \Omega(pq)$.
\end{proof}

\section{Concluding remarks}
\label{sec:conc}

While improving the $\res(\oplus)$ lower bounds for the bit pigeonhole principle remains a natural goal, it would also be interesting to know if there is a non-trivial upper bound on the size of $\res(\oplus)$ proofs of the weak bit pigeonhole principle for any number of pigeons. 
Let us recall two such surprising upper bounds for the unary pigeonhole principle $\php_n^m$. Buss and Pitassi \cite{buss1997resolution} showed that when the number of pigeons is $2^{\sqrt{n \log n}}$, there is a $2^{O(\sqrt{n \log n})}$ size proof in resolution. On the other hand, for $n+1$ pigeons, any resolution proof has size $2^{\Omega(n)}$ \cite{haken85}. 

Oparin \cite{oparin2016tight} showed that there is a tree-like $\res(\oplus)$ proof of $\php_n^{n+1}$ whose size is $2^{O(n)}$, matching the known lower bound \cite{itsykson2020resolution}. This is in contrast to tree-like resolution for which $\php_n^m$ requires proofs of size $2^{\Theta(n \log n)}$ for any number of pigeons $m$ \cite{iwama1999tree, dantchev2001tree, beyersdorff2010lower}.

It may also be useful to have a better understanding of how the localization-based approach to analyzing parity decision trees \cite{cmss23, bi25} that we use here compares with directly analyzing the closure. 
At first glance, one might hope that the revealed blocks during these simulations should always contain the closure. This is not true, however, as the following example shows. Consider a parity decision tree which makes the following non-adaptive queries in order: $x_{1,1} + x_{3, 1}, x_{2, 1} + x_{3,2}, x_{1, 2}, x_{2, 2}$. When simulating this PDT, the first query is localized to $x_{1, 1}$ and all other bits in $x_{1}$ are fixed. The second query is similarly localized to $x_{2, 1}$ and all other bits in $x_2$ are fixed. Since the next two queries are already determined by the fixed bits, no other blocks are fixed. On the other hand, the closure of these four linear forms is $\{1, 2, 3\}$.

Similar to how \cite{bi25} gave a localization-based proof of the random walk analysis from \cite{egi24}, it is possible to give such a proof of the random walk theorem \cite[Theorem 4.3]{alek2024}, if we make the stronger assumptions that the starting system is empty and the number of steps is $w/2$ instead of $w$. By arguments similar to those in Section \ref{sec:lifting}, this weaker statement still suffices to prove the size-depth lower bound in \cite[Theorem 5.4]{ei25} based on the advanced Prover-Delayer game up to a constant factor. As we show in Appendix \ref{app:pd_games}, the lifting theorem in \cite{alek2024, ei25} is also implied by Theorem \ref{thm:lift_gen}.

On the other hand, it seems that a localization-based approach does not suffice to prove the bounded-width lifting theorem for tree-like $\res(\oplus)$ using stifling gadgets \cite{ei25}. The proof in \cite{ei25} is based on tracking the amortized closure. Such a lifting theorem with the smaller class of strongly stifling gadgets had been shown earlier by Chattopadhyay and Dvo\v{r}\'{a}k \cite{chattopadhyay2025super} by a localization-based approach which extended the simulation of Chattopadhyay, Mande, Sanyal and Sherif \cite{cmss23} to make it width-preserving.

\section*{Acknowledgments}

We thank Lutz Warnke for discussions on balls in bins, and Anthony Ostuni for comments on an earlier draft of this paper that helped improve the presentation. FB thanks Jackson Morris for discussions on balls in bins and other related topics.
\pagebreak

\printbibliography[heading=bibintoc]
\appendix
\section{Block-DT-hardness of $\bphp$ and $\fphp$}
\label{app:dt_hard}

For completeness, in this appendix, we give a very simple proof of the block-DT-hardness of the bit pigeonhole principle and the functional pigeonhole principle. 

\subsection{Bit pigeonhole principle}
The hard assignments and distributions for $\bphp_n^{m}$ are essentially the same as those considered in Section \ref{sec:sbphp}, but we will provide a self-contained description and proof. For $m = n+1$, we implicitly already discussed a slightly different collection of hard assignments and distributions in Section \ref{subsec:s_bphp_overview} where we sketched an $\Omega(n^2/\log S)$ bound on the depth of any size $S$ subcube DAG for $\coll_n^{n+1}$. 

Here we will  consider $m = n + k$ for $k \leq n/48$.
It will be convenient to identify $\B^l$ with $[n]$ via any bijection. 

 The set $\cP$ contains all partial assignments for which all fixed pigeons are sent to distinct holes. More formally, a partial assignment $\rho \in ([n] \cup \{*\})^m$ belongs to $\cP$ if and only if for all distinct $i_1, i_2 \in [m]$, $\rho(z_{i_1}) = *$ or $\rho(z_{i_2}) = *$ or $\rho(z_{i_1}) \neq \rho(z_{i_2})$. It is clear that $*^m \in \cP$, no partial assignment in $\cP$ falsifies a clause of $\bphp_n^{n+k}$ and $\cP$ is downward closed.

We will show that $\cP$ is a $(p, q, \eps)$-block-DT-hard set for $p = 15n/32$, $q = d$ and $\eps = \exp(-8d^2k/n^2)$ for any $d \leq n/4$.
 For each $\rho \in \cP$ such that $|\fix(\rho)| \leq p$, we define the distribution $\nu_{\rho}$ as follows. Let $A \subseteq [n]$ contain all the holes which have not already been assigned by $\rho$ to some pigeon. We have $|A| \geq 17n/32 \geq k$. Let $\cS$ be any multiset of size $m' \coloneqq m - |\fix(\rho)| = |A| + k$ containing only holes from $A$ such that each element of $A$ appears at least once and at most twice in $\cS$. So exactly $k$ holes appear twice in $\cS$. Then $\nu_{\rho}$ is the uniform distribution on permutations of $\cS$. 

\begin{lemma}
    For any depth $d$ decision tree $T$, $\Pr[\rho \cup \sigma_{T, \nu_\rho} \in \cP] \geq \exp(-8d^2k/n^2)$.
\end{lemma}

\begin{proof}
Since $\cP$ is permutation-invariant and so is the distribution $\nu_{\rho}$, we may assume that $T$ queries the first $d$ blocks not fixed by $\rho$. Our goal now is simply to estimate the probability that in a random permutation of $\cS$, no hole appears twice among the first $d$ entries of the permutation. 

For this, it is convenient to think of sampling a uniform random permutation of $\cS$ in the following way. Let $h_1, h_2, \dots, h_k \in A$ be the $k$ holes which appear twice in $\cS$. We iteratively do the following. For each $i \in [k]$, we pick two positions uniformly at random out of all available positions to place $h_i$. Finally we assign a uniform random permutation of the remaining holes $A \setminus \{h_1, h_2, \dots, h_k\}$ to the available positions. 

Now the probability that there is a collision among the first $d$ entries corresponds to picking two available positions among the first $d$ entries for some $h_j, j \in [k]$. Conditioned on events in the previous iterations, the probability that for hole $h_j$, we pick two positions among the first $d$ entries is at most $\binom{d}{2}/\binom{m'-2(j-1)}{2} \leq \frac{d^2}{(m' - 2k)^2} \leq \frac{d^2}{(n-k-p)^2} \leq \frac{4d^2}{n^2}$. So the probability that we do not have a collision is at least $(1-4d^2/n^2)^k \geq \exp(-8d^2k/n^2)$ where the inequality uses $4d^2/n^2 \leq 1/4$ so that we can invoke $1-x \geq \exp(-2x)$ which holds for all $x \leq 0.75$.
\end{proof}

To use Theorem \ref{thm:lift_gen} with the gadget being block-indexing and $\cP$ above, we need $\exp(-8d^2k/n^2) \geq 2\exp(-d/18)$. This holds if $d \geq 72$ since $8d^2 k/n^2 \leq d/24$. Applying Theorem \ref{thm:lift_gen} for suitable choices of $d$ now gives the same bound as in Theorem \ref{thm:s_bphp} for $\lift(\bphp_n^{n+k})$.

\subsection{Functional pigeonhole principle}
$\fphp_n^m$ is defined on $mn$ variables, $x_{i, j}\; (i \in [m], j \in [n])$ where $x_{i, j}$ being $1$ indicates that pigeon $i$ goes to hole $j$.
$\fphp_n^m$ contains the following clauses:
\begin{itemize}
    \item Pigeon clauses: For each $i \in [m]$, $\vee_{j \in [n]} x_{i, j}$, which encodes that each pigeon goes to at least one hole.
    \item Functionality clauses: For each $i \in [m]$ and distinct $j_1, j_2 \in [n]$, $\neg x_{i, j_1} \vee \neg x_{i, j_2}$, which together encode that each pigeon goes to at most one hole.
    \item Hole clauses: For each $j \in [n]$ and distinct $i_1, i_2 \in [m]$, $\neg x_{i_1, j} \vee \neg x_{i_2, j}$, which encode that each hole receives at most one pigeon.
\end{itemize}

By considering the unary version of $\cP$ defined for $\bphp_n^{n+k} \;(k \leq n/48)$ in the previous subsection, we get a block-DT-hard set for $\fphp_n^{n+k}$ with the same parameters, and therefore the same bound as Theorem \ref{thm:s_bphp} holds for $\lift(\fphp_n^{n+k})$.

We now show that $\varphi \coloneqq \fphp_n^{n+1}$ has a $\res(\oplus)$ proof of depth $O(n \log n)$. We will do this by giving a parity decision tree for the corresponding false clause search problem $\cR^\varphi$. This suffices since any parity decision tree for $\cR^{\varphi}$ can be converted to a tree-like $\res(\oplus)$ proof of $\varphi$ while preserving the size and depth \cite{itsykson2020resolution}. We will then show how to use this PDT for $\cR^\varphi$ to give a PDT for $\cR^{\lift(\varphi)}$ whose depth is only larger by $n+1$, thereby showing that there is a $\res(\oplus)$ proof of $\lift(\fphp_n^{n+1})$ of depth $O(n \log n)$.
\begin{lemma}
    There is a PDT solving $\cR^{\varphi}$ whose depth is $O(n \log n)$.
\end{lemma}
\begin{proof}
The PDT works as follows.
    For each $i \in [n+1]$, query $\sum_{j \in [n]} x_{i, j}$. If any of these is $0$, say for pigeon $i$, then a functionality clause or pigeon clause for pigeon $i$ must be violated. So we query all $x_{i, j}\; (j \in [n])$ to find such a clause.

    If none of the above $n+1$ parities is $0$, then for each pigeon $i \in [n+1]$, we can do a binary search to find some $j_i \in [n]$ such that $x_{i, j_i} = 1$. Now that we know for each pigeon some hole it goes to, we have found a falsified hole clause.

    Overall we make at most $n+1 + \max(n, (n+1)\ceil{\log n}) = O(n \log n)$ queries.
\end{proof}

\begin{lemma}
    For any CNF formula $\psi$ containing $m$ blocks, if there exists a depth $d$ PDT solving $\cR^{\psi}$, then there exists a depth $d+m$ PDT solving $\cR^{\lift(\psi)}$.
\end{lemma}
\begin{proof}
    Query all the selector bits $s_i\; (i \in [m])$. Now use the depth $d$ PDT for $\cR^{\psi}$, on the relevant input blocks $\hat{x} = (x_1^{s_1}, x_2^{s_2}, \dots, x_m^{s_m})$. This gives some clause $C$ of $\psi$ which is falsified by $\hat{x}$. Since all the selector bits are known, this also gives a clause of $\lift(C)$ which is falsified by the original input.
\end{proof}
Combining the two lemmas gives a $\res(\oplus)$ proof of $\lift(\fphp_n^{n+1})$ whose depth is $O(n \log n)$.

\begin{remark}
We briefly explain what fails if we try to give a lower bound for $\fphp_n^{n+1}$ in bounded-depth $\res(\oplus)$ by following the proof of the lower bound for $\bphp_n^{n+1}$ in Section \ref{sec:sbphp} but with the hard distributions now being the unary versions of those for $\bphp_n^{n+1}$. For such a distribution, each bit is extremely likely to be $0$ which makes it difficult to simulate a parity query efficiently. We also do not get a useful upper bound on the probability of a safe system being satisfied under this distribution.

This is to be expected, because if we only consider assignments corresponding to near-permutations where each pigeon goes to only one hole and all holes have at least one pigeon, then there is a very efficient deterministic PDT to solve the false clause search problem for $\fphp_n^{n+1}$. By doing a binary search over the holes, we can first find the unique hole $j$ with two pigeons. Then with $\ceil{\log (n+1)}$ additional parity queries, we can find a subset $I \subseteq [n+1]$ such that some pigeon in $I$ goes to hole $j$ and some pigeon in $[n+1]\setminus I$ goes to hole $j$. Finally we do a binary search in both $I$ and $[n+1]\setminus I$ separately to find the two pigeons. Totally we make only $O(\log n)$ queries.
\end{remark}
\section{Lifting from advanced Prover-Delayer games}
\label{app:pd_games}

In this appendix, we show how Theorem \ref{thm:lift_gen} lets us recover and slightly strengthen the lifting theorem from \cite{alek2024, ei25} which is based on advanced Prover-Delayer games introduced by Alekseev and Itsykson \cite{alek2024}.

Let $\varphi$ be a CNF formula. Let $\cP$ be a collection of partial assignments to the variables of $\varphi$ satisfying all the conditions of being DT-hard except for the condition about the hard distributions (note that the parameters $p, q$ for DT hardness do not appear in these conditions). The hard distribution condition is replaced by the following game condition (which is a variant of the Prover-Delayer game \cite{pudlak_impagliazzo} used for proving size lower bounds for decision trees).

In the $(\varphi, \cP)$-game starting from an assignment $\rho_0 \in \cP$, there are two players, Prover and Delayer. They maintain a partial assignment $\rho$ which is initially $\rho_0$. In each round, Prover picks a variable $x$ not already fixed by $\rho$ and then Delayer does one of the following:
\begin{itemize}
    \item Delayer responds with $*$, in which case Prover can choose $a \in \B$ to assign to $x$.
    \item Delayer chooses $a \in \B$ to assign to $x$. In this case, Delayer pays one black coin.
\end{itemize}
Delayer earns a white coin in each round. So the number of white coins earned at the end is just the total number of rounds. The game ends when $\rho \notin \cP$.

We are interested in Delayer strategies which make the game last for many rounds but without Delayer paying many black coins. Concretely, the final condition we need from $\cP$ is the following for some parameters $p$, $t$ and $c$. For every $\rho_0 \in \cP$ fixing at most $p$ variables, there is a Delayer strategy for the $(\varphi, \cP)$-game starting at $\rho_0$ which guarantees that Delayer earns more than $t$ white coins while paying at most $c$ black coins. Note that the condition is more than $t$ to ensure that the game has not ended even after $t$ rounds.

Let us briefly mention the differences between the definition above and that considered in \cite{alek2024} and following works. A minor point is that earlier works consider a single parameter $t$ controlling both the maximum number of variables fixed by the starting position and the number of rounds (white coins earned), while here we consider $p$ and $t$ separately. A more significant requirement in \cite{alek2024} is that the Delayer strategy must be linearly described. We omit the formal definition of linearly described since we do not need it, but just mention that the strategy is allowed to depend on the current partial assignment in only a limited way. In contrast, we will allow Delayer's strategy to depend in any way on the choices so far, just like strategies in the standard Prover-Delayer game.

We now verify that a set of partial assignments $\cP$ for which we have the above properties is also a $(p, t, 1/2^c)$-DT-hard set for $\cR^\varphi$. Then using Theorem \ref{thm:lift_gen} will immediately recover the lifting theorem of \cite{alek2024, ei25}\footnote{Strictly speaking, to directly apply Theorem \ref{thm:lift_gen}, we would need that $t \geq D_g c$ where $D_g$ is a quantity depending on the gadget $g$. This is not a real concern, since to get a good size-depth lower bound, one has $t = \omega(c)$. When this is not the case, the size-depth lower bound is essentially just a depth lower bound for the lifted problem for which one can argue in a simpler way like in the proof of Theorem \ref{thm:lift_restate}.}.

For each $\rho_0 \in \cP$ fixing at most $p$ variables, the hard distribution $\nu_{\rho_0}$ is simply the uniform distribution.
Consider any depth $t$ decision tree which only queries variables not fixed by $\rho_0 \in \cP$. The argument for showing that with probability at least $2^{-c}$ we reach a leaf where the partial assignment reached still lies in $\cP$, is essentially the same as that for showing that a Delayer strategy in the standard Prover-Delayer game gives a lower bound on decision tree size.

We assume that each leaf in the tree is at depth exactly $t$ so that each leaf is reached with the same probability $2^{-t}$ under the uniform distribution. Consider the paths in the tree corresponding to responses in the Prover-Delayer game which are consistent with the Delayer strategy. Since Delayer's strategy guarantees earning $t$ white coins and paying at most $c$ coins, along any such path of $t$ queries at least $t-c$ were free choices. So there are at least $2^{t - c}$ leaves which are consistent with Delayer's strategy. This gives the desired $2^{t-c}/2^t = 2^{-c}$ lower bound. We remark that the random walk theorem in \cite{alek2024} for the lift of $\varphi$ essentially proceeds by lifting this argument for the uniform distribution on the unlifted variables in a white-box way to the uniform distribution for the lifted variables by using a $2$-stifling gadget.

Now Theorem \ref{thm:lift_gen} implies the lifting theorem from \cite{alek2024, ei25}, where we are using that any stifling gadget on $\ell$ bits is $(1/\ell)$-balanced (Subsection \ref{subsec:lift_prel}).

\begin{theorem}\label{thm:pd_lift}
Let $\varphi$ be a CNF formula. Let $\cP$ be a collection of partial assignments which are hard for the above Prover-Delayer game with parameters $p$, $t$ and $c$. Let $g: \B^\ell \rightarrow \B$ be any stifling gadget.
Suppose there exists a $\res(\oplus)$ proof of $\varphi \circ g$ whose depth is $D$ and size is $S$. Then $D(\log S + c) = \Omega\left(\frac{pt}{\ell^2}\right)$.
\end{theorem}

The dependence on the gadget size in the above lower bound is slightly worse than the bound in \cite{ei25}. Specifically the bound in \cite{ei25} is essentially $D(\log S + c\ell) = \Omega (pt)$. Since the gadget is typically taken to have constant size, this is not a real issue. On the other hand, \cite{alek2024, ei25} use $2$-stifling gadgets which form a strict subset of stifling gadgets. As mentioned earlier, another difference is that \cite{alek2024, ei25} require the Delayer strategy to be linearly described while we do not have such a restriction in the above statement.

Since we allow for more general strategies, one may hope that Theorem \ref{thm:pd_lift} may be useful for getting more examples of hard formulas for bounded-depth $\res(\oplus)$ with better parameters than those achieved in \cite{alek2024, ei25, ik25}. However, this is not possible and these earlier works achieve essentially the best parameters one could hope for by a direct application of such a lifting theorem. \cite{bc25} pointed this out for the Tseitin formulas used in \cite{alek2024, ei25} but the idea there also applies to any formula of small width.

\begin{observation}
    Let $\varphi$ be a CNF formula of width $k$. For any Delayer strategy for the $(\varphi, \cP)$-game  which guarantees earning $t$ white coins while paying at most $c$ black coins, we must have $c \geq \floor{t/k}$.
\end{observation}
\begin{proof}
Consider the following Prover strategy (which is essentially a well-known simple backtracking algorithm for $k$-SAT). Let $C$ be a clause appearing in $\varphi$ which has not already been satisfied by the current partial assignment $\rho$. Prover asks for the free variables appearing in $C$ one by one. When given the choice to fix a variable, Prover always sets it to ensure that the corresponding literal in $C$ is set to false. Since Delayer's strategy ensures that the game does not end after these moves (unless more than $t-k$ moves had already been made previously), to avoid $C$ being falsified, Delayer must pay a black coin at some point during these at most $k$ moves since $C$ has width at most $k$. This shows that Delayer must pay at least $\floor{t/k}$ black coins totally if the game lasts for $t$ rounds.
\end{proof}

Treating the gadget size as constant from now on, the above observation implies that for a $k$-CNF formula $\varphi$ on $n$ variables, Theorem \ref{thm:pd_lift} cannot give size lower bounds for proofs of $\varphi \circ g$ of depth beyond $O(nk)$ since $p \leq n$. In particular, this means that lower bounds for proofs of superlinear depth cannot be obtained from Theorem \ref{thm:pd_lift} for constant-width formulas and if we want the size of the lifted formula to be polynomial, then we can only go up to depth $O(n \log n)$.

\end{document}